\documentclass[11pt]{article}

\usepackage{times}
\usepackage{latexsym}
\usepackage[margin=1in]{geometry}
\usepackage{graphicx}
\usepackage{amsthm}
\usepackage{amsmath}
\usepackage{amssymb}
\usepackage{amsfonts}
\usepackage{algorithm}
\usepackage{algpseudocode}
\usepackage{fullpage}
\usepackage{mathptmx}
\usepackage{varwidth}
\usepackage{tabu}
\usepackage{xspace}
\usepackage{xcolor}

\newcommand{\eat}[1]{}

\newcommand{\mr}{\text{\sc mr}\xspace}
\newcommand{\sol}{\text{\sc sol}\xspace}
\newcommand{\alg}{\text{\sc alg}\xspace}
\newcommand{\apx}{\text{\sc apx}\xspace}
\newcommand{\opt}{\text{\sc opt}\xspace}
\newcommand{\mrs}{\text{\sc mrs}\xspace}
\newcommand{\med}{\text{\sc mid}\xspace}
\newcommand{\fracc}{\text{\sc frac}\xspace}

\newcommand{\cost}{\text{\sc cost}\xspace}
\newcommand{\regalgsol}{\text{\sc reg}\xspace}
\newcommand{\norm}[1]{\left|\left|#1\right|\right|}

\newcommand{\I}{{\bf I}}
\newcommand{\poly}{{\bf P}}
\newcommand{\np}{{\bf NP}}

\newcommand{\lp}{\text{\sc lp}\xspace}

\newtheorem{theorem}{Theorem}
\newtheorem{lemma}[theorem]{Lemma}

\newtheorem{corollary}[theorem]{Corollary}
\newtheorem{fact}[theorem]{Fact}

\DeclareMathOperator*{\argmax}{argmax}
\newcommand{\bd}{\I}

\title{Universal Algorithms for Clustering Problems}
\author{
  Arun Ganesh
  \footnote{Department of Computer Science, UC Berkeley. Email: \texttt{arunganesh@berkeley.edu}}
  \and
  Bruce M. Maggs
  \footnote{Department of Computer Science, Duke University and Emerald Innovations. Email: \texttt{bmm@cs.duke.edu}}
  \and 
  Debmalya Panigrahi
  \footnote{Department of Computer Science, Duke University. Email: \texttt{debmalya@cs.duke.edu}}\\ 
  }
\date{}

\begin{document}

\maketitle


\thispagestyle{empty}

\begin{abstract}
This paper presents {\em universal} algorithms for clustering problems, including the widely studied $k$-median, $k$-means, and $k$-center objectives. The input is a metric space containing all {\em potential} client locations. The algorithm must select $k$ cluster centers such that they are a good solution for {\em any} subset of clients that actually realize. Specifically, we aim for low {\em regret}, defined as the maximum over all subsets of the difference between the cost of the algorithm's solution and that of an optimal solution. A universal  algorithm's solution $\sol$ for a clustering problem is said to be an $(\alpha, \beta)$-approximation if for all subsets of clients $C'$, it satisfies $\sol(C') \leq \alpha \cdot \opt(C') + \beta \cdot \mr$, where $\opt(C')$ is the cost of the optimal solution for clients $C'$ and $\mr$ is the minimum regret achievable by any solution. 

Our main results are universal algorithms for the standard clustering objectives of $k$-median, $k$-means, and $k$-center that achieve $(O(1), O(1))$-approximations. These results are obtained via a novel framework for universal algorithms using linear programming (LP) relaxations. These results generalize to other $\ell_p$-objectives and the setting where some subset of the clients are \textit{fixed}. We also give hardness results showing that $(\alpha, \beta)$-approximation is NP-hard if $\alpha$ or $\beta$ is at most a certain constant, even for the widely studied special case of Euclidean metric spaces. This shows that in some sense, $(O(1), O(1))$-approximation is the strongest type of guarantee obtainable for universal clustering.


\end{abstract}

\setcounter{page}{0}

\clearpage

\section{Introduction}

 In {\em universal}\footnote{In the context of clustering, universal facility location sometimes refers to facility location where facility costs scale with the number of clients assigned to them. This problem is unrelated to the notion of universal algorithms studied in this paper.} approximation (e.g.,~\cite{BertsimasG89,BhalgatCK11,BuschDRRS12,GorodezkyKSS10,Gupta:2006:OND:1109557.1109665,Hajiaghayi:2006:ILU:1109557.1109628,10.1007/11776178_18,Platzman:1989:SCP:76359.76361,SCHALEKAMP2008}), the algorithm is presented with a set of {\em potential} input points and must produce a solution.  After seeing the solution, an adversary selects some subset of the points as the actual {\em realization} of the input, and the cost of the solution is based on this realization. 
 The goal of a universal algorithm is to obtain a solution that is near-optimal for {\em every} possible input realization.   For example, suppose that a network-based-service provider can afford to deploy servers at $k$ locations around the world and hopes to minimize latency between clients and servers.  The service provider does not know in advance which clients will request service, but knows where clients are located.  A universal solution provides guarantees on the quality of the solution regardless of which clients ultimately request service.  As another example, suppose that a program committee chair wishes to invite $k$ people to serve on the committee.  The chair knows the areas of expertise of each person who is qualified to serve.  Based on past iterations of the conference, the chair also knows about many possible topics that might be addressed by submissions.  The chair could use a universal algorithm to select a committee that will cover the topics well, regardless of the topics of the papers that are submitted. The situation also arises in targeting advertising campaigns to client demographics. Suppose a campaign can spend for $k$ advertisements, each targeted to a specific client type. While the entire set of client types that are potentially interested in a new product is known, the exact subset of clients that will watch the ads, or eventually purchase the product, is unknown to the advertiser. How does the advertiser target her $k$ advertisements to address the interests of any realized subset of clients?

Motivated by these sorts of applications, this paper presents the first universal algorithms for clustering problems, including the classic $k$-median, $k$-means, and $k$-center problems. The input to these algorithms is a metric space containing all locations of {\em clients} and {\em cluster centers}. The algorithm must select $k$ cluster centers such that this is a good solution for {\em any} subset of clients that actually realize.  

It is tempting to imagine that, in general, for some large enough value of $\alpha$, one can find a solution $\sol$ such that for all realizations (i.e., subsets of clients) $C'$, $\sol(C') \leq \alpha \cdot \opt(C')$, where $\sol(C')$ denotes $\sol$'s cost in realization $C'$ and $\opt(C')$ denotes the optimal cost in realization $C'$.
But this turns out to be impossible for many problems, including the clustering problems we study, and indeed this difficulty may have limited the study of universal algorithms. For example, suppose that the input for the $k$-median problem is a uniform metric on $k+1$ points, each with a cluster center and client.  In this case,
for any solution $\sol$ with $k$ cluster centers, there is some realization $C'$ consisting of a single client that is not co-located with any of the $k$ cluster centers in $\sol$. Then, $\sol(C') > 0$ but $\opt(C') = 0$. 
Since it is not possible to provide a strict approximation guarantee for every realization, we instead seek to minimize the {\em regret}, defined as the maximum difference between the cost of the algorithm's solution and the optimal cost across all realizations.  The solution that minimizes regret is called the {\em minimum regret solution}, or $\mrs$ for short, and its regret is termed {\em minimum regret} or $\mr$. More formally, $\mr = \min_\sol \max_{C'} [\sol(C') - \opt(C')]$. We now seek a solution $\sol$ that achieves, for all input realizations $C'$, $\sol(C') - \opt(C') \leq \mr$, i.e., $\sol(C') \leq \opt(C') + \mr$. But, obtaining such a solution turns out to be \np-hard for many problems, and one has to settle for an approximation: $\sol(C') \leq \alpha \cdot \opt(C') + \beta \cdot\mr$. The algorithm is then called an $(\alpha, \beta)$-approximate universal algorithm for the problem. Note that in the aforementioned example with $k+1$ points, any solution must pay $\mr$ (the distance between any two points) in some realization where $\opt(C') = 0$ and only one client appears (in which case paying $\mr$ might sound avoidable or undesirable). This example demonstrates that stricter notions of regret and approximation than $(\alpha, \beta)$-approximation are infeasible in general, suggesting that $(\alpha, \beta)$-approximation is the least relaxed guarantee possible for universal clustering.

\subsection{Problem Definitions and Results}\label{section:defs}

We are now ready to formally define our problems and state our results. In all the clustering problems that we consider in this paper, the input is a metric space on all the potential client locations $C$ and cluster centers $F$.\footnote{We only consider finite $C$ in this paper. If $C$ is infinite (e.g. $C = \mathbb{R}^d$), then the minimum regret will usually also be infinite. If one restricts to realizations where, say, at most $m$ clients appear, it suffices to consider realizations that place $m$ clients at one of finitely many points, letting us reduce to universal $k$-center with finite $C$.}\footnote{The special case where $F = C$ has also been studied in the clustering literature, e.g., in \cite{HochbaumS85, CharikarGST99}, although the more common setting (as in our work) is to not make this assumption. Of course, all results (including ours) without this assumption also apply to the special case. If $F=C$, the constants in our bounds improve, but the results are qualitatively the same. We note that some sources refer to the $k$-center problem when $F \neq C$ as the $k$-supplier problem instead, and use $k$-center to refer exclusively to the case where $F = C$.} Let $c_{ij}$ denote the metric distance between points $i$ and $j$. The solution produced by the algorithm comprises $k$ cluster centers in $F$; let us denote this set by $\sol$. Now, suppose a subset of clients $C'\subseteq C$ realizes in the actual input. Then, the cost of each client $j\in C'$ is given as the distance from the client to its closest cluster center, i.e., $\cost(j, \sol) = \min_{i \in \sol} c_{ij}$. The clustering problems differ in how these costs are combined into the overall minimization objective. The respective objectives are given below:
\begin{itemize}
    \item {\bf $k$-median} (e.g.,~\cite{CharikarGST99,JainV01,AryaGKMMP01,LiS13,ByrkaGRS13}): {\em sum} of client costs, i.e., $\sol(C') = \sum_{j\in C'} \cost(j, \sol)$.
    \item {\bf $k$-center} (e.g.,~\cite{HochbaumS85, Gonzalez85, HochbaumS86, KhullerS00, NagarajanSS13}): {\em maximum} client cost, i.e., $\sol(C') = \max_{j\in C'} \cost(j, \sol)$.
    \item {\bf $k$-means} (e.g.,~\cite{Lloyd82, KanungoMNPSW02, KumarSS04, GuptaT08, AhmadianNSW17}): {\em $\ell_2$-norm} of client costs, i.e., $\sol(C') = \sqrt{\sum_{j\in C'} \cost(j, \sol)^2}$.
\end{itemize}
We also consider {\bf $\ell_p$-clustering}~(e.g., \cite{GuptaT08}) which generalizes all these individual clustering objectives.
In $\ell_p$-clustering, the objective is the $\ell_p$-norm of the client costs for a given value $p\geq 1$,
i.e., $$\sol(C') = \left(\sum_{j\in C'} \cost(j, \sol)^p\right)^{1/p}.$$ Note that $k$-median and $k$-means are special cases of 
$\ell_p$-clustering for $p=1$ and $p=2$ respectively. $k$-center can also be defined in the $\ell_p$-clustering framework as 
the limit of the objective for $p\rightarrow \infty$; moreover, it is well-known that $\ell_p$-norms only differ
by constants for $p > \log n$ (see Appendix~\ref{section:vector}), 
thereby allowing the $k$-center objective to be approximated within a constant by $\ell_p$-clustering
for $p = \log n$.

Our main result is to obtain $(O(1), O(1))$-approximate universal algorithms for $k$-median, $k$-center, and $k$-means.
We also generalize these results to the $\ell_p$-clustering problem.

\begin{theorem}
\label{thm:upper}
    There are $(O(1), O(1))$-approximate universal algorithms for the $k$-median, $k$-means, and $k$-center 
    problems.
    More generally, there are $(O(p), O(p^2))$-approximate universal algorithms for 
    $\ell_p$-clustering problems, for any $p\geq 1$.
\end{theorem}

\noindent{\bf Remark:}
The bound for $k$-means is by setting $p=2$ in $\ell_p$-clustering.
For $k$-median and $k$-center, we use separate algorithms to obtain improved
bounds than those provided by the $\ell_p$-clustering result. This is particularly noteworthy for $k$-center
where $\ell_p$-clustering only gives poly-logarithmic approximation.

\eat{
\noindent{\bf Remarks:} The precise approximation bounds that we obtain for these problems are:
$(27, 49)$ for $k$-median, $(3, 3)$ for $k$-center, $(108,412)$ for $k$-means,
and  $(54p, 103p^2)$ for $\ell_p$-clustering.\footnote{The exact 
constants have not been optimized, and it might be possible to improve them
with a more careful analysis, particularly for $k$-median and $k$-means.}
The bound for $k$-means is achieved by setting $p=2$ in the bound for $\ell_p$-clustering.
However, for $k$-median and $k$-center, we use specialized tools to improve on the
bound provided by the $\ell_p$-clustering result; this is particularly noteworthy for $k$-center
where $\ell_p$-clustering only gives us (poly)logarithmic approximation factors.
}

\paragraph{Universal Clustering with Fixed Clients.} We also consider a more general setting where 
some of the clients are {\em fixed}, i.e.,
are there in any realization, but the remaining clients may or may not realize as in the previous
case. (Of course, if no client is fixed, we get back the previous setting as a special case.)
This more general model is inspired by settings where a set of clients is already
present but the remaining clients are mere predictions. 
This surprisingly creates new technical challenges, 
that we overcome to get:

\begin{theorem}
\label{thm:upper-fixed}
    There are $(O(1), O(1))$-approximate universal algorithms for the $k$-median, $k$-means,
    and $k$-center problems
    with fixed clients. More generally, there are $(O(p^2), O(p^2))$-approximate universal algorithms for 
    $\ell_p$-clustering problems, for any $p\geq 1$.
\end{theorem}

\paragraph{Hardness Results.}

Next, we study the limits of approximation for universal clustering. In particular, we show that 
the universal clustering problems for all the objectives considered in this paper are \np-hard
in a rather strong sense. Specifically, we show that both $\alpha$ and $\beta$ are separately 
bounded away from $1$, {\em irrespective of the value of the other parameter}, showing the necessity of both $\alpha$ and $\beta$ in our approximation bounds. 
Similar lower bounds continue to hold for universal clustering 
in Euclidean metrics, even when PTASes are known in the offline (non-universal) 
setting~\cite{AroraRR98, KolliopoulosR99, KumarSS04, NagarajanSS13, AhmadianNSW17}.
\begin{theorem}
\label{thm:lower}
    In universal $\ell_p$-clustering for any $p \geq 1$, obtaining $\alpha < 3$  
    or $\beta < 2$ is \np-hard.
    Even for Euclidean metrics, obtaining $\alpha < 1.8$ or $\beta \leq 1$ is \np-hard.
    The lower bounds on $\alpha$ (resp., $\beta$) are independent of the 
    value of $\beta$ (resp., $\alpha$).
\end{theorem}

Interestingly, our lower bounds rely on realizations where sometimes as few as one client appears. This suggests that e.g. redefining regret to be some function of the number of clients that appear (rather than just their cost) cannot subvert these lower bounds.

\subsection{Techniques}\label{subsection:techniques}

Before discussing our techniques, we discuss why standard approximations for clustering problems are insufficient. It is known that the {\em optimal} solution for the realization that includes all clients gives a $(1, 2)$-approximation for universal $k$-median (this is a corollary of a more general result in \cite{KasperskiZ07}; we do not know if their analysis can be extended to e.g. $k$-means), giving universal algorithms for ``easy'' cases of $k$-median such as tree metrics. But, the clustering problems we consider in this paper are \np-hard in general; so, the best we can hope for in polynomial time is to obtain optimal {\em fractional} solutions, or {\em approximate} integer solutions. Unfortunately, the proof of \cite{KasperskiZ07} does not generalize to {\em any} regret guarantee for the optimal {\em fractional} solution. Furthermore, for all problems considered in this paper, even $(1+\epsilon)$-approximate (integer) solutions for the ``all clients'' instance are not guaranteed to be $(\alpha, \beta)$-approximations for any finite $\alpha, \beta$ (see the example in Appendix~\ref{sec:approx-fail}). These observations fundamentally distinguish universal approximations for \np-hard problems like the clustering problems in this paper from those in \poly, and require us to develop new techniques for universal approximations.


In this paper, 
we develop a general framework for universal approximation based on linear programming (\lp) relaxations that forms the basis of our results on $k$-median, $k$-means, and $k$-center (Theorem~\ref{thm:upper}) as well as the extension to universal clustering with fixed clients (Theorem~\ref{thm:upper-fixed}). 

The first step in our framework is to write an \lp relaxation of the regret minimization problem.  In this formulation, we introduce a new regret variable that we seek to minimize and is constrained to be at least the difference between the (fractional) solution obtained by the \lp and the optimal integer solution {\em for every realizable instance}. Abstractly, if the \lp relaxation of the optimization problem is given by $\min\{{\bf c\cdot x}: {\bf x} \in P\}$, then the new {\em regret minimization \lp}\footnote{For problems like $k$-means with non-linear objectives, the constraint ${\bf c}(I){\bf \cdot x} \leq \opt(I) + {\bf r}$ cannot be replaced with a constraint that is simultaneously linear in ${\bf x, r}$. However, for a fixed value of ${\bf r}$, the corresponding non-linear constraints still give a convex feasible region, and so the techniques we discuss in this section can still be used.} is given by $$\min\{{\bf r}: {\bf x} \in P; {\bf c}(I){\bf \cdot x} \leq \opt(I) + {\bf r},~\forall I\}.$$ 

\eat{
For our algorithmic guarantees, we first write a regret-minimization polytope for clustering problems and find a fractional solution in this polytope with universal approximation guarantees. We write the polytope of fractional solutions with regret at most $r$ by taking the standard LP relaxation for clustering problems and adding additional set of constraints that bound the regret of the fractional solution by $r$ for every possible realization of demands:

\begin{equation*}
\begin{array}{ll}
\forall j \in C: &\sum_{i \in F} y_{ij} \geq 1\\
\forall i \in F,j \in C: &y_{ij} \leq x_i\\
 &\sum_{i \in F} x_i \leq k\\
 \forall C' \subseteq C:& \sum_{j \in C'}\sum_{i \in F} c_{ij}^p y_{ij} \leq (\opt(C') + r)^p\\
 \forall i \in F: &x_i \geq 0 \\
 \forall i \in F,j \in C: &y_{ij} \geq 0
\end{array}
\end{equation*}
}

Here, $I$ ranges over all realizable instances of the problem. Hence, the \lp is exponential in size, and we need to invoke the ellipsoid method via a separation oracle to obtain an optimal fractional solution. We first note that the constraints ${\bf x}\in P$ can be handled using a separation oracle for the optimization problem itself. So, our focus is on designing a separation oracle for the new set of constraints ${\bf c}(I){\bf \cdot x} \leq \opt(I) + {\bf r},~\forall I$. This amounts to determining the regret of a fixed solution given by $\bf x$, which unfortunately, is \np-hard for our clustering problems. So, we settle for designing an approximate separation oracle, i.e., approximating the regret of a given solution. For $k$-median, we reduce this to a submodular maximization problem subject to a cardinality constraint, which can then be (approximately) solved via standard greedy algorithms. For $k$-means, and more generally $\ell_p$-clustering, the situation is more complex. Similarly to $k$-median, maximizing regret for a fixed solution can be reduced to a set of submodular maximization problems, but deriving the functional value for a given set now requires solving the \np-hard knapsack problem. We overcome this difficulty by showing that we can use {\em fractional} knapsack solutions as surrogates for the optimal integer knapsack solution in this reduction, thereby restoring polynomial running time of the submodular maximization oracle. Finally, in the presence of fixed clients, we need to run the submodular maximization algorithm over a set of combinatorial objects called {\em independence systems}. In this case, the resulting separation oracle only gives a bicriteria guarantee, i.e., the solutions it considers feasible are only guaranteed to satisfy $\forall I: {\bf c}(I){\bf \cdot x} \leq \alpha \cdot \opt(I) + \beta \cdot {\bf r}$ for constants $\alpha$ and $\beta$. Note that the bi-criteria guarantee suffices for our purposes since these constants are absorbed in the overall approximation bounds.

The next step in our framework is to round these fractional solutions to integer solutions for the regret minimization \lp. Typically, in clustering problems such as $k$-median, \lp rounding algorithms give {\em average} guarantees, i.e., although the overall objective in the integer solution is bounded against that of the fractional solution, individual connection costs of clients are not (deterministically) preserved in the rounding. But, average guarantees are too weak for our purpose:  in a realized instance, an adversary may only select the clients whose connection costs increase by a large factor in the rounding thereby causing a large regret. Ideally, we would like to ensure that the connection cost of {\em every} individual client is preserved up to a constant in the rounding. However, this may be impossible in general, i.e., no integer solution might satisfy this requirement. Consider a uniform metric over $k+1$ points. One fractional solution is to make $\frac{k}{k+1}$ fraction of each point a cluster center. Then, each client has connection cost $\frac{1}{k+1}$ in the fractional solution since it needs to connect $\frac{1}{k+1}$ fraction to a remote point. However, in any integer solution, since there are only $k$ cluster centers but $k+1$ points overall, there is one client that has connection cost of $1$, which is $k+1$ times its fractional connection cost. 

To overcome this difficulty, we allow for a uniform {\em additive} increase in the connection cost of every client. We show that such a rounding also preserves the regret guarantee of our fractional solution within constant factors. The clustering problem we now solve has a modified objective: for every client, the distance to the closest cluster center is now discounted by the additive allowance, with the caveat that the connection cost is $0$ if this difference is negative. This variant is a generalization of a problem appearing in \cite{GuhaM09}, and we call it clustering {\em with discounts} (e.g., for $k$-median, we call this problem \textit{$k$-median with discounts}.) 
Our main tool in the rounding then becomes an approximation algorithm to clustering problems with discounts. For $k$-median, we use a Lagrangian relaxation of this problem to the classic facility location problem to design such an approximation. For $k$-means and $\ell_p$-clustering, the same general concept applies, but we need an additional ingredient called a {\em virtual} solution that acts as a surrogate between the regret of the (rounded) integer solution and that of the fractional solution obtained above. For $k$-center, we give a purely combinatorial (greedy) algorithm.

\subsection{Related Work}

For all previous universal algorithms, the approximation factor corresponds to our parameter $\alpha$, i.e., these algorithms are $(\alpha,0)$-approximate. The notion of regret was not considered. As we have explained, however, it is not possible to obtain such results for universal clustering.  Furthermore, it may be possible to trade-off some of the large values of $\alpha$ in the results below, e.g., $\Omega(\sqrt{n})$ for set cover, by allowing $\beta > 0$.

Universal algorithms have been of large interest in part because of their applications as online algorithms where all the computation is performed ahead of time. Much of the work on universal algorithms has focused on TSP. For Euclidean TSP in the plane, Platzman and Bartholdi~\cite{Platzman:1989:SCP:76359.76361} gave an $O(\log n)$-approximate universal algorithm. Hajiaghayi {\em et al.}~\cite{Hajiaghayi:2006:ILU:1109557.1109628} generalized this result to an $O(\log^2 n)$-approximation for minor-free metrics, and Schalekamp and Shmoys~\cite{SCHALEKAMP2008} gave an $O(\log n)$-approximation for tree metrics. For arbitrary metrics, Jia {\em et al.}~\cite{JiaLNRS05} presented an $O(\log^4 n / \log\log n)$-approximation, which improves to an $O(\log n)$-approximation for doubling metrics. The approximation factor for arbitrary metrics was improved to $O(\log^2 n)$ by Gupta {\em et al.}~\cite{Gupta:2006:OND:1109557.1109665}.  It is also known that these logarithmic bounds are essentially tight for universal TSP~\cite{BertsimasG89,Hajiaghayi:2006:ILU:1109557.1109628,GorodezkyKSS10,BhalgatCK11}. For the metric Steiner tree problem, Jia {\em et al.}~\cite{JiaLNRS05} adapted their own TSP algorithm to provide an $O(\log^4 n/ \log\log n)$-approximate universal algorithm, which is also tight up to the exponent of the log~\cite{AlonA92,JiaLNRS05,BhalgatCK11}.  Busch {\em et al.}~\cite{BuschDRRS12} present an $O(2^{\sqrt{\log n}})$-approximation for universal Steiner tree on general graphs and an $O(\mathrm{polylog}(n))$-approximation for minor-free graphs.  Finally, for universal (weighted) set cover, Jia {\em et al.}~\cite{JiaLNRS05} (see also \cite{GrandoniGLMSS08}) provide an $O(\sqrt{n \log n})$-approximate universal algorithm and an almost matching lower bound.

The problem of minimizing regret has been studied in the context of robust optimization. The robust $1$-median problem was introduced for tree metrics by Kouvelis and Yu in \cite{KouvelisY97} and several faster algorithms and for general metrics were developed in the following years (e.g. see \cite{AverbakhB00}).  For robust $k$-center, Averbakh and Berman\cite{AverbakhB00} gave a reduction to a linear number of ordinary $k$-center problems, and thus for classes of instances where the ordinary $k$-center problem is polynomial time solvable (e.g., instances with constant $k$ or on tree metrics) this problem is also polynomial time solvable \cite{AverbakhB97}. 
  A different notion of robust algorithms is one where a set $S$ of possible scenarios is provided as part of the input to the problem. This model was originally considered for network design problems (see the survey by Chekuri~\cite{Chekuri07}).  Anthony {\em et al.}~\cite{AnthonyGGN10} gave an $O(\log n + \log |S|)$-approximation algorithm for solving $k$-median and a variety of related problems in this model (see also \cite{BhattacharyaCMN14}) on an $n$-point metric space.
 However, note that $|S|$ can be exponential in $|C|$ in general. 
%

Another popular model for uncertainty is two-stage optimization (e.g., \cite{SwamyS06,ShmoysS06,SwamyS12,
DhamdhereGRS05,FeigeJMM07,KhandekarKMS13,GuptaNR14,GuptaNR16, GuptaPRS04,CharikarCP05}). Here, the first stage presents a set of realizable instances (or a distribution over them) and the second stage chooses one of those realizations. The algorithm is free to make choices at either stage but those choices come at a higher cost in the second stage when it has more information about the input. Because of the different costs, results in this model have no bearing on our setting.


\medskip\noindent
{\bf Roadmap.}
We present the constant approximation algorithms (Theorem~\ref{thm:upper}) for universal $k$-median, $\ell_p$-clustering ($k$-means is a special case), and $k$-center in Sections~\ref{section:kmalg}, \ref{section:lpalg}, and \ref{section:kcenteralg} respectively. In describing these algorithms, we defer the clustering with discounts algorithms used in the rounding to Appendix~\ref{section:kmwithdiscounts}. We also give the extensions to universal clustering with fixed clients for $k$-median, $k$-means/$\ell_p$-clustering, and $k$-center (Theorem~\ref{thm:upper-fixed}) in Sections~\ref{section:kmalg2}, \ref{sec:lpalg2}, and \ref{sec:kcalg2}. Finally, the hardness results for general metrics and for Euclidean metrics (Theorem~\ref{thm:lower}) appear in Sections~\ref{section:hardness} and \ref{sec:euclidean} respectively.

\section{Universal $k$-Median}\label{section:kmalg}

%
%
%

In this section, we prove the following theorem:

\begin{theorem}\label{thm:kmalgorithm}
There exists a $(27, 49)$-approximate universal algorithm for the $k$-median problem.
\end{theorem}

We follow the recipe described in Section~\ref{subsection:techniques}. Namely, 
the algorithm has two components. The first component is a separation oracle for the regret minimization \lp based on 
submodular maximization, which we define below.

\paragraph{Submodular Maximization with Cardinality Constraints.} A (non-negative) function $f:2^E \rightarrow \mathbb{R}_0^+$ is said to be {\em submodular} if for all $S \subseteq T \subseteq E$ and $x \in E$, we have $f(T \cup \{x\}) - f(T) \leq f(S \cup \{x\}) - f(S)$. It is said to be {\em monotone} if for all $S \subseteq T \subseteq E$, we have $f(T) \geq f(S)$. The following theorem for maximizing monotone submodular functions subject to a cardinality constraint is well-known.

\begin{theorem}[Fisher {\em et al.}~\cite{FisherNW78}]
\label{thm:cardinality}
For the problem of finding $S \subseteq E$ that maximizes a monotone submodular function $f:2^E \rightarrow \mathbb{R}_0^+$, the natural greedy algorithm that starts with $S = \emptyset$ and repeatedly adds $x \in E$ that maximizes $f(S \cup \{x\})$ until $|S| = k$, is a $\frac{e}{e-1} \approx 1.58$-approximation. 
\end{theorem}

We give the reduction from the separation oracle to submodular maximization in Section~\ref{sec:km-fractional}, and then employ the above theorem.
The second component of our framework is a rounding algorithm that employs the $k$-median with discounts problem, which we define below. 

\paragraph{$k$-median with Discounts.}
In the $k$-median with discounts problem, we are given a $k$-median instance, 
but where each client $j$ has an additional (non-negative) parameter $r_j$ called its {\em discount}. 
Just as in the $k$-median problem, our goal is to place $k$ cluster centers that minimize the total connection costs of all clients. 
But, the connection cost for client $j$ can now be discounted by up to $r_j$, i.e., client $j$ with connection cost $c_j$ contributes 
$(c_j - r_j)^+ := \max\{0, c_j - r_j\}$ to the objective of the solution. 

Let $\opt$ be the cost of an optimal solution to the $k$-median with discounts problem. 
We say an algorithm $\alg$ that outputs a solution with connection cost $c_j$ for client $j$ is a $(\gamma, \sigma)$-approximation if:
$$\sum_{j \in C} (c_j - \gamma \cdot r_j)^+ \leq \sigma \cdot \opt.$$
That is, a $(\gamma, \sigma)$-approximate algorithm outputs a solution whose objective function when computed using discounts 
$\gamma\cdot r_j$ for all $j$ is at most $\sigma$ times the optimal objective using discounts $r_j$. In the case where all $r_j$ are equal, \cite{GuhaM09} gave a $(9, 6)$-approximation algorithm for this problem based on the classic primal-dual algorithm for $k$-median.
The following lemma generalizes their result to the setting where the $r_j$ may differ:
\begin{lemma}\label{lemma:kmbicriteria}
        There exists a (deterministic) polynomial-time $(9, 6)$-approximation algorithm for the $k$-median with discounts problem.
\end{lemma}
We give details of the algorithm and the proof of this lemma in Appendix~\ref{section:kmwithdiscounts}. We note that when all $r_j$ are equal, the constants in \cite{GuhaM09} can be improved (see e.g. \cite{ChakrabartyS19}); we do not know of any similar improvement when the $r_j$ may differ. In Section~\ref{sec:km-rounding}, we give the reduction from rounding the fractional solution for universal $k$-median to the $k$-median with discounts problem, and then employ the above lemma.

\subsection{Universal $k$-median: Fractional Algorithm}
\label{sec:km-fractional}


The standard $k$-median polytope (see e.g., \cite{JainV01}) is given by:
$$P = \{(x, y): \sum_{i} x_i \leq k; \forall i,j: y_{ij} \leq x_i; \forall j: \sum_{i} y_{ij} \geq 1; \forall i,j: x_i, y_{ij} \in [0, 1]\}.$$
Here, $x_i$ represents whether point $i$ is chosen as a cluster center, and $y_{ij}$ represents whether client $j$ connects to $i$ as its cluster center.
Now, consider the following \lp formulation for minimizing regret $r$:
\begin{equation}\label{eq:mainlp}
\min\{r: (x, y)\in P; \forall C' \subseteq C: \sum_{j \in C'} \sum_i c_{ij} y_{ij} - \opt(C') \leq r\},
\end{equation}
%
%
%
%
%
%
%
%
%
%
where $\opt(C')$ is the cost of the (integral) optimal solution in realization $C'$. Note that the new constraints: $\forall C' \subseteq C: \sum_{j \in C'} \sum_i c_{ij} y_{ij} - \opt(C') \leq r$ (we call it the regret constraint set) require that the regret is at most $r$ in all realizations. 

In order to solve \lp~\eqref{eq:mainlp}, we need a separation oracle for the regret constraint set. Note that there are exponentially many constraints corresponding to realizations $C'$; moreover, even for a single realization $C'$, computing $\opt(C')$ is \np-hard. So, we resort to designing an {\em approximate} separation oracle.
%
%
%
Fix some fractional solution $(x, y, r)$. Overloading notation, let $S(C')$ denote the cost of the solution with cluster centers $S$ in realization $C'$. By definition, $\opt(C') = \min_{S \subseteq F, |S| = k} S(C')$. Then designing a separation oracle for the regret constraint set is equivalent to determining if the following inequality holds:
%
%
%
%
%
%
%

$$\max_{C'\subseteq C} \quad \max_{S \subseteq F, |S| = k} \left[\sum_{j \in C'} \sum_i c_{ij} y_{ij} - S(C')\right] \leq r.$$
We flip the order of the two maximizations, and define $f_y(S)$ as follows:
$$f_y(S) = \max_{C'\subseteq C} \left[\sum_{j \in C'} \sum_i c_{ij} y_{ij} - S(C')\right].$$
Then designing a separation oracle is equivalent to maximizing $f_y(S)$ for $S \subseteq F$ subject to $|S| = k$. The rest of the proof consists of showing that this function is monotone and submodular, and efficiently computable. 
%
%
%
%
%
%

\begin{lemma}\label{lemma:lbsubmodularity}
Fix $y$. Then, $f_y(S)$ is a monotone submodular function in $S$. Moreover, $f_y(S)$ is efficiently computable for a fixed $S$.
\end{lemma}
\begin{proof}
Let $d(j, S) := \min_{i'\in S} c_{i'j}$ denote the distance from client $j$ to the nearest cluster center in $S$. If $S = \emptyset$, we say $d(j, S) := \infty$. 
The value of $C'$ that defines $f_y(S)$ is the set of all clients closer to $S$ than to the fractional solution $y$, i.e., $\sum_i c_{ij}y_{ij} > \min_{i'\in S} c_{i'j}$. This immediately establishes efficient computability of $f_y(S)$. Moreover, we can equivalently write $f_y(S)$ as follows:
$$f_y(S) = \sum_{j \in C} (\sum_i c_{ij} y_{ij} - d(j, S))^+.$$
A sum of monotone submodular functions is a monotone submodular function, so it suffices to show that for all clients $j$, the new function $g_{y, j}(S) := (\sum_i c_{ij} y_{ij} - d(j, S))^+$ is monotone submodular. 
\begin{itemize}
\item
$g_{y, j}$ is {\em monotone}: for $S \subseteq T$, $d(j, T) \leq d(j, S)$, and thus $(\sum_i c_{ij} y_{ij} - d(j, S))^+ \leq (\sum_i c_{ij} y_{ij} - d(j, T))^+$.
\item
$g_{y, j}$ is submodular if: 
$$\forall S \subseteq T \subseteq F, \forall x \in F: g_{y,j}(S \cup \{x\}) - g_{y,j}(S) \geq g_{y,j}(T \cup \{x\}) - g_{y,j}(T)$$
Fix $S$, $T$, and $x$. Assume $g_{y,j}(T \cup \{x\}) - g_{y,j}(T)$ is positive (if it is zero, by monotonicity the above inequality trivially holds). This implies that $x$ is closer to client $j$ than any cluster center in $T$ (and hence $S$ too), i.e., $d(j, x) \leq d(j, T) \leq d(j, S)$. Thus, $d(j, x) = d(j, S \cup \{x\}) = d(j, T \cup \{x\})$ which implies that $g_{y,j}(S \cup \{x\}) = g_{y,j}(T \cup \{x\})$. Then we just need to show that $g_{y,j}(S) \leq g_{y,j}(T)$, but this holds by monotonicity. \qedhere
\end{itemize}
\end{proof}
%
%
%

%
%

By standard results (see e.g., GLS~\cite{GrotschelLS81}), 
we get an $(\alpha, \beta)$-approximate fractional solution for universal $k$-median via the ellipsoid method if we have an approximate separation oracle for \lp~\eqref{eq:mainlp} that given a fractional solution $(x, y, r)$ does either of the following:
\begin{itemize}
    \item Declares $(x, y, r)$ feasible, in which case $(x, y)$ has cost at most $\alpha \cdot \opt(\bd) + \beta \cdot r$ in all realizations, or
    \item Outputs an inequality violated by $(x, y, r)$ in \lp~\eqref{eq:mainlp}.
\end{itemize}
%
%
%
%
%

The approximate separation oracle does the following for the regret constraint set (all other constraints can be checked exactly): Given a solution $(x, y, r)$, find an $\frac{e-1}{e}$-approximate maximizer $S$ of $f_y$ via Lemma~\ref{lemma:lbsubmodularity} and Theorem~\ref{thm:cardinality}. Let $C'$ be the set of clients closer to $S$ than the fractional solution $y$ (i.e., the realization that maximizes $f_y(S)$).
If $f_y(S) > r$, the separation oracle returns the violated inequality $\sum_{j \in C'} \sum_i c_{ij} y_{ij} - S(C') \leq r$; else, it declares the solution feasible. Whenever the actual regret of $(x, y)$ is at least $\frac{e}{e-1}\cdot r$, this oracle will find $S$ such that $f_y(S) > r$ and output a violated inequality. Hence, we get the following lemma:

\begin{lemma}\label{lemma:lpsolution}
There exists a deterministic algorithm that in polynomial time computes a fractional $\frac{e}{e-1} \approx 1.58$-approximate solution for \lp~\eqref{eq:mainlp} representing the universal $k$-median problem.
\end{lemma}

\subsection{Universal $k$-Median: Rounding Algorithm}
\label{sec:km-rounding}

%
%
%
%

Let $\fracc$ denote the $\frac{e}{e-1}$-approximate fractional solution to the universal $k$-median problem provided by Lemma~\ref{lemma:lpsolution}.
We will use the following property of $k$-median, shown by Archer {\em et al.}~\cite{ArcherRS03}.
\begin{lemma}[\cite{ArcherRS03}]
\label{lma:km-intgap}
The integrality gap of the natural \lp relaxation of the $k$-median problem is at most $3$.
\end{lemma}
Lemmas~\ref{lemma:lpsolution} and \ref{lma:km-intgap} imply that that for any set of clients $C'$, 
\begin{equation}
\label{eq:frac}
\frac{1}{3} \cdot \opt(C') \leq \fracc(C') \leq \opt(C') + \frac{e}{e-1} \cdot \mr.    
\end{equation}
Our overall goal is to obtain a solution \sol that minimizes $\max_{C'\subseteq C}\left[\sol(C') - \opt(C')\right]$.
But, instead of optimizing over the exponentially many different $\opt(C')$ solutions, we use the 
surrogate $3\cdot \fracc(C')$ which has the advantage of being defined by a fixed solution $\fracc$,
but still 3-approximates $\opt(C')$ by Eq.~\ref{eq:frac}.
This suggests minimizing the following objective instead:
$\max_{C'}[\sol(C') - 3\cdot\fracc(C')]$.  
For a given solution $\sol$, the set of clients $C'$ that maximizes the new expression are the clients whose connection costs in $\sol$ (denoted $c_j$) exceeds $3$ times their cost in $\fracc$ (denoted $f_j$):
$$\max_{C'}[\sol(C') - 3\cdot\fracc(C')] = \sum_{j \in C} (c_j-3f_j)^+.$$  
But, minimizing this objective is precisely the aim of the $k$-median with discounts problem, where the discount for client $j$ is $3 f_j$. This allows us to invoke  Lemma~\ref{lemma:kmbicriteria} for the $k$-median with discounts problem.

Thus, our overall algorithm is as follows. First, use Lemma~\ref{lemma:lpsolution} to find a fractional solution $\fracc = (x, y, r)$. Let $f_j := \sum_i c_{ij}y_{ij}$ be the connection cost of client $j$ in $\fracc$. Then, construct a $k$-median with discounts instance where client $j$ has discount $3f_j$, and use Lemma~\ref{lemma:kmbicriteria} on this instance to obtain the final solution to the universal $k$-median problem.

We now complete the proof of Theorem~\ref{thm:kmalgorithm} using the above lemmas.

\begin{proof}[Proof of Theorem \ref{thm:kmalgorithm}]
Let $m_j$ be the connection cost of \mrs to client $j$. 
%
%
%
%
%
%
Then,
\begin{align*}
\mr = \max_{C' \subseteq C}[\mrs(C') - \opt(C')] &\geq \max_{C' \subseteq C}[\mrs(C') - 3\cdot\fracc(C')] \quad \text{(by Eq.~\eqref{eq:frac})} \\
&= \sum_{j \in C: m_j > 3f_j}(m_j - 3f_j) =\sum_{j \in C}(m_j - 3f_j)^+.
\end{align*}
%
%
Thus, $\mr$ upper bounds the optimal objective in the $k$-median with discounts instance that we construct. 
Let $c_j$ be the connection cost of client $j$ in the solution output by the algorithm. Then by Lemma \ref{lemma:kmbicriteria} we get that: 
\begin{equation}
\label{eq:disc}
\sum_{j \in C}(c_j - 27f_j)^+ \leq 6 \cdot \sum_{j \in C} (m_j - 3f_j)^+ \leq 6\cdot\mr.
\end{equation}
As a consequence, we have:
$$\forall C' \subseteq C: \sum_{j \in C'} c_j 
= \sum_{j \in C'} [27f_j + (c_j - 27f_j)] \leq \sum_{j \in C'}27f_j + \sum_{j \in C'}(c_j - 27f_j)^+ 
\leq 27\cdot \fracc(C') + 6\cdot\mr,$$
where the last step uses the definition of $f_j$ and Eq.~\eqref{eq:disc}.
Now, using the bound on $\fracc(C')$ from Eq.~\eqref{eq:frac} in the inequality above, we have
the desired bound on the cost of the algorithm:
$$\forall C' \subseteq C: 
\sum_{j \in C'} c_j 
\leq 27\cdot\fracc(C') + 6\cdot\mr 
\leq 27\left[\opt(C') + \frac{e}{e-1}\cdot\mr\right] + 6\cdot\mr 
\leq 27 \cdot\opt(C') + 49\cdot\mr.\mbox{\qedhere}$$
\end{proof}

\section{Universal $\ell_p$-Clustering and Universal $k$-means}\label{section:lpalg}

In this section, we give universal algorithms for $\ell_p$-clustering with the following guarantee:

\begin{theorem}\label{thm:lpalg}
For all $p \geq 1$, there exists a $(54p, 103p^2)$-approximate 
universal algorithm for the $\ell_p$-clustering problem.
\end{theorem}
As a corollary, we obtain the following result for universal $k$-means ($p=2$).
\begin{corollary}\label{thm:kmeansalg}
There exists a $(108, 412)$-approximate universal algorithm for the $k$-means problem.
\end{corollary}

Before describing further details of the universal $\ell_p$-clustering algorithm, 
we note a rather unusual feature of the universal clustering framework. Typically,
in standard $\ell_p$-clustering, the algorithms effectively optimize the $\ell_p^p$
objective (e.g., sum of squared distances for $k$-means) because these are equivalent
in the following sense: an $\alpha$-approximation for the $\ell_p$ objective is equivalent 
to an $\alpha^p$-approximation for the $\ell_p^p$ objective. But, this equivalence fails
in the setting of universal algorithms for reasons that we discuss below. 
Indeed, we first give a universal
$\ell_p^p$-clustering algorithm, which is a simple extension of the $k$-median
algorithm, and then give universal $\ell_p$-clustering algorithms, which turns out to be much more
challenging.


\subsection{Universal $\ell_p^p$-Clustering}
As in universal $k$-median, we can write an $\lp$ formulation for universal $\ell_p^p$-clustering, i.e. clustering with the objective $\sol(C') = \sum_{j\in C'} \cost(j, \sol)^p$:

\begin{equation}\label{eq:mainlp-ellpp}
\min\{r: (x, y)\in P; \forall C' \subseteq C: \sum_{j \in C'} \sum_i c_{ij}^p y_{ij} - \opt(C') \leq r\},
\end{equation}
where $P$ is still the $k$-median polytope defined in Section~\ref{sec:km-fractional}.

The main difficulty is that the $\ell_p^p$ distances no longer form a metric, i.e.,
do not satisfy triangle inequality.
Nevertheless, the distances still have a metric connection, that they are the $p$th power of metric distances. 
We show that this connection is sufficient to prove the following result:
\begin{theorem}\label{thm:lpp}
For all $p \geq 1$, there exists a $(27^p,27^p \cdot \frac{e}{e-1} + \frac{2}{3} \cdot 9^p)$-approximate 
algorithm for the universal  $\ell_p^p$ clustering problem.
\end{theorem}

As in universal $k$-median, a key component in proving Theorem~\ref{thm:lpp} is a rounding algorithm
that employs a bi-criteria approximation to the {\em $\ell_p^p$-clustering with discounts} problem.
Indeed, this result will also be useful in the next subsection, when we consider the universal 
$\ell_p$-clustering problem. So, we formally define $\ell_p^p$-clustering with discounts problem 
below and state our result for it.

\paragraph{$\ell_p^p$-clustering with Discounts.}
In this problem, are given a $\ell_p^p$-clustering instance, 
but where each client $j$ has an additional (non-negative) parameter $r_j$ called its {\em discount}. 
Our goal is to place $k$ cluster centers that minimize the total connection costs of all clients. 
But, the connection cost for client $j$ can now be discounted by up to $r_j^p$, i.e., client $j$ with connection cost $c_j$ contributes 
$(c_j^p - r_j^p)^+ := \max\{0, c_j^p - r_j^p\}$ to the objective of the solution. 
(Note that the $k$-median with discounts problem that we described in the previous section is a special case of 
this problem for $p=1$.)

Let $\opt$ be the cost of an optimal solution to the $\ell_p^p$-clustering with discounts problem. 
We say an algorithm $\alg$ that outputs a solution with connection cost $c_j$ for client $j$ is a $(\gamma^p, \sigma)$-approximation\footnote{We refer to this as a $(\gamma^p, \sigma)$-approximation instead of a $(\gamma, \sigma)$-approximation to emphasize the difference between the scaling factor for discounts $\gamma$ and the loss in approximation factor $\gamma^p$.} if:
$$\sum_{j \in C} (c_j^p - \gamma^p \cdot r_j^p)^+ \leq \sigma \cdot \opt.$$
That is, a $(\gamma^p, \sigma)$-approximate algorithm outputs a solution whose objective function computed using discounts $\gamma \cdot r_j$  for all $j$ is at most $\sigma$ times the optimal objective using discounts $r_j$. 
We give the following result about the $\ell_p^p$-clustering with discounts problem (see Appendix~\ref{section:kmwithdiscounts} for details):
\begin{lemma}\label{lemma:kmbicriteria-lpp}
        There exists a (deterministic) polynomial-time $(9^p, \frac{2}{3}\cdot 9^p)$-approximation algorithm for the $\ell_p^p$-clustering with discounts problem.
\end{lemma}
We now employ this lemma in obtaining Theorem~\ref{thm:lpp}. 
Recall that the universal $k$-median result in the previous section had three main ingredients: 
\begin{itemize}
    \item 
    Lemma~\ref{lemma:lpsolution} to obtain an $\frac{e}{e-1}$-approximate fractional solution. This continues to hold
    for the $\ell_p^p$ objective, since Lemma ~\ref{lemma:lpsolution} does not use any metric property.
    \item
    An upper bound of $3$ on the integrality gap of the natural \lp relaxation of $k$-median from \cite{ArcherRS03}. 
    The same result now gives an upper bound of $3^p$ on the integrality gap of $\ell_p^p$-clustering.
    \item
    Lemma~\ref{lemma:kmbicriteria} to obtain an approximation guarantee for the $k$-median with discounts problem. 
    This is where the metric property of the connection costs in the $k$-median problem was being used. Nevertheless,
    Lemma~\ref{lemma:kmbicriteria-lpp} above gives a generalization of Lemma~\ref{lemma:kmbicriteria}
    to the $\ell_p^p$-clustering with discounts problem. 
\end{itemize}
Theorem~\ref{thm:lpp} now follows from these three observations using exactly the same steps as Theorem~\ref{thm:kmalgorithm}
in the previous section; we omit these steps for brevity.\qed

\eat{
\begin{proof}
The proof of Lemma~\ref{lemma:lpsolution} does not use any properties of metric distance functions, and so still gives a $\frac{e}{e-1}$-approximate fractional solution for universal $\ell_p^p$-clustering. Thus, it remains to extend the analysis of the $k$-median with discounts algorithm to $\ell_p^p$ clustering. This is  (see Appendix~\ref{section:kmwithdiscounts}):

Finally, the result of \cite{ArcherRS03} also gives an upper bound of $3^p$ on the integrality gap of $\ell_p^p$ clustering (for the same reasons as in Lemma~\ref{lemma:kmbicriteria-lpp}; wherever the factor of $3$ is lost due to triangle inequality in their analysis for $k$-median, a factor of $3^p$ is lost instead for $\ell_p^p$ clustering). The algorithm and analysis follow as that of Theorem~\ref{thm:kmalgorithm}, with Lemma~\ref{lemma:kmbicriteria-lpp} and the $3^p$ integrality gap substituted in.
\end{proof}
}

\bigskip

The rest of this section is dedicated to the universal $\ell_p$-clustering problem.
As for $k$-median, we have two stages, the fractional algorithm and the rounding 
algorithm, that we present in the next two subsections.

\subsection{Universal $\ell_p$-Clustering: Fractional Algorithm}\label{sec:ellp-frac}
Let us start by describing the fractional relaxation of the universal $\ell_p$-clustering problem\footnote{The constraints are not simultaneously linear in $y$ and $r$, although fixing $r$, we can write these constraints as $\sum_{j \in C'} \sum_{i} c_{ij}^p y_{ij} \leq (\opt(C') + r)^p$, which is linear in $y$. In turn, to solve this program we bisection search over $r$, using the ellipsoid method to determine if there is a feasible point for each fixed $r$.} 
(again, $P$ is the $k$-median polytope defined as in Section~\ref{sec:km-fractional}):

\begin{equation}\label{eq:mainlp-ellp}
\min\{r: (x, y)\in P; \forall C' \subseteq C: \left(\sum_{j \in C'} \sum_i c_{ij}^p y_{ij}\right)^{1/p} - \opt(C') \leq r\},
\end{equation}

As described earlier, when minimizing regret, the $\ell_p$ and $\ell_p^p$ objectives are no longer equivalent. For instance,  recall that one of the key steps in Lemma~\ref{lemma:lpsolution} was to establish the submodularity of the function $f_y(S)$ denoting the maximum regret caused by any realization when comparing two given solutions: a fractional solution $y$ and an integer solution $S$. Indeed, the worst case realization had a simple structure: choose all clients that have a smaller connection cost for $S$ than for $y$. This observation continues to hold for the $\ell_p^p$ objective because of the linearity of $f_y(S)$ as a function of the realized clients once $y$ and $S$ are fixed. But, the $\ell_p$ objective is not linear even after fixing the solutions, and as a consequence, we lose both the simple structure of the maximizing realization as well as the submodularity of the overall function $f_y(S)$. For instance, consider two clients: one at distances $1$ and $0$, and another at distances $1+\epsilon$ and $1$, from $y$ and $S$ respectively. Using the $\ell_p$ objective, the regret with both clients is $(2+\epsilon)^{1/p} - 1$, whereas with just the first client the regret is $1$, which is larger for $p \geq 2$.


The above observation results in two related difficulties: first, that $f_y(S)$ is not submodular and hence standard submodular maximization techniques do not apply, but also that given $y$ and $S$, we cannot even compute the function $f_y(S)$ efficiently. To overcome this difficulty, we further refine the function $f_y(S)$ to a collection of functions $f_{y, Y} (S)$ by also fixing the cost of the fractional solution $y$ to at most a given value $Y$. As we will soon see, this allows us to relate the $\ell_p$ objective to the $\ell_p^p$ objective, but under an additional ``knapsack''-like packing constraint.
It is still not immediate that $f_{y, Y}$ is efficiently computable because of the knapsack constraint that we have introduced. Our second observation is that relaxing the (\np-hard) integer knapsack problem to the corresponding (poly-time) {\em fractional} knapsack problem does not affect the optimal value of $f_{y, Y}(S)$ (i.e., allowing fractional clients does not increase $y$'s regret), while making the function efficiently computable. As a bonus, the relaxation to fractional knapsack also restores submodularity of the function, allowing us to use standard maximization tools as earlier. We describe these steps in detail below.

To relate the regret in the $\ell_p$ and $\ell_p^p$ objectives, let $\fracc_p^q(C')$ and $S_p^q(C')$  denote the $\ell_p$ objective to the $q$th power for $y$ and $S$ respectively in realization $C'$ (and let $\fracc_p(C'), S_p(C')$ denote the corresponding $\ell_p$ objectives). Assume that $y$'s regret against $S$ is non-zero. Then:

\begin{align*}
\max_{C' \subseteq C} \left[\fracc_p(C') - S_p(C')\right] &= \max_{C' \subseteq C} \left[\frac{\fracc_p^p(C') - S_p^p(C')}{\sum_{q = 0}^{p-1} \fracc_p^q(C') S_p^{p-1-q}(C')}\right]\\
&\simeq_p \max_{C' \subseteq C} \left[\frac{\fracc_p^p(C') - S_p^p(C')}{\fracc_p^{p-1}(C')}\right]\\
&= \max_{Y} \max_{C' \subseteq C: \fracc_p^p(C') \leq Y} \left[\frac{\fracc_p^p(C') - S_p^p(C')}{Y^{1-1/p}} \right]\\
&= \max_{Y}  \left[\frac{\max_{C' \subseteq C, \fracc_p^p(C') \leq Y}[\fracc_p^p(C') - S_p^p(C')]}{Y^{1-1/p}} \right].\\
\end{align*}

The $\simeq_p$ denotes equality to within a factor of $p$, and uses the fact that if the regret is non-zero, then for every $C'$ such that $\fracc_p(C') > S_p(C')$ (one of which is always the maximizer of all expressions in this equation), every term in the sum in the denominator is upper bounded by $\fracc_p^{p-1}(C')$. 

We would like to argue that the denominator, 

\[\max\{\fracc_p^p(C') - S_p^p(C'): C' \subseteq C, \fracc_p^p(C') \leq Y\},\]

is a submodular function of $S$. If we did this, then we could find an adversary and realization of clients that (approximately) maximizes the regret of $y$ by iterating over all (discretized) values of $Y$.
But, as described above, it is easier to work with its fractional analog 
because of the knapsack constraint: 
$$f_{y, Y}(S) := \max\{\fracc_p^p(\bd) - S_p^p(\bd): \bd \in [0, 1]^C, \fracc_p^p(\bd) \leq Y\}.$$

Here, $\fracc_p^p(\bd) := \sum_{j \in C} d_j \cdot \sum_{i \in F} c_{ij}^p y_{ij}$ and $S_p^p(\bd) := \sum_{j \in C} d_j \cdot \min_{i \in S} c_{ij}^p$ are the natural extensions of the $\ell_p^p$ objective to fractional clients. 
The next lemma shows that allowing fractional clients does not affect the maximum regret:

\begin{lemma}\label{lemma:maxatintegral}
For any two solutions $y, S$, there exists a global maximum of $\fracc_p(\bd) - S_p(\bd)$ over $\bd\in[0, 1]^C$ where all the clients are integral, i.e., $\bd\in \{0, 1\}^C$. Therefore, 
$$\max_{\bd \in [0, 1]^C}\left[\fracc_p(\bd) - S_p(\bd)\right] = \max_{C' \subseteq C} \left[\fracc_p(C') - S_p(C')\right].$$
\end{lemma}
We remark that unlike for the $\ell_p^p$ objective, integrality of the maximizer is not immediate for the $\ell_p$ objective because the regret of $y$ compared to $S$ is not a linear function of $\bd$. 

\begin{proof}
We will show that the derivative of $\fracc_p(\bd) - S_p(\bd)$ when $\fracc_p(\bd) > S_p(\bd)$ with respect to a fixed $d_j$ is either always positive or always negative for $d_j \in (0, 1)$, or negative while $0 < d_j < d'$ and then positive afterwards. This gives that any $\bd$ with a fractional coordinate where  $\fracc_p(\bd) > S_p(\bd)$ (which is necessary for a fractional $\bd$ to be a global maximum but not the integral all-zeroes vector) cannot be a local maximum, giving the lemma.

To show this property of the derivative, letting $\bd_{-j}$ denote $\bd$ with $d_j = 0$, we have $\fracc_p(\bd) - S_p(\bd) = (\fracc_p^p(\bd_{-j}) + c_1d_j)^{1/p} - (S_p^p(\bd_{-j}) + c_2d_j)^{1/p}$ where $c_1, c_2$ are the $\ell_p^p$ distance from $j$ to $\fracc, S$ respectively. For positive $d_j$, the derivative with respect to $d_j$ is well-defined and equals

$$\frac{1}{p}\left(\frac{c_1}{(\fracc_p^p(\bd_{-j}) + c_1d_j)^{1-1/p}} - \frac{c_2}{(S_p^p(\bd_{-j}) + c_2d_j)^{1-1/p}}\right).$$

The derivative is positive if the following inequality holds:

$$\frac{S_p^p(\bd_{-j}) + c_2d_j}{\fracc_p^p(\bd_{-j}) + c_1d_j} > (\frac{c_2}{c_1})^{\frac{p}{p-1}}.$$

We first focus on the case where $c_1 > c_2$. The left-hand side starts at $S_p^p(\bd_{-j}) / \fracc_p^p(\bd_{-j})$ and monotonically and asymptotically approaches $c_2/c_1$ as $d_j$ increases. This implies that either it is always at least $c_2 / c_1$ or it is increasing and approaching $c_2/c_1$, i.e. at least $(c_2/c_1)^{p/(p-1)}$ for $d_j > 0$ or less than $(c_2/c_1)^{p/(p-1)}$ for all $d_j < d'$ for some $d'$ and then greater than $(c_2/c_1)^{p/(p-1)}$ for all $d_j > d'$ (here, we are using the fact that $c_1 > c_2$ and so $(c_2/c_1)^{p/(p-1)} < c_2 / c_1$). In turn, the desired property of the derivative holds.

In the case where $c_1 \leq c_2$, at any optimum $S_p^p(\bd_{-j}) + c_2d_j < \fracc_p^p(\bd_{-j}) + c_1d_j$ (otherwise $\fracc_p^p(\bd) < S_p^p(\bd)$ and so $\bd$ cannot be a maximum because the all zeroes vector achieves a better objective) and so the derivative is always negative in this case as desired.
\end{proof}


It turns out that relaxing to fractional clients not only helps in efficient computability of the function $f_{y, Y}(S)$, but also simplifies the proof of submodularity of the function. 


\begin{lemma}\label{lemma:fracsubmodularity}
The function $f_{y,Y}(S)$ as defined above is submodular.
\end{lemma}
\begin{proof}[Proof of Lemma~\ref{lemma:fracsubmodularity}]
Fix a universal clustering instance, fractional solution $y$, and value $Y$. Consider any $S \subseteq T \subseteq F$ and $x \in F$.  $f_{y,Y}(T \cup \{x\}) - f_{y,Y}(T) \leq f_{y,Y}(S \cup \{x\}) - f_{y,Y}(S)$. $f_{y,Y}(S)$ is the optimum of a fractional knapsack instance where each client is an item with value equal to the difference between its contribution to $\fracc_p^p(\bd) - S_p^p(\bd)$ and weight equal to its contribution to $\fracc_p^p(\bd)$, with the knapsack having total weight $Y$. For simplicity we can assume there is a dummy item with value 0 and weight $Y$ in the knapsack instance as well (a value 0 item cannot affect the optimum). Note that the weights are fixed for any $S$, and the values increase monotonically with $S$. We will refer to the latter fact as \textit{monotonicity of values}. The optimum of a fractional knapsack instance is given by sorting the items in decreasing ratio of value to weight, and taking the (fractional) prefix of the items sorted this way that has total weight $Y$. We will refer to this fact as the \textit{prefix property}. We will show that we can construct a fractional knapsack solution that when using cluster centers $S \cup \{x\}$ has value at least $f_{y,Y}(S) + f_{y,Y}(T \cup \{x\}) - f_{y,Y}(T)$, proving the lemma. For brevity, we will refer to fractions of clients which may be in the knapsack as if they were integral clients.

Consider $f_{y,Y}(T \cup \{x\}) - f_{y,Y}(T)$. We can split this difference into four cases:
\begin{enumerate}
    \item A client in the optimal knapsack for $S$, $T$, and $T \cup \{x\}$.
    \item A client in the optimal knapsack for $S$ and $T \cup \{x\}$ but not for $T$.
    \item A client in the optimal knapsack for $T$ and $T \cup \{x\}$ but not for $S$.
    \item A client in the optimal knapsack for $T \cup \{x\}$ but not for $S$ or $T$. 
\end{enumerate}
In every case, the client's value must have increased (otherwise, it cannot contribute to the difference in cases 1 and 3, or it must also be in $T$'s knapsack in cases 2 and 4), i.e. $x$ is the closest cluster center to the client in $T \cup \{x\}$ (and thus $S \cup \{x\}$). Let $w_1, w_2, w_3, w_4$ be the total weight of clients in each case. The total weight of clients in $T$'s knapsack but not $T \cup \{x\}$ is $w_2 + w_4$. We will refer to these clients as replaced clients. The increase in value due to cases 2 and 4 can be thought of as replacing the replaced clients with case 2 and 4 clients. In particular, we will think of the case 2 clients as replacing the replaced clients of weight $w_2$ with the smallest total value for $T$, and the case 4 clients as replacing the remaining replaced clients (i.e. those with the largest total value for $T$).

Without loss of generality, we assume there are no case 1 clients. By the prefix property, any of the knapsack instances for $S, T, T \cup \{x\}$ (and also $S \cup \{x\}$ by monotonicity of values and the prefix property) has optimal value equal to the total value of case 1 clients plus the optimal value of a smaller knapsack instance with total weight $Y - w_1$ and all clients except case 1 clients available. The value of case 1 clients for $S \cup \{x\}$ and $T \cup \{x\}$ is the same (since the values are determined by $x$), and can only be smaller for $S$ than $T$ by monotonicity of values. In turn, we just need to construct a knapsack for $S \cup \{x\}$ for the smaller instance with no case 1 clients whose value is at least that of $S$ plus the contribution to $f_{y,Y}(T \cup \{x\}) - f_{y,Y}(T)$ from cases 2-4.

To build the desired knapsack for $S \cup \{x\}$, we start with the knapsack for $S$. The case 2 clients in $S$'s knapsack by the prefix property have less value for $S$ than for $T$ by monotonicity of values. By the prefix property, for $T$ the case 2 clients have less value than the replaced clients of total weight $w_2$ with the smallest total value for $T$ (since the former are not in the optimal knapsack for $T$, and the latter are). So, the increase in value of the case 2 clients in $S$'s knapsack is at least the contribution to $f_{y,Y}(T \cup \{x\}) - f_{y,Y}(T)$ due to replacing clients in $T$'s knapsack with case 2 clients.

To account for the case 4 clients, we take the clients in $S$'s knapsack which are not case 2 clients with weight $w_4$ and the least total value for $T$, and replace them with the case 4 clients. These clients are among the clients of total weight $w_2 + w_4$ with the lowest value-to-weight ratios (for $T$) in $S$'s knapsack (they aren't necessarily the clients of total weight $w_4$ with the lowest value-to-weight ratios, because we chose not to include case 2 clients in this set). On the other hand, the replaced clients in $T$'s knapsack all have value-to-weight ratios for $T$ greater than at least $w_2$ weight of other clients in $T$'s knapsack (those replaced by the case 2 clients). So by monotonicity of values and the prefix property, the clients we replace in $S$'s knapsack with case 4 clients have lower value for $S$ than the clients being replaced by case 4 clients do for $T$, and so we increase the value of our knapsack by more than the contribution to $f_{y,Y}(T \cup \{x\}) - f_{y,Y}(T)$ due to case 4 clients.

Lastly, we take any clients in $S$'s knapsack which are not in $T$'s knapsack or case 2 or case 4 clients with total weight $w_3$ and replace them with the case 3 clients. Since these clients are not in $T$'s knapsack, their value for $T$ (and thus their value for $S$ by monotonicity of values) is less than the case 3 clients' value for $T$ by the prefix property. In turn, this replacement increases the value of our knapsack by more than the contribution to $f_{y,Y}(T \cup \{x\}) - f_{y,Y}(T)$ due to case 3 clients.
\end{proof}

This lemma allows us to now give an approximate separation oracle for fractional solutions of universal $\ell_p$-clustering by trying all guesses for $Y$ (\textproc{$\ell_p$-SepOracle} in Figure~\ref{fig:lpsep}). 

\begin{figure}[ht]
\fbox{\begin{minipage}{\textwidth}
\textproc{$\ell_p$-SepOracle}($(x, y, r)$, $F$, $C$):\newline
\textbf{Input:} Fractional solution $(x,y,r)$, set of cluster centers $F$, set of all clients $C$
\begin{algorithmic}[1]
\If{Any constraint in \eqref{eq:mainlp-ellp} except the regret constraint set is violated}
\State \Return the violated constraint
\EndIf
\State $c_{\min} \leftarrow \min_{i \in F, j \in C} c_{ij}^p, c_{\max} \leftarrow \sum_{j \in C} \max_{i \in F} c_{ij}^p$
\For {$Y \in \{ c_{\min}, c_{\min}(1+\epsilon'), c_{\min}(1+\epsilon')^2, \ldots c_{\min} (1+\epsilon')^{\lceil \log_{1+\epsilon'} c_{\max}/c_{\min}  \rceil}\}$}
\State $S \leftarrow \frac{e-1}{e}$-maximizer of $f_{y, Y}$ subject to $|S| \leq k$ via Theorem~\ref{thm:cardinality} 
\State $\bd' \leftarrow \argmax_{\bd \in [0, 1]^C: \sum_{j \in C} d_j \sum_{i \in F}  c_{ij}^p y_{ij}\leq Y} \sum_{j \in C} d_j \sum_{i \in F} c_{ij}^p y_{ij} - \sum_{j \in C} d_j \min_{i \in S} c_{ij}^p$ \label{line:demand} 
\If{$\frac{1}{pY^{1-1/p}}\left[\sum_{j \in C} d_j' \sum_{i \in F} c_{ij}^p y_{ij} - \sum_{j \in C} d_j' \min_{i \in S} c_{ij}^p\right] > r$}
\State \Return $\frac{1}{pY^{1-1/p}}\left[\sum_{j \in C} d_j' \sum_{i \in F} c_{ij}^p y_{ij} - \sum_{j \in C} d_j' \min_{i \in S} c_{ij}^p\right] \leq r$
\EndIf
\EndFor 
\State \Return ``Feasible''
\end{algorithmic}
\end{minipage}}
\caption{Separation Oracle for $\ell_p$-Clustering}
\label{fig:lpsep}
\end{figure}


One complication of using the above as a separation oracle in the ellipsoid algorithm is that it outputs linear constraints whereas the actual constraints in the fractional relaxation given in \eqref{eq:mainlp-ellp} are non-linear. So, we need some additional work to show that violation of the linear constraints output by \textsc{$\ell_p$-SepOracle} also implies violation of the non-linear constraints in \eqref{eq:mainlp-ellp}. We give these details in Appendix~\ref{sec:proofs} and prove the next lemma.

\begin{lemma}\label{lemma:lp-fractional}
For any $p \geq 1, \epsilon > 0$ there exists an algorithm which finds a $(\frac{e}{e-1}\cdot p+\epsilon)$-approximation to the
\lp for the universal $\ell_p$-clustering problem.
\end{lemma}
\begin{proof}
The algorithm is to use the ellipsoid method with \textproc{$\ell_p$-SepOracle} as the separation oracle. We note that $f_{y, Y}(S)$ can be evaluated in polynomial time by computing optimal fractional knapsacks as discussed in the proof of Lemma~\ref{lemma:fracsubmodularity}. In addition, there are polynomially values of $Y$ that are iterated over, since $\lceil \log_{1+\epsilon'} c_{\max}/c_{\min}  \rceil$ is $O(p \log n/\epsilon')$ times the number of bits needed to describe the largest value of $c_{ij}$. So each call to \textproc{$\ell_p$-SepOracle} takes polynomial time.

By Lemma~\ref{lemma:maxatintegral} and the observation that $\fracc_p^p(\bd) - S_p^p(\bd) \leq p \cdot \fracc_p^{p-1}(\bd) (\fracc_p(\bd) - S_p(\bd))$ if $\fracc_p(\bd) - S_p(\bd) > 0$, we get that \textproc{$\ell_p$-SepOracle} cannot output a violated inequality if the input solution is feasible to Eq.~\eqref{eq:mainlp-ellp}. So we just need to show that if $\fracc_p(C') - S_p(C') \geq (\frac{e}{e-1}p + \epsilon) r$ for some $S, C'$, \textproc{$\ell_p$-SepOracle} outputs a violated inequality, i.e. does not output ``Feasible''. Let $Y'$ be the smallest value of $Y$ iterated over by \textproc{$\ell_p$-SepOracle} that is at least $\fracc_p^p(C')$, and $S'$ the $\frac{e-1}{e}$-maximizer of $f_{y, Y'}$ found by \textproc{$\ell_p$-SepOracle}. We have for the $\bd'$ found on Line~\ref{line:demand} of \textproc{$\ell_p$-SepOracle}:

$$\frac{1}{p(Y')^{1-1/p}}\left[ \fracc_p^p(\bd') - {S'}_p^p(\bd')\right] \geq \frac{e-1}{pe(Y')^{1-1/p}} \max_{S, \bd: \fracc_p^p(\bd) \leq Y'}[\fracc_p^p(\bd') - {S}_p^p(\bd')] \stackrel{(i)}{\geq} $$
$$\frac{e-1}{pe(Y')^{1-1/p}} \max_{S, \bd: \fracc_p^p(\bd) \leq \fracc_p^p(C')}[\fracc_p^p(\bd) - {S}_p^p(\bd)] \geq \frac{e-1}{pe(Y')^{1-1/p}}\left[\fracc_p^p(C') - S_p^p(C')\right] \geq$$
$$\frac{e-1}{pe(1+\epsilon')}\frac{\fracc_p^p(C') - S_p^p(C')}{\fracc_p^{p-1}(C')} = \frac{e-1}{pe(1+\epsilon')} \left[ \fracc_p(C') - S_p(C')\right]  > r.$$

In $(i)$, we use the fact that for a fixed $S$, $\max_{\bd: \fracc_p^p(\bd) \leq Y}[\fracc_p^p(\bd') - {S}_p^p(\bd')]$ is the solution to a fractional knapsack problem with weight $Y$, and that decreasing the weight allowed in a fractional knapsack instance can only reduce the optimum. For the last inequality to hold, we just need to choose $\epsilon' < \epsilon\frac{(e-1)}{pe}$ in \textproc{$\ell_p$-SepOracle} for the desired $\epsilon$. This shows that for $Y'$, \textproc{$\ell_p$-SepOracle} will output a violated inequality as desired.
\end{proof}

\subsection{Universal $\ell_p$-Clustering: Rounding Algorithm}

At a high level, we use the same strategy for rounding the fractional $\ell_p$-clustering solution as we did with $k$-median. Namely, we solve a discounted version of the problem where the discount for each client is equal to the (scaled) cost of the client in the fractional solution. However, if we apply this directly to the $\ell_p$ objective, we run into several problems. In particular, the linear discounts are incompatible with the non-linear objective defined over the clients. A more promising idea is to use these discounts on the $\ell_p^p$ objective, which in fact is defined as a linear combination over the individual client's objectives. But, for this to work, we will first need to relate the regret bound in the $\ell_p^p$ objective to that in the $\ell_p$ objective. This is clearly not true in general, i.e., for all realizations. However, we show that the realization that maximizes the regret of an algorithm $\alg$ against a fixed solution $\sol$ in both objectives is the same under the following ``farness'' condition: {\em for every client, either $\alg$'s connection is smaller than $\sol$'s or it is at least $p$ times as large as $\sol$'s}. Given any solution $\sol$, it is easy to define a \textit{virtual} solution $\widetilde{\sol}$ whose individual connection costs are bounded by $p$ times that in $\sol$, and $\widetilde{\sol}$ satisfies the farness condition. This allows us to relate the regret of $\alg$ against $\widetilde{\sol}$ (and thus against $p$ times $\sol$) in the $\ell_p^p$ objective to its regret in the $\ell_p$ objective.

We first state the technical lemma relating the $\ell_p^p$ and $\ell_p$ objectives under the farness condition. The proof of this lemma appears in Appendix~\ref{sec:proofs}. Informally, this lemma says that if we want to choose a realization maximizing the regret of $\alg$ against $\sol$ in (an approximation of) the $\ell_p$-objective, we should always include a client whose distance to $\alg$ exceeds their distance to $\sol$ by a factor larger than $p$. This contrasts (and limits) the example given at the beginning of this section, where we showed that including clients whose distance to $\alg$ exceeds the distance to $\sol$ by a smaller factor can actually reduce the regret of $\alg$ against $\sol$. In turn, if \textit{all} clients are closer to $\sol$ are closer by a factor of $p$, then the realization that maximizes regret in the $\ell_p$-objective is also the realization that maximizes regret in the $\ell_p^p$-objective. 

\begin{lemma}\label{lemma:gaprealization}
Suppose $\alg$ and $\sol$ are two 
solutions to an $\ell_p$-clustering instance, such that there is a subset of clients $C^*$ with the following property: for every client in $C^*$, the connection cost in $\alg$ is greater than $p$ times the connection cost in $\sol$, while for every client not in $C^*$, the connection cost in $\sol$ is at least the connection cost in $\alg$. Then, $C^*$ maximizes the following function:
\begin{equation*}
f(C') :=  
\begin{cases} 
\frac{\alg_p^p(C') - \sol_p^p(C')}{\alg_p^{p-1}(C')} & \alg_p^p(C') > 0 \\
0 & \alg_p^p(C') = 0
\end{cases}
\end{equation*}
\end{lemma}

Intuitively, this lemma connects the $\ell_p$ and $\ell_p^p$ objectives as this subset of clients $C^*$ will also be the set that maximizes the $\ell_p^p$ regret of $\alg$ vs $\sol$, and $f(C')$ is (within a factor of $p$) equal to the $\ell_p$ regret.
%
We use this lemma along with the $\ell_p^p$-clustering with discounts approximation in Lemma~\ref{lemma:kmbicriteria-lpp} to design the rounding algorithm for universal $\ell_p$-clustering.
As in the rounding algorithm for universal $k$-median, let $\sol$ denote a (virtual) solution whose connection costs are 3 times that of the fractional solution $\fracc$ for all clients. The rounding algorithm solves an $\ell_p^p$-clustering with discounts instance, where the discounts are $2$ times $\sol$'s connection costs. (Recall that in $k$-median, the discount was equal to $\sol$'s connection cost. Now, we need the additional factor of $2$ for technical reasons.) Let $\alg$ be the solution output by the algorithm of Lemma~\ref{lemma:kmbicriteria-lpp} for this problem. We prove the following bound for $\alg$ (the full proof of this lemma is given in Appendix~\ref{sec:proofs}; here, we present a sketch with the main ideas):

\begin{lemma}\label{lemma:lp-rounding}
There exists an algorithm which given any $(\alpha, \beta)$-approximate fractional solution $\fracc$ for $\ell_p$-clustering, outputs a  $(54p \alpha , 54p \beta + 18p^{1/p})$-approximate integral solution.
\end{lemma}\begin{proof}
Let $\sol$ denote a (virtual) solution whose connection costs are 3 times that of the fractional solution $\fracc$ for all clients. The rounding algorithm solves an $\ell_p^p$-clustering with discounts instance, where the discounts are $2$ times $\sol$'s connection costs. Let $\alg$ be the solution output by the algorithm of Lemma~\ref{lemma:kmbicriteria-lpp} for this problem. We also consider an additional virtual solution $\widetilde{\sol}$, whose connection costs are defined as follows: For clients $j$ such that $\alg$'s connection cost is greater than $18$ times $\sol$'s but less than $18p$ times $\sol$'s, we multiply  $\sol$'s connection costs by $p$ to obtain connection costs in $\widetilde{\sol}$. For all other clients, the connection cost in $\widetilde{\sol}$ is the same as that in $\sol$. Now, $\alg$ and $18 \cdot \widetilde{\sol}$ satisfy the condition in Lemma~\ref{lemma:gaprealization} and $18 \cdot \widetilde{\sol}$ is a $(54p\alpha, 54p\beta)$-approximation.

Our goal in the rest of the proof is to bound the regret of $\alg$ against (a constant times) $\widetilde{\sol}$ by (a constant times) the minimum regret $\mr$. 
Let us denote this regret: 
$$\regalgsol := \max_{C'\subseteq C}\left[\alg_p(C') - 18 \cdot \widetilde{\sol}_p(C')\right].$$ 
Note that if $\regalgsol = 0$ (it can't be negative), then for all realizations $C'$, $\alg_p(C') \leq 18 \cdot \widetilde{\sol}_p(C')$. In that case, the lemma follows immediately.
So, we  assume that $\regalgsol > 0$.

Let $C_1 = \argmax_{C'\subseteq C}\left[\alg_p(C') - 18 \cdot \widetilde{\sol}_p(C')\right]$, i.e., the realization defining $\regalgsol$ that maximizes the regret for the $\ell_p$ objective. We need to relate $\regalgsol$ to the regret in the $\ell_p^p$-objective for us to use the approximation guarantees of $\ell_p^p$-clustering with discounts from Lemma~\ref{lemma:kmbicriteria-lpp}. Lemma~\ref{lemma:gaprealization} gives us this relation, since it tells us that $C_1$ is exactly the set of clients for which $\alg$'s closest cluster center is at a distance of more than $18$ times that of $\widetilde{\sol}$'s closest cluster center. But, this means that $C_1$ also maximizes the regret for the $\ell_p^p$ objective, i.e., $C_1 = \argmax_{C'}\left[\alg^p_p(C') - 18^p \cdot \widetilde{\sol}^p_p(C')\right]$. Then, we have:
\begin{eqnarray*}
\regalgsol
& = & \frac{\max_{C'\subseteq C}\left[\alg_p^p(C') - 18^p\cdot {\widetilde{\sol}}_p^p(C')\right]}{\sum_{j=0}^{p-1} \alg_p^{j}(C_1) \cdot \left(18\cdot  {\widetilde{\sol}}_p(C_1)\right)^{p-1-j}} 
\quad \leq \frac{\max_{C'\subseteq C}\left[\alg_p^p(C') - 18^p\cdot {\widetilde{\sol}}_p^p(C')\right]}{\alg_p^{p-1}(C_1)} \\
& \leq & \frac{\max_{C'\subseteq C}\left[\alg_p^p(C') - 18^p\cdot {\sol}_p^p(C')\right]}{\alg_p^{p-1}(C_1)},
\text{ since connection costs in $\widetilde{\sol}$ are at least those in $\sol$}.
\end{eqnarray*}
Note that the numerator in this last expression is exactly the value of the objective for the $\ell_p^p$-clustering with discounts 
problem from Lemma~\ref{lemma:kmbicriteria-lpp}. Using this lemma, we can now bound the numerator by the optimum for this
problem, which in turn is bounded by the objective produced 
by the minimum regret solution $\mrs$ for the $\ell_p^p$-clustering with discounts instance:
\begin{equation}
\label{eq:kc}
\regalgsol
\leq \frac{\max_{C'\subseteq C}\left[\alg_p^p(C') - 18^p\cdot {\sol}_p^p(C')\right]}{\alg_p^{p-1}(C_1)}
\leq \frac{2}{3} \cdot 9^p  \cdot \frac{\max_{C'\subseteq C}\left[\mrs_p^p(C') - 2^p \cdot \sol_p^p(C')\right]}{\alg_p^{p-1}(C_1)}.
\end{equation}
First, we bound the numerator in the above expression.
Let $C_2 := \argmax_{C'\subseteq C}[\mrs_p^p(C') -$ $2^p \cdot \sol_p^p(C')]$
be the realization that maximizes this term. 
We now relate this term to $\mr$ (the first step is by factorization
and the second step holds because $2\cdot \sol = 6\cdot \fracc$ exceeds 
the optimal integer solution by to the upper bound of $3$ on the integrality gap~\cite{ArcherRS03}):
\eat{
\begin{eqnarray*}
& & \frac{2}{3} \cdot 9^p  \cdot \frac{\max_{C'}\left[\mrs_p^p(C') - 2^p \cdot \sol_p^p(C')\right]}{\alg_p^{p-1}(C_1)} \\
& = & \frac{2}{3} \cdot 9^p \cdot \frac{\mrs_p^p(C_2) - 2^p \cdot \sol_p^p(C_2)}{\sum_{j=0}^{p-1} \mrs_p^{j}(C_2) 2^{p-1-j} \cdot {\sol}_p^{p-1-j}(C_2)} \cdot \frac{\sum_{j=0}^{p-1} \mrs_p^{j}(C_2)2^{p-1-j}\cdot  {\sol}_p^{p-1-j}(C_2)}{\alg_p^{p-1}(C_1)} \\
& = & \frac{2}{3} \cdot 9^p \cdot\left(\mrs_p(C_2) - 2 \cdot \sol_p(C_2)\right) \cdot \frac{\sum_{j=0}^{p-1} \mrs_p^{j}(C_2) 2^{p-1-j}{\sol}_p^{p-1-j}(C_2)}{\alg_p^{p-1}(C_1)} \\
& \leq & \frac{2}{3} \cdot 9^p \cdot \mr \cdot \frac{\sum_{j=0}^{p-1} \mrs_p^{j}(C_2) 2^{p-1-j}{\sol}_p^{p-1-j}(C_2)}{\alg_p^{p-1}(C_1)}.
\end{eqnarray*}
}
\begin{eqnarray*}
\max_{C'}\left[\mrs_p^p(C') - 2^p \cdot \sol_p^p(C')\right] 
& = & \left(\mrs_p(C_2) - 2 \cdot \sol_p(C_2)\right) \cdot \sum_{j=0}^{p-1} \mrs_p^{j}(C_2) \cdot (2 \cdot {\sol}_p(C_2))^{p-1-j}\\ 
& \leq & \mr \cdot \sum_{j=0}^{p-1} \mrs_p^{j}(C_2) \cdot (2 \cdot {\sol}_p(C_2))^{p-1-j}.
\end{eqnarray*}
%
%
%
Using the above bound in Eq.~\eqref{eq:kc}, we get:
\begin{equation}
    \label{eq:kc2}
    \regalgsol \leq \frac{2}{3} \cdot 9^p \cdot \mr \cdot \frac{\sum_{j=0}^{p-1} \mrs_p^{j}(C_2) \cdot (2 \cdot {\sol}_p(C_2))^{p-1-j}}{\alg_p^{p-1}(C_1)}.    
\end{equation}
In the rest of proof, we obtain a bound on the last term in Eq.~\eqref{eq:kc2}. 
We consider two cases. If $\alg_p(C_1) \geq 9p^{\frac{1}{p}} \cdot \mrs_p(C_2)$ 
(intuitively, the denominator is large compared to the numerator), 
then:

$$\frac{\sum_{j=0}^{p-1} \mrs_p^{j}(C_2) 2^{p-1-j}{\sol}_p^{p-1-j}(C_2)}{\alg_p^{p-1}(C_1)} \leq p \cdot \frac{\mrs_p^{p-1}(C_2)}{\alg_p^{p-1}(C_1)} \leq p \cdot (9p^{\frac{1}{p}})^{-p+1} = 9^{-p+1}p^{1/p},$$

The first step uses the fact that $\mrs_p(C_2) \geq  2 \sol_p(C_2)$, so the largest term in the sum is $\mrs_p^{p-1}(C_2)$. Combined with Eq.~\eqref{eq:kc2} this $\regalgsol \leq 6p^{1/p} \cdot \mr$, giving the lemma statement. 

If $\alg_p(C_1) < 9p^{\frac{1}{p}} \cdot \mrs_p(C_2)$, then we cannot hope to meaningfully bound Eq~\eqref{eq:kc2}. In this case, however, $\regalgsol$ is also bounded by $\mrs_p(C_2)$, which we will eventually bound by $\mr$. More formally, by our assumption that $\regalgsol > 0$, $\mrs_p^p(C_2) - 2^p \cdot \sol_p^p(C_2) > 0$ and we have $2 \sol_p(C_2) \leq \mrs_p(C_2) \leq \sol_p(C_2) + \mr$. The first inequality is by our assumption that $\mrs_p^p(C_2) - 2^p \cdot \sol_p^p(C_2) > 0$, the second inequality is by definition of $\mrs, \mr$ and the fact that $\sol(C_2)$ upper bounds $\opt(C_2)$. In turn, $\sol_p(C_2) \leq \mr$ which gives $\mrs_p(C_2) \leq 2 \mr$ and thus $\regalgsol  \leq \alg_p(C_1) \leq 18p^{\frac{1}{p}} \cdot \mr$. We note that for all $p \geq 1$, $p^{1/p} \leq e^{1/e} \leq 1.5$. 
\end{proof}

We note that the final step requires using discounts equal to twice $\sol$'s connection costs instead of just $\sol$'s connection costs. If we did the latter, we would have started with the inequality $\sol_p(C_2) \leq \mrs_p(C_2) \leq \sol_p(C_2) + \mr$ instead, which does not give us any useful bound on $\sol(C_2)$ or $\mrs(C_2)$ in terms of just $\mr$. We also note that we chose not to optimize the constants in the final result of Lemma~\ref{lemma:lp-rounding} in favor of simplifying the presentation.

Theorem~\ref{thm:lpalg} now follows by using the values of $(\alpha, \beta)$ from Lemma~\ref{lemma:lp-fractional} (and a sufficiently small choice of the error parameter $\epsilon$) in the statement of Lemma~\ref{lemma:lp-rounding} above.
\section{Universal $k$-Center}\label{section:kcenteralg}

In the previous section, we gave universal algorithms for general $\ell_p$-clustering problems. 
Recall that the $k$-center objective, defined as the maximum distance of a client from its 
closest cluster center, can also be interpreted as the $\ell_{\infty}$-objective in the 
$\ell_p$-clustering framework. Moreover, it is well known that for any $n$-dimensional 
vector, its $\ell_{\log n}$ and $\ell_\infty$ norms differ only a constant factor
(see Fact~\ref{fact:plogn} in Appendix~\ref{sec:proofs}).
Therefore, choosing $p = \log n$ in Theorem~\ref{thm:lpp} gives 
poly-logarithmic approximation bounds for the universal $k$-center problem. 
In this section, we give direct techniques that improve these bounds 
to constants:

\begin{theorem}\label{thm:kcenteralg}
There exists a $(3, 3)$-approximate algorithm for the universal $k$-center problem.
\end{theorem}

First, note that for every client $j$, its distance to the closest cluster center 
in the minimum regret solution $\mrs$ is at most $\mr_j := \min_{i \in F} c_{ij} + \mr$; 
otherwise, in the realization with only client $j$, $\mrs$ would have regret $> \mr$. 
We first design an algorithm $\alg$ that $3$-approximates these distances $\mr_j$, i.e., 
for every client $j$, its distance to the closest cluster center in $\alg$ is at 
most $3 \mr_j$. Indeed, this algorithm satisfies a more general property: given 
{\em any} value $r$, it produces a set of cluster centers such that every 
client $j$ is at a distance $\leq 3 r_j$ from its closest cluster center, 
where $r_j := \min_{i \in F} c_{ij} + r$. Moreover, if $r\geq \mr$, then the 
number of cluster centers selected by $\alg$ is at most $k$ (for smaller 
values of $r$, $\alg$ might select more than $k$ cluster centers).

Our algorithm $\alg$ is a natural greedy algorithm. We order clients $j$ 
in increasing order of $r_j$, and if a client $j$ does not have a cluster center 
within distance $3 r_j$ in the current solution, then we add its closest cluster 
center in $F$ to the solution. 
\begin{lemma}
Given a value $r$, the greedy algorithm $\alg$ selects cluster centers
that satisfy the following properties:
\begin{itemize}
\label{lemma:kcenter}
    \item every client $j$ is within a distance of 
    $3r_j = 3(\min_{i \in F} c_{ij} + r)$ from their closest cluster center.
    \item If $r\geq \mr$, then $\alg$ does not select more than $k$ cluster centers,
    i.e., the solution produced by $\alg$ is feasible for the $k$-center problem.
\end{itemize}
\end{lemma}
\begin{proof}
    The first property follows from the definition of $\alg$.

    To show that $\alg$ does not pick more than $k$ cluster centers, we map
    the cluster center $i$ added in $\alg$ by some client $j$ to its closest
    cluster center $i'$ in $\mrs$. Now, we claim that no two cluster centers
    $i_1, i_2$ in $\alg$ can be mapped to the same cluster center $i'$ in 
    $\mrs$. Clearly, this proves the lemma since $\mrs$ has only $k$ cluster 
    centers.
    
    \begin{figure}[ht]
        \centering
        \includegraphics[width=0.5\textwidth]{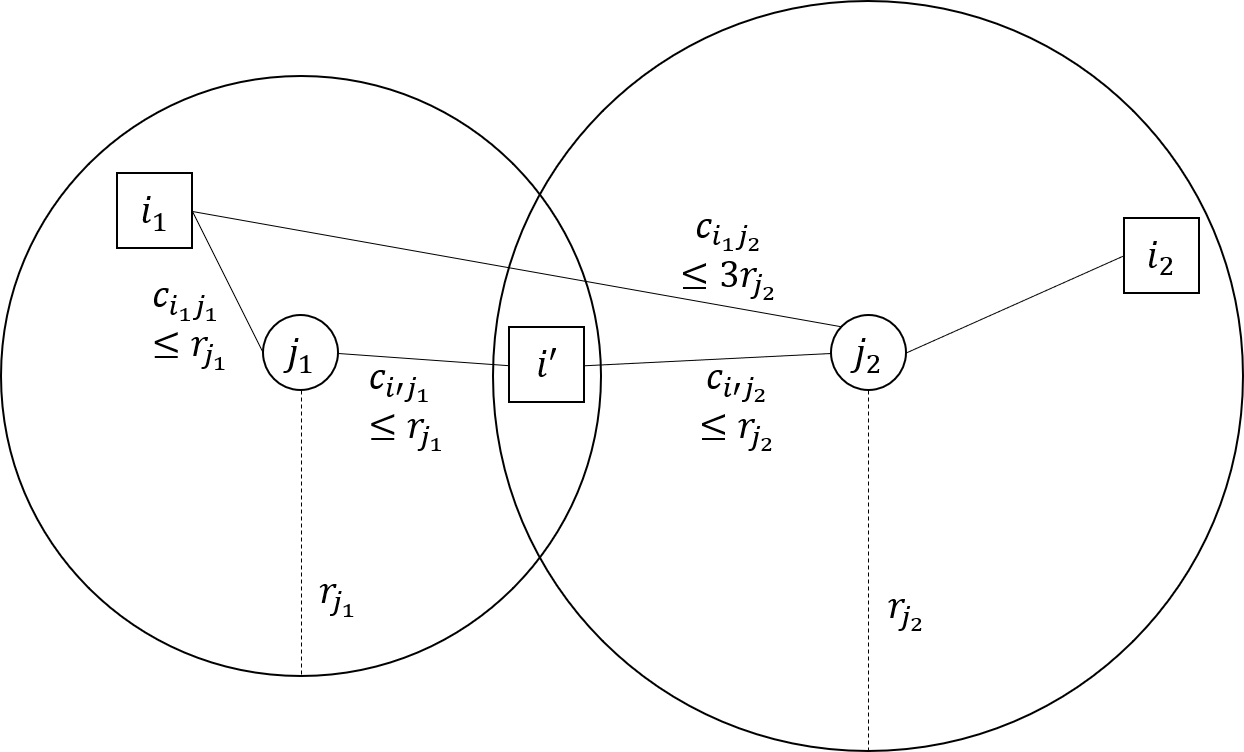}
        \caption{Two clients $j_1, j_2$ that are distance at most $r_{j_1}, r_{j_2}$ 
        respectively from the same cluster center $i'$ in $\mrs$ cannot cause $\alg$ 
        to add two different cluster centers $i_1, i_2$.}
        \label{fig:greedy}
    \end{figure}
    
    Suppose $i_1, i_2$ are two cluster centers in $\alg$ mapped to the same
    cluster center $i'$ in $\mrs$. Assume without loss of generality
    that $i_1$ was added to $\alg$ before $i_2$. Let $j_1, j_2$ be the clients 
    that caused $i_1, i_2$ to be added; since $i_2$ was added later, we have 
    $r_{j_1}\leq r_{j_2}$. 
    The distance from $j_2$ to $i_1$ is at most the length of the path 
    $(j_2, i', j_1, i_1)$ (see Fig.~\ref{fig:greedy}), which is at most
    $2r_{j_2} + r_{j_1} \leq 3r_{j_2}$. 
    But, in this case $j_2$ would not have added a new cluster center $i_2$,
    thus arriving at a contradiction.
\end{proof}

We now use the above lemma to prove Theorem~\ref{thm:kcenteralg}.

\begin{proof}[Proof of Theorem~\ref{thm:kcenteralg}]
In any realization $C'\subseteq C$, the optimal value of the $k$-center objective is $\opt(C') \geq \max_{j \in C'} \min_{i \in F} c_{ij}$, whereas the solution produced by the algorithm $\alg$ given above has objective value at most $3 (\max_{j \in C'} \min_{i \in F} c_{ij} + r)$. So, $\alg$'s solution costs at most $3\cdot \opt(C') + 3\cdot r$ for all realizations $C'\subseteq C$. So, if we were able to choose $r = \mr$, we would prove the theorem. But, we do not know the value of $\mr$ (in fact, this is \np-hard). Instead, we increase the value of $r$ continuously until $\alg$ produces a solution with at most $k$ clients. By Lemma~\ref{lemma:kcenter}, we are guaranteed that this will happen for some $r\leq \mr$, which then proves the theorem. 

Our final observation is that this algorithm can be implemented in polynomial time, since there are only polynomially many possibilities for the $k$-center objective across all realizations (namely, the set of all cluster center to client distances) and thus polynomially many possible values for $\mr$ (the set of all differences between all possible solution costs). So, we only need to run $\alg$ for these values of $r$ in increasing order.
\end{proof}

We note that the greedy algorithm described above can be viewed as an extension of the $k$-center algorithm in \cite{HochbaumS86} to a $(3, 3)$-approximation for the ``$k$-center with discounts'' problem, where the discounts are the minimum distances $\min_{i \in F} c_{ij}$.
\section{Universal $k$-Median with Fixed Clients}\label{section:kmalg2}

In this section, we extend the techniques from Section~\ref{section:kmalg} to prove the following theorem:

\begin{theorem}\label{thm:fixedkmalgorithm}
If there exists a deterministic polynomial time $\gamma$-approximation algorithm for the $k$-median problem, then for every $\epsilon > 0$ there exists a $(54\gamma+\epsilon, 60)$-approximate universal algorithm for the universal $k$-median problem with fixed clients. 
\end{theorem}

By using the derandomized version of the $(2.732 + \epsilon)$-approximation algorithm of Li and Svensson~\cite{LiS13} for the $k$-median problem, and appropriate choice of both $\epsilon$ parameters, we obtain the following corollary from Theorem~\ref{thm:fixedkmalgorithm}. 

\begin{corollary}\label{cor:kmalgorithm}
For every $\epsilon > 0$, there exists a $(148 + \epsilon, 60)$-approximate universal algorithm for the $k$-median problem with fixed clients. 
\end{corollary}

Our high level strategy comprises two steps. In Section~\ref{subsection:designingseparationoracle}, we show how to find a good fractional solution by approximately solving a linear program. In Section~\ref{subsection:kmedian-rounding}, we then round the fractional solution in a manner that preserves its regret guarantee within constant factors. As discussed in Section~\ref{section:defs}, for simplicity our algorithm's description and analysis will avoid the notion of demands and instead equivalently view the input as specifying a set of fixed and unfixed clients, of which multiple might exist at the same location.

\subsection{Preliminaries}
\label{sec:kmedian-prelim}
In addition to the preliminaries of Section~\ref{section:kmalg}, we will use the following tools:

\smallskip
\noindent
{\bf Submodular Maximization over Independence Systems.}
An \textit{independence system} comprises a ground set $E$ and a set of subsets (called {\em independent sets}) $\mathcal{I} \subseteq 2^E$ with the property that if $A \subseteq B$ and $B \in \mathcal{I}$ then $A \in \mathcal{I}$ (the \textit{subset closed} property). An independent set $S$ in $\mathcal{I}$ is \textit{maximal} if there does not exist $S' \supset S$ such that $S' \in \mathcal{I}$. Note that one can define an independence system by specifying the set of maximal independent sets $\mathcal{I}'$ only, since the subset closed property implies $\mathcal{I}$ is simply all subsets of sets in $\mathcal{I}'$. An independence system is a \textit{$1$-independence system} (or \textit{$1$-system} in short) if all maximal independent sets are of the same size. The following result on maximizing submodular functions over 1-independence systems follows from a more general result given implicitly in \cite{FisherNW78} and more formally in \cite{CalinescuCPV11}.

\begin{theorem}\label{thm:independencesystems}
There exists a polynomial time algorithm that given a $1$-independence system $(E, \mathcal{I})$ and a non-negative monotone submodular function $f:2^E \rightarrow \mathbb{R}^+$ defined over it, finds a $\frac{1}{2}$-maximizer of $f$, i.e. finds $S' \in \mathcal{I}$ such that $f(S') \geq \frac{1}{2} \max_{S \in \mathcal{I}} f(S)$.
\end{theorem}
The algorithm in the above theorem is the natural greedy algorithm, which starts with $S' = \emptyset$ and repeatedly adds to $S'$ the element $u$ that maximizes $f(S' \cup \{u\})$ while maintaining that $S' \cup \{u\}$ is in $\mathcal{I}$, until no such addition is possible. 

\smallskip
\noindent
{\bf Incremental $\ell_p$-Clustering.}
We will also use the \textit{incremental $\ell_p$-clustering} problem  which is defined as follows: Given an $\ell_p$-clustering instance and a subset of the cluster centers $S$ (the ``existing'' cluster centers), find the minimum cost solution to the $\ell_p$-clustering instance with the additional constraint that the solution must contain all cluster centers in $S$. When $S = \emptyset$, this is just the standard $\ell_p$-clustering problem, and this problem is equivalent to the standard $\ell_p$-clustering problem by the following lemma:

\begin{lemma}\label{lemma:fixedfacilities}
If there exists a $\gamma$-approximation algorithm for the $\ell_p$-clustering problem, there exists a $\gamma$-approximation for the incremental  $\ell_p$-clustering problem.
\end{lemma} 

\begin{proof}[Proof of Lemma~\ref{lemma:fixedfacilities}]

The $\gamma$-approximation for incremental $\ell_p$-clustering is as follows: Given an instance $I$ of incremental $\ell_p$-clustering with clients $C$ and existing cluster centers $S$, create a $\ell_p$-clustering instance $I'$ which has the same cluster centers and clients as the $\ell_p$-clustering instance except that at the location of every cluster center in $S$, we add a client with demand $\gamma|C|^{1/p}\max_{i,j}c_{ij}+1$. 

Let $T^*$ be the solution that is a superset of $S$ of size $k$ that achieves the lowest cost of all such supersets in instance $I$. Let $T$ be the output of running a $\gamma$-approximation algorithm for $\ell_p$-clustering on $I'$. Then we wish to show $T$ is a superset of $S$ and has cost at most $\gamma$ times the cost of $T^*$ in instance $I$.

Any solution that buys all cluster centers in $S$ has the same cost in $I$ and $I'$. Then we claim it suffices to show that $T$ is a superset of $S$. If $T$ is a superset of $S$, then since both $T$ and $T^*$ are supersets of $S$ and since $T$ is a $\gamma$-approximation in instance $I'$, its cost in $I'$ is at most $\gamma$ times the cost of $T^*$ in $I'$. This in turn implies $T$ has cost at most $\gamma$ times the cost of $T^*$ in $I$, giving the Lemma. 

Assume without loss of generality no two cluster centers are distance 0 away from each other. To show that $T$ is a superset of $S$, note that in instance $I'$ any solution that does not buy a superset of $S$ is thus at least distance $1$ from the location of some cluster center in $S$ and thus pays cost at least $\gamma|C|^{1/p}\max_{i,j}c_{ij}+1$ due to one of the added clients. On the other hand, any solution that is a superset of $S$ is distance $0$ from all the added clients and thus only has to pay connection costs on clients in $C$, which in turn means it has cost at most $|C|^{1/p}\max_{i,j}c_{ij}$. Since $T$ is the output of a $\gamma$-approximation algorithm, $T$ thus has cost at most $\gamma|C|^{1/p}\max_{i,j}c_{ij}$, which means $T$ must be a superset of $S$.
\end{proof}

\subsection{Obtaining a Fractional Solution for Universal $k$-Median with Fixed Clients}\label{subsection:designingseparationoracle}

Let $C_f\subseteq C$ denote the set of fixed clients and for any realization of clients $C'$ satisfying $C_f\subseteq C'\subseteq C$, let $\opt(C')$ denote the cost of the optimal solution for $C'$. The universal $k$-median LP is given by:
\begin{eqnarray}
\nonumber & &\min \ r \qquad \text{($r$ denotes maximum regret across all demand realizations)}\\
\nonumber \textnormal{s.t.}& &\sum_{i\in F} x_i \leq k \qquad \text{($x_i$ is the indicator variable for opening cluster center $i$)}\\
\nonumber &\forall i\in F, j\in C: &y_{ij} \leq x_i \qquad \text{($y_{ij}$ is the indicator variable for cluster center $i$ serving client $j$ if it realizes)}\\
\nonumber &\forall j\in C: &\sum_{i\in F} y_{ij} \geq 1\\
&\forall C_f \subseteq C' \subseteq C: &\sum_{j \in C'} \sum_{i\in F} c_{ij} y_{ij} - \opt(C') \leq r \label{eq:regret-constraint}\\ 
\nonumber &\forall i\in F, j\in C: &x_i, y_{ij} \in [0, 1]
\end{eqnarray}
Note that Eq.~\eqref{eq:regret-constraint} and the objective function distinguish this LP from the standard $k$-median LP. We call Eq.~\eqref{eq:regret-constraint} the \textit{regret constraint set}. For a fixed fractional solution $x, y, r$, our goal is to approximately separate the regret constraint set, since all other constraints can be separated exactly. In the rest of this subsection, we describe our approximate separation oracle and give its analysis.

Let $S(C')$ denote the cost of the solution $S\subseteq F$ in realization $C'$ (that is, $S(C') = \sum_{j \in C'} \min_{i \in S} c_{ij}$). Since $\opt(C') = \min_{S: S\subseteq F, |S| = k} S(C')$, separating the regret constraint set exactly is equivalent to deciding if the following holds:
\begin{equation}\label{eq:regretconstraintset}
\forall S: S \subseteq F, |S| = k: \quad \max_{C': C_f \subseteq C' \subseteq C} \left[ \sum_{j \in C'} \sum_{i\in F} c_{ij} y_{ij} - S(C')\right] \leq r.
\end{equation}
By splitting the terms $\sum_{j \in C'} \sum_{i\in F} c_{ij} y_{ij}$ and $S(C')$ into terms for $C_f$ and $C' \setminus C_f$, we can rewrite Eq.~\eqref{eq:regretconstraintset} as follows:

\begin{equation*}
\begin{aligned}
\max_{C_f \subseteq C' \subseteq C, S \subseteq F, |S| = k}& \sum_{j \in C'} \sum_{i \in F} c_{ij} y_{ij} - S(C') \leq r\\
\forall S \subseteq F, |S| = k: \max_{C_f \subseteq C' \subseteq C}& \sum_{j \in C'} \left[\sum_{i \in F} c_{ij} y_{ij} - S(C')\right] \leq r\\
\forall S \subseteq F, |S| = k: \max_{C_f \subseteq C' \subseteq C}& \sum_{j \in C' \setminus C_f} \left[\sum_{i \in F} c_{ij} y_{ij} + \sum_{j \in C_f} \sum_{i \in F} c_{ij} y_{ij} - S(C' \setminus C_f) - S(C_f)\right] \leq r\\
\forall S \subseteq F, |S| = k: \max_{C_f \subseteq C' \subseteq C}&  \left[\sum_{j \in C' \setminus C_f} \sum_{i \in F} c_{ij} y_{ij} - S(C' \setminus C_f) \right]\leq S(C_f) - \sum_{j \in C_f} \sum_{i \in F} c_{ij} y_{ij} + r\\
\forall S \subseteq F, |S| = k: \max_{C^* \subseteq C \setminus C_f}& \left[\sum_{j \in C^*} \sum_{i \in F} c_{ij} y_{ij} - S(C^*)\right] \leq S(C_f) - \sum_{j \in C_f} \sum_{i \in F} c_{ij} y_{ij} + r\\
\end{aligned}
\end{equation*}
For fractional solution $y$, let 
\begin{equation}\label{eq:defn-f}
f_y(S) = \max_{C^*: C^* \subseteq C \setminus C_f} \left[\sum_{j \in C^*} 
\sum_{i \in F} c_{ij} y_{ij} - S(C^*)\right].
\end{equation}
Note that we can compute $f_y(S)$ for any $S$ easily since the maximizing value of $C^*$ is the set of clients $j$ for which $S$ has connection cost less than $\sum_{i \in F} c_{ij}y_{ij}$. We already know $f_y(S)$ is not submodular. But, the term $S(C_f)$ is not fixed with respect to $S$, so maximizing $f_y(S)$ is not enough to separate Eq.~\eqref{eq:regret-constraint}. To overcome this difficulty, for every possible cost $M$ on the fixed clients, we replace $S(C_f)$ with $M$ and only maximize over solutions $S$ for which $S(C_f) \leq M$ (for convenience, we will call any solution $S$ for which $S(C_f) \leq M$ an \textit{$M$-cheap} solution):
%
%
%
\begin{equation}
\forall M \in \left\{0, 1, \ldots, |C_f|\max_{i,j}c_{ij}\right\}:
\max_{S: S \subseteq F, |S| = k, S(C_f) \leq M}  f_y(S) \leq M - \sum_{j \in C_f} \sum_{i\in F} c_{ij} y_{ij} + r.
\end{equation}
Note that this set of inequalities is equivalent to Eq.~\eqref{eq:regret-constraint}, but it has the advantage that the left-hand side is approximately maximizable and the right-hand side is fixed. Hence, these inequalities can be approximately separated. However, there are exponentially many inequalities; so, for any fixed $\epsilon > 0$, we relax to the following polynomially large set of inequalities:
%
%
\begin{equation}
\label{eq:separation}
\forall M \in \left\{0, 1, 1+\epsilon, \ldots, (1+\epsilon)^{\lceil\log_{1+\epsilon}(|C_f|\max_{i,j}c_{ij})\rceil + 1}\right\}:
\max_{S: S \subseteq F, |S| = k, S(C_f) \leq M}  f_y(S) \leq M - \sum_{j \in C_f} \sum_{i\in F} c_{ij} y_{ij} + r.
\end{equation}


Separating inequality Eq.~\eqref{eq:separation} for a fixed $M$ corresponds to submodular maximization of $f_y(S)$, but now subject to the constraints $|S| = k$ and $S(C_f) \leq M$ as opposed to just $|S| = k$. Let $\mathcal{S}_M$ be the set of all $S\subseteq F$ such that $|S| = k$ and $S(C_f) \leq M$. Since $f_y(S)$ is monotone, maximizing $f_y(S)$ over $\mathcal{S}_M$ is equivalent to maximizing $f_y(S)$ over the independence system $(F, \mathcal{I}_M)$ with maximal independent sets $\mathcal{S}_M$. 

Then all that is needed to approximately separate Eq.~\eqref{eq:separation} corresponding to a fixed $M$ is an oracle for deciding membership in $(F, \mathcal{I}_M)$. Recall that $S\subseteq F$ is in $(F, \mathcal{I}_M)$ if there exists a set $S'\supseteq S$ such that $|S'| = k$ and $S'(C_f) \leq M$. But, even deciding membership of the empty set in $(F, \mathcal{I}_M)$ requires one to solve a $k$-median instance on the fixed clients, which is in general NP-hard. More generally, we are required to solve an instance of the incremental $k$-median problem (see Section~\ref{sec:kmedian-prelim}) with existing cluster centers in $S$. 

While exactly solving incremental $k$-median is NP-hard, we have a constant approximation algorithm for it (call it $A$), by Lemma~\ref{lemma:fixedfacilities}. So, we could define a new system $(F, \mathcal{I'}_M)$ that contains a set $S\subseteq F$ if the output of $A$ for the incremental $k$-median instance with existing cluster centers $S$ has cost at most $M$. But, $(F, \mathcal{I'}_M)$ may no longer be a $1$-system, or even an independence system. To restore the subset closed property, the membership oracle needs to ensure that: (a) if a subset $S'\subseteq S$ is determined to not be in $(F, \mathcal{I'}_M)$, then $S$ is not either, and (b) if a superset $S'\supseteq S$ is determined to be in $(F, \mathcal{I'}_M)$, then so is $S$. \eat{Checking these properties explicitly would imply exponential time to decide the membership of each set. But, we observe that the membership oracle for $(F, \mathcal{I'}_M)$ is called by the greedy submodular maximization algorithm from Theorem~\ref{thm:independencesystems}, which implies additional structure over the set of invocations of this oracle. In particular, in any step of the greedy algorithm, there is a candidate set $S$ that is in $(F, \mathcal{I'}_M)$ and the membership oracle is called for sets $S \cup \{i\}$ for elements $i\in F$. We first check if there was a previous invocation of the oracle where the approximation algorithm $A$ returned an $M$-cheap $k$-median solution containing cluster centers $S \cup \{i\}$; if so, the oracle says that $S\cup \{i\}$ is in $(F, \mathcal{I'}_M)$. This ensures that property (b) above is not violated. To ensure property (a), if the oracle outputs that $S\cup \{i\}$ is not in $(F, \mathcal{I'}_M)$ at any stage, then the algorithm removes the cluster center $i$ entirely from future consideration.}

\begin{figure}[ht]
\fbox{\begin{minipage}{\textwidth}
\textproc{GreedyMax}($(x, y, r)$, $F$, $C_f$, $C$, $M$, $T_0$, $f$):\newline
\textbf{Input:} Fractional solution $(x,y,r)$, set of cluster centers $F$, set of fixed clients $C_f$, set of all clients $C$, value $M$, $M$-cheap solution $T_0$, submodular objective $f:2^F \rightarrow \mathbb{R}^+$
\begin{algorithmic}[1]
\State $S_0 \leftarrow \emptyset$
\State $F_0 \leftarrow F$
\For{$l$ from $1$ to $k$}
\For{Each cluster center $i$ in $F_{l-1} \setminus S_{l-1}$}
\If{$S_{l-1} \cup \{i\} \subseteq T_0$ and $T_0$ is $M$-cheap}
\State $T_{l, i} \leftarrow T_0$
\Else \If{For some $l', i'$, $S_{l-1} \cup \{i\} \subseteq T_{l', i'}$ and $T_{l', i'}$ is $M$-cheap}
\State $T_{l, i} \leftarrow T_{l', i'}$ 
\Else
\State $T_{l, i} \leftarrow$ Output of $\gamma$-approximation algorithm on incremental $\ell_p$-clustering instance with\par\hskip2cm cluster centers $F_{l-1}$, existing cluster centers $S_{l-1} \cup \{i\}$ and clients $C_f$.
\EndIf
\EndIf
\EndFor 
\State $F_l \leftarrow F_{l-1}$
\For{Each cluster center $i$ in $F_l \setminus S_{l-1}$}
\If{$i$ does not appear in any $T_{l, i'}$ that is $M$-cheap}
\State $F_l \leftarrow F_l \setminus \{i\}$
\EndIf
\EndFor
\State $S_l \leftarrow S_{l-1} \cup \{\argmax_{i \in F_l \setminus S_{l-1}} f(S_{l-1} \cup \{i\}) \}$
\EndFor
\State \Return $S_k$
\end{algorithmic}
\end{minipage}}
\caption{Modified Greedy Submodular Maximization Algorithm}
\label{fig:submaxalg}
\end{figure}

We now give the modified greedy maximization algorithm \textsc{GreedyMax} that we use to try to separate one of the inequalities in 
Eq. \eqref{eq:separation}, which uses a built-in membership oracle that ensures the above properties hold. Pseudocode is given in Figure~\ref{fig:submaxalg}, and we informally describe it here. \textsc{GreedyMax} initializes $S_0 = \emptyset$, $F_0 = F$, and starts with a $M$-cheap $k$-median solution $T_0$ (generated by running a $\gamma$-approximation on the $k$-median instance involving only fixed clients $C_f$). In iteration $l$, \textsc{GreedyMax} starts with a partial solution $S_{l-1}$ with $l-1$ cluster centers, and it is considering adding cluster centers in $F_{l-1}$ to $S_{l-1}$. For each cluster center $i$ in $F_{l-1}$, \textsc{GreedyMax} generates some $k$-median solution $T_{l,i}$ containing $S_{l-1} \cup \{i\}$ to determine if $S_{l-1} \cup \{i\}$ is in the independence system. If a previously generated solution, $T_0$ or $T_{l', i'}$ for any $l', i'$, contains $S_{l-1} \cup \{i\}$ and is $M$-cheap, then $T_{l,i}$ is set to this solution. Otherwise, \textsc{GreedyMax} runs the incremental $k$-median approximation algorithm on the instance with existing cluster centers in $S_{l-1} \cup \{i\}$, the only cluster centers in the instance are $F_{l-1}$, and the client set is $C_f$. It sets $T_{l,i}$ to the solution generated by the approximation algorithm.

After generating the set of solutions $\{T_{l,i}\}_{i \in F_{l-1}}$, if one of these solutions contains $S_{l-1} \cup \{i\}$ and is $M$-cheap, then \textsc{GreedyMax} concludes that $S_{l-1} \cup \{i\}$ is in the independence system. This, combined with the fact that these solutions may be copied from previous iterations ensures property (b) holds (as the $M$-cheap solutions generated by \textsc{GreedyMax} are implicitly considered to be in the independence system). Otherwise, since \textsc{GreedyMax} was unable to find an $M$-cheap superset of $S_{l-1} \cup \{i\}$, it considers $S_{l-1} \cup \{i\}$ to not be in the independence system. In accordance with these beliefs, \textsc{GreedyMax} initializes $F_l$ as a copy of $F_{l-1}$, and then removes any $i$ such that it did not find an $M$-cheap superset of $S_{l-1} \cup \{i\}$ from $F_l$ and thus from future consideration, ensuring property (a) holds. It then greedily adds to $S_{l-1}$ the $i$ in $F_l$ that maximizes $f_y(S_{l-1} \cup \{i\})$ as defined before to create a new partial solution $S_l$. After the $k$th iteration, \textsc{GreedyMax} outputs the solution $S_k$. 

\begin{figure}[ht]
\fbox{\begin{minipage}{\textwidth}
\textproc{SepOracle}($(x, y, r)$, $F$, $C_f$, $C$):\newline
\textbf{Input:} A fractional solution $x, y, r$, set of cluster centers $F$, set of fixed clients $C_f$, set of all clients $C$
\begin{algorithmic}[1]
\If{Any constraint in the universal $k$-median LP except the regret constraint set is violated} \Return the violated constraint
\EndIf
\State $T_0 \leftarrow$ output of $\gamma$-approximation algorithm $A$ for $k$-median run on instance with cluster centers $F$, clients $C_f$ 
\For{$M \in \{0, 1, 1+\epsilon, (1+\epsilon)^2, \ldots (1+\epsilon)^{\lceil\log_{1+\epsilon}(\gamma |C_f|\max_{i,j}c_{ij})\rceil + 1}\}$ such that $T_0$ is $M$-cheap} 
\State $S' \leftarrow$ \textproc{GreedyMax}($(x, y, r)$, $F$, $C_f$, $C$, $M$, $T_0$, $f_y$)
\State $C' \leftarrow \argmax_{C^* \subseteq C \setminus C_f} [\sum_{j \in C^*} \sum_{i \in F} c_{ij} y_{ij} - S'(C^*)]$
\If{$\sum_{j \in C'} \sum_{i \in F} c_{ij} y_{ij} - S(C') > M - \sum_{j \in C_f} \sum_{i \in F} c_{ij} y_{ij} + r$}
\State\hskip\algorithmicindent\Return $\sum_{j \in C'} \sum_{i \in F} c_{ij} y_{ij} - S(C') \leq M - \sum_{j \in C_f} \sum_{i \in F} c_{ij} y_{ij} + r$
\EndIf
\EndFor
\State \Return ``Feasible''
\end{algorithmic}
\end{minipage}}
\caption{Approximate Separation Oracle for Universal $k$-Median}
\label{fig:separationoracle}
\end{figure}

Our approximate separation oracle, \textsc{SepOracle}, can then use \textsc{GreedyMax} as a subroutine. Pseudocode is given in Figure~\ref{fig:separationoracle}, and we give an informal description of the algorithm here. \textsc{SepOracle} checks all constraints except the regret constraint set, and then outputs any violated constraints it finds. If none are found, it then runs a $k$-median approximation algorithm on the instance containing only the fixed clients to generate a solution $T_0$. For each $M$ that is 0 or a power of $1+\epsilon$ (as in Eq.~\eqref{eq:separation}), if $T_0$ is $M$-cheap, it then invokes \textsc{GreedyMax} for this value of $M$ (otherwise, \textsc{GreedyMax} will consider the corresponding independence system to be empty, so there is no point in running it), passing $T_0$ to \textsc{GreedyMax}. If then checks the inequality $\sum_{j \in C'} \sum_i c_{ij} y_{ij} - S(C') \leq M - \sum_{j \in C_f} \sum_i c_{ij} y_{ij} + r$ for the solution $S$ outputted by \textsc{GreedyMax}, and outputs this inequality if it is violated.

This completes the intuition behind and description of the separation oracle. We now move on to its analysis. First, we show that \textsc{GreedyMax} always finds a valid solution.

\begin{lemma}\label{lemma:nonempty}
\textsc{GreedyMax} always outputs a set $S_k$ of size $k$ when called by \textsc{SepOracle}.
\end{lemma}

\begin{proof}
Note that \textsc{GreedyMax} is only invoked if $T_0$ is $M$-cheap. This implies some $T_{1,i}$ is $M$-cheap since some $T_{1,i}$ will be initialized to $T_0$. Then, it suffices to show that in the $l$th iteration, there is some $i$ that can be added to $S_{l}$. If this is true, it implies $S_k$ is of size $k$ since $k$ elements are added across all $k$ iterations.

This is true in iteration $1$ because some $T_{1, i}$ is $M$-cheap and thus any element of $T_{1, i}$ is in $F_1$ and can be added. Assume  this is inductively true in iteration $l$, i.e. $i$ is added to $S_l$ in iteration $l$ because $i$ is in some $M$-cheap $T_{l, i'}$. Since $T_{l, i'}$ is $M$-cheap, no element of $T_{l, i'}$ is deleted from $F_l$. Then in iteration $l+1$, for all $i''$ in $T_{l, i'} \setminus S_l$ (a set of size $k-l$, i.e. non-empty), $T_{l+1, i''}$ can be initialized to $T_{l, i'}$. Then all such $i''$ can be added to $S_{l+1}$ because all such $i''$ satisfy that $T_{l+1, i''}$ is $M$-cheap and thus are in $F_{l+1}$. By induction, in all iterations, there is some $i$ that can be added to $S_l$, giving the Lemma.
\end{proof}

Then, the following lemma asserts that \textsc{GreedyMax} is indeed performing greedy submodular maximization over some 1-system.
%
%
\begin{lemma}\label{lemma:1system}
Fix any run of \textsc{GreedyMax}. Consider the values of $S_l, T_{l, i}, F_l$ for all $l, i$ (defined as in Figure \ref{fig:submaxalg}) at the end of this run. Let $\mathcal{B}$ be the set containing $S_{l-1} \cup \{i\}$ for each $l$, $i \notin F_l, i \notin S_{l-1}$. Let $(F, \mathcal{S})$ be the independence system for which the set of maximal independent sets $\mathcal{S}_{max}$ consists of all size $k$ subsets $S$ of $F$ such that no subset of $S$ is in $\mathcal{B}$, and $S$ is $M$-cheap. Then the following properties hold:
\begin{enumerate}\vspace{-5pt}
\item For any $l$ and any $i \notin F_l$, $S_{l-1} \cup \{i\}$ is not in $\mathcal{S}$
\vspace{-5pt}
\item For any $l$ and any $i \in F_l$, $S_{l-1} \cup \{i\}$ is in $\mathcal{S}$
\vspace{-5pt}
\item $(F, \mathcal{S})$ is a 1-system
\end{enumerate}
\end{lemma}

\begin{proof}
\textbf{Property 1:} This property is immediate from the definition of $\mathcal{B}$ and $\mathcal{S}$.

\textbf{Property 2:} Fix any $l, i \in F_l$. We want to show that $S_{l-1} \cup \{i\}$ is in $\mathcal{S}$. Since $i \in F_l$, there exists some $T_{l, i'}$ such that $S_{l-1} \cup \{i\}$ is a subset of $T_{l, i'}$ and $T_{l, i'}$ is $M$-cheap (otherwise, $i$ would have been deleted from $F_l$). If we can show $T_{l, i'}$ is in $\mathcal{S}_{max}$ then we immediately get that $S_{l-1} \cup \{i\}$ is in $\mathcal{S}$. 

Suppose not. Since $T_{l, i'}$ is $M$-cheap, this must be because some subset of $T_{l, i'}$ is of the form $S_{l'-1} \cup \{i''\}$ for $i'' \notin F_{l'}, i'' \notin S_{l'-1}$. In particular, consider the smallest value $l'$ for which this is true, i.e. let $l'$ be the iteration in which $i''$ was deleted from $F_{l'}$. 

If $l' < l$, since $i''$ was deleted from $F_{l'}$, $i''$ cannot appear in any $M$-cheap solution containing $S_{l'-1}$ generated by the incremental $k$-median approximation algorithm before the end of iteration $l'$ (otherwise, $T_{l', i''}$ could be initialized to this solution, preventing $i''$ from being deleted). Since $i''$ is not in $F_{l'}$ (and thus not in $F_{l'+1} \ldots F_l$), in iterations $l'+1$ to $l$ the approximation algorithm is not allowed to use $i''$. So no $M$-cheap solution is ever generated by the approximation algorithm which is a superset of $S_{l'-1} \cup \{i''\}$. But $T_{l, i'}$ is a $M$-cheap superset of $S_{l'-1} \cup \{i''\}$ which must have been generated by the approximation algorithm at some point, a contradiction.

Thus we can assume $l' \geq l$. However, recall that $T_{l, i'}$ is an $M$-cheap solution containing $S_{l'} \cup \{i''\}$. If $l' = l$, this prevents $i''$ from being deleted in iteration $l'$, giving a contradiction. If $l' > l$ then $T_{l', i''}$ can be initialized to a $M$-cheap superset of $S_{l'} \cup \{i''\}$, since $T_{l, i'}$ is such a superset. This also prevents $i''$ from being deleted in iteration $l''$, giving a contradiction. 

In all cases assuming $T_{l, i'}$ is not in $\mathcal{S}_{max}$ leads to a contradiction, so $T_{l, i'}$ is in $\mathcal{S}_{max}$ and thus $S_{l-1} \cup \{i\}$ is in $\mathcal{S}$. 

\textbf{Property 3: }$\mathcal{S}$ is defined so that all maximal independent sets are of size $k$, giving the property.
\end{proof}

\begin{corollary}\label{cor:1system}
Any run of \textsc{GreedyMax} outputs an $M$-cheap, $\frac{1}{2}$-maximizer of $f_y(S)$ over the system $(F, \mathcal{S})$ as defined in Lemma \ref{lemma:1system}.
\end{corollary}
\begin{proof}
Properties 1 and 2 in Lemma~\ref{lemma:1system} imply that at each step, \textsc{GreedyMax} adds the element to its current solution that maximizes the objective $f_y(S)$ while maintaining that the current solution is in $\mathcal{S}$. Thus \textsc{GreedyMax} is exactly the greedy algorithm in Theorem \ref{thm:independencesystems} for maximizing a monotone submodular objective over an independence system. By Lemma~\ref{lemma:nonempty}, \textsc{GreedyMax} always finds a maximal independent set, and the definition of $\mathcal{S}$ guarantees that this maximal independent set is $M$-cheap. Lemma~\ref{lemma:lbsubmodularity} gives that $f_y(S)$ is a monotone submodular function of $S$. Then, Property 3 combined with Theorem~\ref{thm:independencesystems} implies the solution output by \textsc{GreedyMax} is a $\frac{1}{2}$-maximizer.
\end{proof}

Of course, maximizing over an arbitrary 1-system is of little use. In particular, we would like to show that the 1-system Lemma \ref{lemma:1system} shows we are maximizing over approximates the 1-system of subsets of solutions whose cost on the fixed clients is at most $M$. The next lemma shows that while all such solutions may not be in this 1-system, all solutions that are $\frac{M}{\gamma}$-cheap are.

\begin{lemma}\label{lemma:inclusion}
In any run of \textsc{GreedyMax}, let $\mathcal{S}$ be defined as in Lemma \ref{lemma:1system}. For the value $M$ passed to \textsc{GreedyMax} in this run and any solution $S$ which is $\frac{M}{\gamma}$-cheap, $S \in \mathcal{S}$. 
\end{lemma}
\begin{proof}
Fix any such $S$. Let $\mathcal{B}$ be defined as in Lemma \ref{lemma:1system}. For any element $B$ of $\mathcal{B}$, it must be the case that running a $\gamma$-approximation on the incremental $k$-median instance with existing cluster centers $B$ produced a solution with cost greater than $M$. This implies that for any $B$ in $\mathcal{B}$, the incremental $k$-median instance with existing cluster centers $B$ has optimal solution with cost greater than $\frac{M}{\gamma}$. However, for any subset $S'$ of $S$, the optimal solution to the incremental $k$-median instance with existing cluster centers $S'$ has cost at most $\frac{M}{\gamma}$ since $S$ is a feasible solution to this instance which is $\frac{M}{\gamma}$-cheap. Thus no subset of $S$ is in $\mathcal{B}$, and hence $S$ is in $\mathcal{S}$. 
\end{proof}

Lastly, we show that \textsc{SepOracle} never incorrectly outputs that a point is infeasible, i.e., that the region \textsc{SepOracle} considers feasible strictly contains the region that is actually feasible in the universal $k$-median LP. 

\begin{lemma}\label{lemma:basicfeasibility}
If $x, y, r$ is feasible for the universal $k$-median LP, \textsc{SepOracle} outputs ``Feasible''. 
\end{lemma}

\begin{proof}
\textsc{SepOracle} can exactly check all constraints besides the regret constraint set, so assume that if \textsc{SepOracle} outputs that $x, y, r$ is not feasible, it outputs that $\sum_{j \in C'} \sum_{i \in F} c_{ij} y_{ij} - S(C') \leq M - \sum_{j \in C_f} \sum_{i \in F} c_{ij} y_{ij} + r$ is violated for some $M, S$. In particular, it only outputs that this constraint is violated if it actually is violated. If this constraint is violated, then since by Corollary \ref{cor:1system} $S$ is $M$-cheap:

\begin{equation*}
\begin{aligned}
\sum_{j \in C'} \sum_{i \in F} c_{ij} y_{ij} - S(C') &> M - \sum_{j \in C_f} \sum_{i \in F} c_{ij} y_{ij} + r\\
\sum_{j \in C'} \sum_{i \in F} c_{ij} y_{ij} + \sum_{j \in C_f} \sum_{i \in F} c_{ij} y_{ij} &> S(C') + M + r\\
\sum_{j \in C' \cup C_f} \sum_{i \in F} c_{ij} y_{ij} &> S(C') + M + r\\
\end{aligned}
\end{equation*}
$$\geq S(C') + S(C_f) + r = S(C' \cup C_f) + r \geq \opt(C' \cup C_f) + r$$

Which implies the point $x, y, r$ is not feasible for the universal $k$-median LP.
\end{proof}

We now have all the tools to prove our overall claim:
\begin{lemma}\label{lemma:fractionalsolution}
If there exists a deterministic polynomial-time $\gamma$-approximation algorithm for the $k$-median problem, then for every $\epsilon > 0$ there exists a deterministic algorithm that outputs a $(2\gamma(1 + \epsilon), 2)$-approximate fractional solution to the universal $k$-median problem in polynomial time.
\end{lemma}
\begin{proof}
We use the ellipsoid method where \textsc{SepOracle} is used as the separation oracle.
By Lemma \ref{lemma:basicfeasibility} since the minimum regret solution is a feasible solution to the universal $k$-median LP, it is also considered feasible by \textsc{SepOracle}. Then, the solution $x^*, y^*, r^*$ output by the ellipsoid method satisfies $r^* \leq \mr$.

Suppose the ellipsoid method outputs $x^*, y^*, r^*$ such that $x^*, y^*$ are not a $(2\gamma(1+\epsilon), 2)$-approximate solution. This means there exists $S, C_f \subseteq C' \subseteq C$ such that:
\begin{equation*}
\begin{aligned}
\sum_{j \in C'} \sum_{i\in F} c_{ij} y_{ij}^* &> 2\gamma(1+\epsilon) S(C') + 2\cdot\mr\\
\sum_{j \in C' \setminus C_f} \sum_{i\in F} c_{ij} y_{ij}^* - S(C' \setminus C_f) &> 2\gamma(1+\epsilon) S(C_f) + (2\gamma(1+\epsilon)-1) S(C' \setminus C_f) - \sum_{j \in C_f} \sum_{i\in F} c_{ij} y_{ij}^* + 2\cdot\mr\\
&\geq 2\left[\gamma(1+\epsilon) S(C_f) - \sum_{j\in C_f}\sum_{i\in F} c_{ij} y_{ij}^* + \mr\right].
\end{aligned}
\end{equation*}
Thus,
for the value of $M$ in the set $\{0, 1, 1+\epsilon, (1+\epsilon)^2, \ldots (1+\epsilon)^{\lceil\log_{1+\epsilon}(\gamma|C_f|\max_{i,j}c_{ij})\rceil + 1}\}$ contained in the interval $[\gamma S(C_f), \gamma(1+\epsilon)S(C_f)]$, we have
\begin{equation*}
\sum_{j \in C' \setminus C_f} \sum_{i\in F} c_{ij} y_{ij}^* - S(C' \setminus C_f)
\geq 2\left[M - \sum_{j \in C_f} \sum_{i\in F} c_{ij} y_{ij}^* + \mr\right]
\geq 2\left[M - \sum_{j \in C_f} \sum_{i\in F} c_{ij} y_{ij}^* + r^*\right],
\text{ since } r^* \leq \mr.
\end{equation*}
Then, consider the iteration in \textsc{SepOracle} where it runs \textsc{GreedyMax} for this value of $M$. Since $M \geq \gamma S(C_f)$, $S$ is $\frac{M}{\gamma}$-cheap. Thus by Lemma \ref{lemma:inclusion}, $S$ is part of the independence system $\mathcal{S}$ specified in Lemma \ref{lemma:1system} which \textsc{GreedyMax} finds a maximizer for in this iteration, and thus the maximum of the objective in this independence system is at least $2[M - \sum_{j \in C_f} \sum_i c_{ij} y_{ij}^* + r^*]$. By Corollary \ref{cor:1system}, \textsc{SepOracle} thus finds some $S', C'' \subseteq C \setminus C_f$ such that $S'$ is $M$-cheap and for which $\sum_{j \in C''} \sum_i c_{ij} y_{ij}^* - S'(C'')$ is at least $M - \sum_{j \in C_f} \sum_i c_{ij} y_{ij}^* + r^*$. But this means \textsc{SepOracle} will output that $x^*, y^*, r^*$ is infeasible, which means the ellipsoid algorithm cannot output this solution, a contradiction.
\end{proof}

\subsection{Rounding the Fractional Solution for Universal $k$-Median with Fixed Clients}
\label{subsection:kmedian-rounding}

\begin{proof}[Proof of Theorem \ref{thm:fixedkmalgorithm}]
The algorithm is as follows: Use the algorithm of Lemma~\ref{lemma:fractionalsolution} with error parameter $\frac{\epsilon}{54\gamma}$ to find a $(2\gamma(1 + \frac{\epsilon}{54\gamma}), 2)$-approximate fractional solution. Let $f_j$ be the connection cost of this fractional solution for client $j$. Construct a $k$-median with discounts instance with the same clients $C$ and cluster centers $F$ where client $j$ has discount 0 if it was originally a fixed client, and discount $3f_j$ if it was originally a unfixed client. The solution to this instance given by Lemma \ref{lemma:kmbicriteria} is the solution for the universal $k$-median instance.

Again using the integrality gap upper bound of 3 for $k$-median, we have:
%
%
%
\begin{equation}\label{eq:frac-opt-2}
\mr = \max_{C'}[\mrs(C') - \opt(C')] \geq \max_{C'}\left[\mrs(C') - 3\sum_{j \in C'}f_j\right] = \sum_{j \in C_f}(m_j - 3f_j) + \sum_{j \in C \setminus C_f} (m_j - 3f_j)^+.
\end{equation}

The cost of the minimum regret solution in the $k$-median with discounts instance is given by:
%
%
\begin{equation}\label{eq:frac-opt-3}
\sum_{j \in C_f}m_j + \sum_{j \in C \setminus C_f} (m_j - 3f_j)^+ = \sum_{j \in C_f}3f_j + \sum_{j \in C_f}(m_j - 3f_j) + \sum_{j \in C \setminus C_f} (m_j - 3f_j)^+ \leq \sum_{j \in C_f}3f_j + \mr, \text{ by Eq.~\eqref{eq:frac-opt-2}}.
\end{equation}

Let $c_j$ be the connection cost of the algorithm's solution for client $j$. Lemma \ref{lemma:kmbicriteria} and Eq.~\eqref{eq:frac-opt-3} give:
%
%
\begin{equation}\label{eq:frac-opt-4}
\begin{aligned}
\sum_{j \in C_f}c_j + \sum_{j \in C \setminus C_f}(c_j - 9 \cdot 3f_j)^+ &\leq 6\left[\sum_{j \in C_f}3f_j + \mr\right]\\
\sum_{j \in C_f}c_j + \sum_{j \in C \setminus C_f}(c_j - 27f_j)^+ &\leq \sum_{j \in C_f}18f_j + 6\cdot\mr\\
\implies \quad  \max_{C'}\left[\sum_{j \in C'}c_j - 27\sum_{j \in C'}f_j\right] = \sum_{j \in C_f} (c_j - 27f_j) + \sum_{j \in C \setminus C_f} (c_j - 27f_j)^+ &\leq 6 \cdot \mr.
\end{aligned}
\end{equation}


Lemma \ref{lemma:fractionalsolution} then gives that for any valid $C'$:
\begin{equation}\label{eq:frac-opt-5}
\fracc(C') = \sum_{j \in C'} f_j \leq 2\gamma\left(1 + \frac{\epsilon}{54\gamma}\right)\cdot\opt(C') + 2\cdot\mr.
\end{equation}
Using Eq.~\eqref{eq:frac-opt-4} and \eqref{eq:frac-opt-5}, we can then conclude that
$$\forall C_f \subseteq C' \subseteq C: \sum_{j \in C'} c_j \leq 27\sum_{j \in C'}f_j + 6\cdot\mr \leq (54\gamma+\epsilon)\cdot\opt(C') + 60\cdot\mr.\mbox{\qedhere}$$
\end{proof}

\section{Universal $\ell_p$-Clustering with Fixed Clients}
\label{sec:lpalg2}

In this section, we give the following theorem:

\begin{theorem}\label{thm:lp-fixed}
For all $p \geq 1$, if there exists a $\gamma$-approximation for $\ell_p$-clustering, then for all $\epsilon > 0$ there exists a $(54p\gamma \cdot 2^{1/p} + \epsilon, 108p^2 + 6p^{1/p}+ \epsilon)$-approximate universal algorithm for $\ell_p$-clustering with fixed clients.

\end{theorem}

In particular, we get from known results \cite{AhmadianNSW17, GuptaT08}:
\begin{itemize}
    \item A $(162p^2 \cdot 2^{1/p} + \epsilon, 108p^2 + 18p^{1/p}+ \epsilon)$-approximate universal algorithm for $\ell_p$-clustering with fixed clients for all $\epsilon > 0$, $p \geq 1$.
    \item A $(459, 458)$-approximate universal algorithm for $k$-means with fixed clients.
\end{itemize}

The algorithm for universal $\ell_p$-clustering with fixed clients follows by combining techniques from $\ell_p$-clustering and $k$-median with fixed clients. 

\subsection{Finding a Fractional Solution}
We reuse the subroutine \textsc{GreedyMax} to do submodular maximization over an independence system whose bases are $M$-cheap solutions (that is, solutions with $\ell_p$-objective at most $M$ on only the fixed clients), and use the submodular function $f_{y,Y}$ with varying choices of $Y$ as we did for $\ell_p$-clustering. We can extend Lemma~\ref{lemma:maxatintegral} as follows:

\begin{lemma}\label{lemma:maxatintegral2}
For any two solutions $y, S$, if the global maximum of $\fracc_p(\bd) - S_p(\bd)$ over $1^{C_f} \times [0, 1]^{C \setminus C_f}$ is positive, then there is a maximizer that is in $1^{C_f} \times \{0, 1\}^{C \setminus C_f}$, i.e. $$\max_{\bd \in 1^{C_f} \times [0, 1]^{C \setminus C_f}:}\left[\fracc_p(\bd) - S_p(\bd)\right] = \max_{C_s \subseteq C' \subseteq C} \left[\fracc_p(C') - S_p(C')\right].$$
\end{lemma}

The proof follows exactly the same way as  Lemma~\ref{lemma:maxatintegral}. In that proof, the property we use of having no fixed clients is that if the global maximum is not the all zeroes vector, then it is positive and so  $\fracc_p(\bd) > S_p(\bd)$. In the statement of Lemma~\ref{lemma:maxatintegral2}, we just assume positivity instead. This shows that it is still fine to output separating hyperplanes based on fractional realizations of clients in the presence of fixed clients. The only time it is maybe not fine is in a fractional realization where if the ``regret'' of $\fracc$ is negative, but in this case we will not output a separating hyperplane anyway.

\begin{figure}[ht]
\fbox{\begin{minipage}{.98\textwidth}
\textproc{Fixed-$\ell_p$-SepOracle}($(x, y, r)$, $F$, $C_f$, $C$):\newline
\textbf{Input:} A fractional solution $x, y, r$, set of cluster centers $F$, set of fixed clients $C_f$, set of all clients $C$
\begin{algorithmic}[1]
\If{Any constraint in the universal $\ell_p$-clustering LP except the regret constraint set is violated} \Return the violated constraint
\EndIf
\State $T_0 \leftarrow$ output of $\gamma$-approximation algorithm $A$ for $\ell_p$-clustering run on instance with cluster centers $F$, clients $C_f$ 
\State $c_{\min} \leftarrow \min_{i \in F, j \in C \setminus C_f} c_{ij}^p, c_{\max} \leftarrow \sum_{j \in C \setminus C_f} \max_{i \in F} c_{ij}^p$
\State $Y_f \leftarrow \sum_{j \in C_f} \sum_{i \in F} c_{ij}^p y_{ij}$
\For{$M \in \{0, 1, 1+\epsilon, (1+\epsilon)^2, \ldots (1+\epsilon)^{\lceil\log_{1+\epsilon}(\gamma |C_f|^{1/p}\max_{i,j}c_{ij})\rceil + 1}\}$ such that $T_0$ is $M$-cheap} 
\For {$Y \in \{0, c_{\min}, c_{\min}(1+\epsilon'), c_{\min}(1+\epsilon')^2, \ldots c_{\min} (1+\epsilon')^{\lceil \log_{1+\epsilon'} c_{\max}/c_{\min}  \rceil}\}$}
\State $S' \leftarrow$ \textproc{GreedyMax}($(x, y, r)$, $F$, $C_f$, $C$, $M$, $T_0$, $f_{y,Y}$)
\State $\bd' \leftarrow \argmax_{\bd \in 1^{C_f} \times [0, 1]^{C \setminus C_f}: \sum_{j \in C \setminus C_f} d_j \sum_{i \in F}  c_{ij}^p y_{ij}\leq Y} \sum_{j \in {C \setminus C_f}} d_j \sum_{i \in F} c_{ij}^p y_{ij} - \sum_{j \in {C \setminus C_f}} d_j \min_{i \in S} c_{ij}^p$ \label{line:demand2} 
\If{$\frac{1}{p(Y_f + Y)^{1-1/p}}\left[\sum_{j \in C} d_j' \sum_{i \in F} c_{ij}^p y_{ij} - \sum_{j \in C} d_j' \min_{i \in S} c_{ij}^p\right] > r$}
\State \Return $\frac{1}{p(Y_f + Y)^{1-1/p}}\left[\sum_{j \in C} d_j' \sum_{i \in F} c_{ij}^p y_{ij} - \sum_{j \in C} d_j' \min_{i \in S} c_{ij}^p\right] \leq r$
\EndIf
\EndFor
\EndFor
\State \Return ``Feasible''
\end{algorithmic}
\end{minipage}}
\caption{Approximate Separation Oracle for Universal $\ell_p$-Clustering with Fixed Clients. \textsc{GreedyMax} is the same algorithm as presented in Figure~\ref{fig:submaxalg} for $k$-median.}
\label{fig:lpseparationoracle}
\end{figure}

\begin{lemma}\label{lemma:lp-fixed-solving}
If there exists a $\gamma$-approximation for $\ell_p$-clustering, then for all $\epsilon > 0$, $\alpha = 2^{1/p}\gamma (1+\epsilon), \beta = 2p(1+\epsilon)$ there exists an algorithm that outputs an $(\alpha, \beta)$-approximate universal fractional solution for $\ell_p$-clustering with fixed clients.
\end{lemma}
\begin{proof}
If \textsc{Fixed-$\ell_p$-SepOracle} ever outputs an inequality in the regret constraint set, for the corresponding $Y_f, Y, \bd', S$, let $\fracc_p^q(\bd), \sol_p^q(\bd)$ denote the $\ell_p^q$ costs of the fractional solution and $S$ as before. Then we have by definition of $Y_f$ and the constraint that $\sum_{j \in C} d_j \sum_{i \in F}  c_{ij}^p y_{ij}\leq Y$:

$$r < \frac{1}{p(Y_f + Y)^{1-1/p}}\left[\sum_{j \in C} d_j' \sum_{i \in F} c_{ij}^p y_{ij} - \sum_{j \in C} d_j' \min_{i \in S} c_{ij}^p\right]  \leq$$
$$\frac{1}{\sum_{j = 0}^{p-1} \fracc_p^j(\bd') \sol_p^{p-1-j}(\bd')}\left[\fracc_p^p(\bd') - \sol_p^p(\bd')\right] = \fracc_p(\bd') - \sol_p(\bd').$$

The second inequality uses that \textsc{Fixed-$\ell_p$-SepOracle} only outputs an inequality in the regret constraint set such that $\fracc_p(\bd') > \sol_p(\bd')$. 
We then have by Lemma~\ref{lemma:maxatintegral2} that for any feasible fractional solution, the inequality output by \textsc{Fixed-$\ell_p$-SepOracle} is satisfied.

Now, suppose there exists some $\bd, \sol$ such that $\fracc_p(\bd) > \alpha \sol_p(\bd) + \beta r$ (for $r \geq 0$). Consider the values of $Y, M$ iterated over by \textsc{Fixed-$\ell_p$-SepOracle} such that $\sum_{j \in C \setminus C_f} d_j \sum_{i \in F}  c_{ij}^p y_{ij}\leq Y \leq (1+\epsilon)( \sum_{j \in C \setminus C_f} d_j \sum_{i \in F}  c_{ij}^p y_{ij})$ and $\gamma \sol_p(C_f) \leq M \leq\gamma (1+\epsilon) \sol_p(C_f)$. Then:

\begin{equation*}
\begin{aligned}[rll]
\left(\sum_{j \in C'} \sum_{i\in F} c_{ij}^p y_{ij}\right)^{1/p} &> \alpha \sol_p(C') + \beta r& \\
\left(\sum_{j \in C'} \sum_{i\in F} c_{ij}^p y_{ij}\right)^{1/p} - \alpha \sol_p(C')&> \beta r& \\
\frac{\sum_{j \in C'} \sum_{i\in F} c_{ij}^p y_{ij} - \alpha^p \sol_p^p(C')}{\left(\sum_{j \in C'} \sum_{i\in F} c_{ij}^p y_{ij}\right)^{1-1/p}}&> \beta r&(i) \\
\frac{(1+\epsilon)\left(\sum_{j \in C'} \sum_{i\in F} c_{ij}^p y_{ij} - \alpha^p \sol_p^p(C')\right)}{\left(Y_f + Y\right)^{1-1/p}}&> \beta r&(ii) \\
\sum_{j \in C'} \sum_{i\in F} c_{ij}^p y_{ij} - \alpha^p \sol_p^p(C')&> \frac{\beta r \left(Y_f + Y\right)^{1-1/p}}{1 + \epsilon} \\
\sum_{j \in C' \setminus C_f} \sum_{i\in F} c_{ij}^p y_{ij} - \alpha^p \sol_p^p(C' \setminus C_f)&> \frac{\beta r (Y_f + Y)^{1-1/p}}{1 + \epsilon}+\alpha^p \sol_p^p(C_f) -  \sum_{j \in C_f} \sum_{i\in F} c_{ij}^p y_{ij} \\
\sum_{j \in C' \setminus C_f} \sum_{i\in F} c_{ij}^p y_{ij} - \alpha^p \sol_p^p(C' \setminus C_f)&> \frac{\beta r (Y_f + Y)^{1-1/p}}{1 + \epsilon}+\alpha^p \left(\frac{M}{\gamma(1+\epsilon)}\right)^p -  \sum_{j \in C_f} \sum_{i\in F} c_{ij}^p y_{ij}. & (iii)\\
\end{aligned}
\end{equation*}

$(i)$ follows from the fact that $\sol_p(C') > \sum_{j \in C'}\sum_{i \in F}c_{ij}^p$ if $a > b$. $(ii)$ follows from definitions of $Y_f, Y$. $(iii)$ follows from the choice of $M$. 
Let $\bd'_{C'}$ denote the vector whose $j$th element is $\bd'_j$ if $j \in C'$ and $0$ otherwise.  By the analysis in Section~\ref{thm:fixedkmalgorithm}, since $\sol$ is $M/\gamma$-cheap it is in the independence system that \textsc{GreedyMax} finds a $1/2$-maximizer for. That is, \textsc{GreedyMax} outputs some $S$ and \textsc{Fixed-$\ell_p$-Oracle} finds some $\bd'$ such that $S$ is $M$-cheap, $\sum_{j \in C \setminus C_f} d_j' \sum_{i \in F} c_{ij}^p y_{ij} \leq Y$, and such that:

\begin{equation*}
\begin{aligned}[rll]
\sum_{j \in C' \setminus C_f} d_j'\sum_{i\in F} c_{ij}^p y_{ij} - \alpha^p S_p^p(\bd'_{C \setminus C_f})&> \frac{1}{2}\left[ \frac{\beta r (Y_f + Y)^{1-1/p}}{1 + \epsilon}+\alpha^p \left(\frac{M}{\gamma(1+\epsilon)}\right)^p -  \sum_{j \in C_f} \sum_{i\in F} c_{ij}^p y_{ij}\right] &\\
\sum_{j \in C' \setminus C_f} d_j'\sum_{i\in F} c_{ij}^p y_{ij} - S_p^p(\bd'_{C \setminus C_f})&> \frac{1}{2}\left[ \frac{\beta r (Y_f + Y)^{1-1/p}}{1 + \epsilon}+\alpha^p \left(\frac{S_p(C_f)}{\gamma(1+\epsilon)}\right)^p\right] -  \sum_{j \in C_f} d_j' \sum_{i\in F} c_{ij}^p y_{ij} & (iv)\\
\sum_{j \in C' \setminus C_f} d_j'\sum_{i\in F} c_{ij}^p y_{ij} - S_p^p(\bd'_{C \setminus C_f})&> \frac{\beta r (Y_f + Y)^{1-1/p}}{2(1 + \epsilon)}+S_p^p(C_f) -  \sum_{j \in C_f} d_j' \sum_{i\in F} c_{ij}^p y_{ij} & (v)\\
\sum_{j \in C'} d_j'\sum_{i\in F} c_{ij}^p y_{ij} - S_p^p(\bd')&> \frac{\beta r (Y_f + Y)^{1-1/p}}{2(1 + \epsilon)} &\\
\frac{\sum_{j \in C'} d_j'\sum_{i\in F} c_{ij}^p y_{ij} - S_p^p(\bd')}{p (Y_f + Y)^{1-1/p}}&> \frac{\beta r }{2p(1 + \epsilon)} &\\
\frac{\sum_{j \in C'} d_j'\sum_{i\in F} c_{ij}^p y_{ij} - S_p^p(\bd')}{p (Y_f + Y)^{1-1/p}}&> r. & (vi)\\
\end{aligned}
\end{equation*}
$(iv)$ follows the $M$-cheapness of $S$. $(v)$ follows from the choice of $\alpha$. $(vi)$ follows from the choice of $\beta$. So, \textsc{Fixed-$\ell_p$-SepOracle} outputs an inequality as desired.
\end{proof}

\subsection{Rounding the Fractional Solution}

Again, we show how to generalize the approach for rounding fractional solutions for $k$-median with fixed clients to round fractional solutions for $\ell_p$ clustering with fixed clients. We extend Lemma~\ref{lemma:gaprealization} as follows:

\begin{lemma}\label{lemma:gaprealization2}
Suppose $\alg$ and $\sol$ are two (possibly virtual) solutions to an $\ell_p$-clustering instance with fixed clients $C_f$, such that there is a subset of clients $C^* \subset (C \setminus C_f)$ such that for every client in $C^*$ $\alg$'s connection cost is greater than $p$ times $\sol$'s connection cost, and for every client in $C \setminus C_f \setminus C^*$, $\sol$'s connection cost is at least $\alg$'s connection cost. Then 

\begin{equation*}
f(C') :=  
\begin{cases} 
\frac{\alg_p^p(C') - \sol_p^p(C')}{\alg_p^{p-1}(C')} & \alg_p^p(C') > 0 \\
0 & \alg_p^p(C') = 0
\end{cases}
\end{equation*}
is maximized by $C_f \cup C^*$.

\end{lemma}
The proof follows exactly as that of Lemma~\ref{lemma:gaprealization}.

\begin{lemma}\label{lemma:lp-fixed-rounding}
There exists an algorithm that given any $(\alpha, \beta)$-approximate universal fractional solution for $\ell_p$-clustering with fixed clients, outputs a $(54p\alpha, 54p\beta + 18p^{1/p})$-approximate universal integral solution.
\end{lemma}
\begin{proof}
Let $\sol$ be the virtual solution whose connection costs are 3 times the fractional solution's for all clients. The algorithm is to solve the $\ell_p^p$-clustering with discounts instance using Lemma~\ref{lemma:kmbicriteria-lpp} where the discounts are 0 for fixed clients and $2$ times $\sol$'s connection costs for the remaining clients. Note that using these discounts, the $\ell_p^p$-clustering with discounts objective equals $\max_{C_f \subseteq C' \subseteq C}\left[\alg^p_p(C') - 2^p \cdot \sol^p_p(C' \setminus C_f)\right]$ instead of $\max_{C_f \subseteq C' \subseteq C}\left[\alg^p_p(C') - 2^p \cdot \sol^p_p(C')\right]$. Let $\alg$ be the output solution. We will again bound $\alg$'s cost against the virtual solution $\widetilde{\sol}$ whose connection costs are $\sol$'s connection costs times $p$ for non-fixed clients $j$ such that $\alg$'s connection cost to $j$ is at least $18$ times $\sol$'s but less than $18p$ times $18 \cdot \sol$'s, and the same as $\sol$'s for the remaining clients.

We use $\max_{C'}$ to denote $\max_{C_f \subseteq C' \subseteq C}$. If $\max_{C'}\left[\alg_p(C') - 18 \widetilde{\sol}(C')\right] \leq 0$ then $\alg$'s cost is always bounded by 18 times $\widetilde{\sol}$'s cost and we are done. So assume $\max_C'[\alg_p(C') - 18\widetilde{\sol}_p(C')] > 0$.  Let $C_1 = \argmax_{C'}\left[\alg_p(C') - 18 \widetilde{\sol}_p(C')\right]$ and $C_2 = \argmax_{C'}\left[\mrs_p^p(C') - 2^p \cdot \sol_p^p(C')\right]$. Like in the proof of Lemma~\ref{lemma:lp-rounding}, via Lemma~\ref{lemma:gaprealization2} we have:

$$\max_{C'}\left[\alg_p(C') - 18 \widetilde{\sol}_p(C') \right] =
\frac{\max_{C'}\left[\alg_p^p(C') - 18^p{\sol}_p^p(C')\right]}{\alg_p^{p-1}(C_1)} =$$
$$\frac{\max_{C'}\left[\alg_p^p(C') - 18^p{\sol}_p^p(C' \setminus C_f)\right] - 18^p{\sol}_p^p(C_f)}{\alg_p^{p-1}(C_1)} \leq $$
$$\frac{2}{3} \cdot 9^p\frac{ \max_{C'}\left[\mrs_p^p(C') - 2^p{\sol}_p^p(C' \setminus C_f)\right] - 18^p{\sol}_p^p(C_f)}{\alg_p^{p-1}(C_1)} \leq $$
$$\frac{2}{3} \cdot 9^p \frac{\max_{C'}\left[\mrs_p^p(C') - 2^p{\sol}_p^p(C')\right]}{\alg_p^{p-1}(C_1)}.$$

Using the same analysis as in Lemma~\ref{lemma:lp-rounding} we can upper bound this final quantity by $18p^{1/p} \cdot \mr$, proving the lemma. 
\end{proof}

Theorem~\ref{thm:lp-fixed} follows from Lemmas~\ref{lemma:lp-fixed-solving} and~\ref{lemma:lp-fixed-rounding}.
\section{Universal $k$-Center with Fixed Clients}
\label{sec:kcalg2}

In this section, we discuss how to extend the proof of Theorem~\ref{thm:kcenteralg} to prove the following theorem:
\begin{theorem}
There exists a $(9, 3)$-approximate  algorithm for universal $k$-center with fixed clients.
\end{theorem}
\begin{proof}
To extend the proof of Theorem~\ref{thm:kcenteralg} to the case where fixed clients are present, let $\apx(C')$ denote the cost of a $3$-approximation to the $k$-center problem with client set $C'$; it is well known how to compute $\apx(C')$ in polynomial time \cite{HochbaumS86}. A solution with regret $r$ must be within distance $r_j := \apx(C_s \cup \{j\}) + r$ of client $j$, otherwise in realization $C_s \cup \{j\}$ the solution has regret larger than $r$ due to client $j$. The same algorithm as in the proof of Theorem~\ref{thm:kcenteralg} using this definition of $r_j$ finds $\alg$ within distance $3r_j = 3 \cdot \apx(C_s \cup \{j\}) + 3 \cdot \mr \leq 9 \cdot \opt(C_s \cup \{j\}) + 3 \mr$ of client $j$. $\opt(C') \geq \opt(C_s \cup \{j\})$ for any realization $C'$ and any client $j \in C'$, so this solution is a $(9, 3)$-approximation.
\end{proof}
\section{Hardness of Universal Clustering for General Metrics}\label{section:hardness}

In this section we give some hardness results to help contextualize the algorithmic results.  Much like the hardness results for $k$-median, all our reductions are based on the NP-hardness of approximating set cover (or equivalently, dominating set) due to the natural relation between the two types of problems. We state our hardness results in terms of $\ell_p$-clustering. Setting $p = 1$ gives hardness results for $k$-median, and setting $p = \infty$ (and using the convention $1/\infty = 0$ in the proofs as needed) gives hardness results for $k$-center.

\subsection{Hardness of Approximating $\alpha$}

\begin{theorem}
\label{theorem:alphahardness}
For all $p \geq 1$, finding an $(\alpha, \beta)$-approximate solution for universal $\ell_p$-clustering where $\alpha < 3$ is NP-hard.
\end{theorem}
\begin{proof}

We will show that given a deterministic $(\alpha, \beta)$-approximate algorithm where $\alpha < 3$, we can design an algorithm (using the $(\alpha, \beta)$-approximate algorithm as a subroutine) that solves the set cover problem (i.e. finds a set cover of size $k$ if one exists) giving the lemma by NP-hardness of set cover. The algorithm is as follows: Given an instance of set cover, construct the following instance of universal $\ell_p$-clustering:

\begin{itemize}
\item For each element, there is a corresponding client in the universal $\ell_p$-clustering instance.
\item For each set $S$, there is a cluster center which is distance 1 from the clients corresponding to elements in $S$ and 3 from other all clients.
\end{itemize}

Then, we just run the universal $\ell_p$-clustering algorithm on this instance, and output the sets corresponding to cluster centers this algorithm buys. 

Assume for the rest of the proof that a set cover of size $k$ exists. Then the corresponding $k$ cluster centers are as close as possible to every client, and are always an optimal solution. This gives that $\mr = 0$ for this universal $\ell_p$-clustering instance. 

Now, suppose by contradiction that this algorithm does not solve the set cover problem. That is, for some set cover instance we run an $(\alpha, \beta)$-approximate algorithm where $\alpha < 3$ on the produced $\ell_p$-clustering instance, and it produces a solution $\alg$ that does not choose cluster centers corresponding to a set cover. This means it is distance 3 from some client $j$. For realization $C' = \{j\}$, we have by the definition of $(\alpha, \beta)$-approximation:

$$\alg(C') \leq \alpha \cdot \opt(C') + \beta \cdot \mr \implies 3 \leq \alpha \cdot 1 + \beta \cdot 0 = \alpha$$

Which is a contradiction, giving the lemma. 
\end{proof}

Note that for e.g. $k$-median, we can classically get an approximation ratio of less than 3. So this theorem shows that the universal version of the problem is harder, even if we are willing to use arbitrary large $\beta$.

\subsection{Hardness of Approximating $\beta$}

We give the following result on the hardness of universal $\ell_p$-clustering.

\begin{theorem}
\label{theorem:betahardness}
For all $p \geq 1$, finding an $(\alpha, \beta)$-approximate solution for universal $\ell_p$-clustering where $\beta < 2$ is NP-hard.
\end{theorem}
%
%
%
%
\begin{proof}
We will show that given a deterministic $(\alpha, \beta)$-approximate algorithm where $\beta < 2$, we can design an algorithm (using the $(\alpha, \beta)$-approximate algorithm as a subroutine) that solves the dominating set problem (i.e. outputs at most $k$ vertices which are a dominating set of size $k$ if a dominating set of size $k$ exists) giving the lemma by NP-hardness of dominating set. The algorithm is as follows: Given an instance of dominating set $G=(V, E)$, construct the following instance of universal $\ell_p$-clustering:

\begin{itemize}
\item For each vertex $v \in V$, there is a corresponding $k$-clique of clients in the universal $\ell_p$-clustering instance.
\item For each $(u, v) \in E$, connect all clients in $u$'s corresponding clique to all those in $v$'s.
\item Impose the shortest path metric on the clients, where all edges are length $1$.
\end{itemize}

Then, we just run the universal $\ell_p$-clustering algorithm on this instance, and output the set of vertices corresponding to cluster centers this algorithm buys. 

Assume for the rest of the proof that a dominating set of size $k$ exists in the dominating set instance. Then, a dominating set of size $k$ also exists in the constructed universal $\ell_p$-clustering instance (where the cliques this set resides in correspond to the vertices in the dominating set in the original instance). Thus, there is a solution to the universal $\ell_p$-clustering instance that covers all clients at distance at most $1$. 

We will first show this dominating set solution is a minimum regret solution. Given a dominating set solution, note that in any realization of the demands, $\opt$ can cover $k$ locations at distance $0$, and must cover the rest of the clients at distance at least $1$. Thus, to maximize the regret of a dominating set solution, we pick any $k$ clients covered at distance $1$ by the dominating set, and choose the realization including only these clients.

Now, consider any solution which is not a dominating set. For such a solution, there is some $k$-clique covered at distance 2. We can make such a solution incur regret $2k^{1/p}$ by including all $k$ clients in this clique, with the optimal solution being to buy all cluster centers in this clique at cost $0$. Thus, the dominating set solution is a minimum regret solution, and $\mr = k^{1/p}$.
%
%
%

Now consider any $(\alpha, \beta)$-approximation algorithm and suppose this algorithm when run on the reduced dominating set instance does not produce a  dominating set solution while one exists. Consider the realization $C'$ including only the clients in some $k$-clique covered at distance 2. By definition of $(\alpha, \beta)$-approximation we get:

\begin{equation}
\begin{aligned}
\alg(C') &\leq \alpha \cdot \opt(C') + \beta \cdot \mr\\
2k^{1/p} &\leq 0 + \beta k^{1/p}\\
2 &\leq \beta
\end{aligned}
\end{equation}

If $\beta < 2$, this is a contradiction, i.e. the algorithm will always output a dominating set of size $k$ if one exists. Thus an $(\alpha, \beta)$-approximation algorithm where $\beta < 2$ can be used to solve the dominating set problem, proving the theorem.
\end{proof}

\section{Hardness of Universal Clustering for Euclidean Metrics}
\label{sec:euclidean}

\subsection{Hardness of Approximating $\alpha$}
We can consider the special case of $\ell_p$-clustering where the cluster center and client locations are all points in $\mathbb{R}^d$, and the metric is a $\ell_q$-norm in $\mathbb{R}^d$. One might hope that e.g. for $d = 2$, $\alpha = 1+\epsilon$ is achievable since for the classic Euclidean $k$-median problem, a PTAS exists \cite{AroraRR98}. We show that there is still a lower bound on $\alpha$ even for $\ell_p$-clustering in $\mathbb{R}^2$. 

\begin{theorem}\label{thm:euclideanalphahardness}

For all $p \geq 1$, finding an $(\alpha, \beta)$-approximate solution for universal $\ell_p$-clustering in $\mathbb{R}^2$ using the $\ell_q$-norm where $\alpha < \frac{1+\sqrt{7}}{2}$ for $q = 2$ or $\alpha < 2$ for $q = 1, \infty$ is NP-hard.
\end{theorem}

\begin{proof}

The hardness is via reduction from the discrete $k$-center problem in $\mathbb{R}^2$. Section 3 of \cite{Mentzer16} shows how to reduce an instance of planar 3-SAT (which is NP-hard) to an instance of Euclidean $k$-center in $\mathbb{R}^2$ using the $\ell_q$ norm as the metric such that:

\begin{itemize}
    \item For every client, the distance to the nearest cluster center is 1. 
    \item There exists a $k$-cluster center solution which is distance 1 from all clients if the planar 3-SAT instance is satisfiable, and none exists if the instance is unsatisfiable.
    \item Any solution that is strictly less than distance $\alpha - \epsilon$ away from all clients can be converted in polynomial time to a solution within distance 1 of all clients for $\alpha = \frac{1+\sqrt{7}}{2}$ if $q = 2$, $\alpha = 2$ if $q = 1, \infty$.
\end{itemize}

We note that \cite{Mentzer16}'s reduction is actually to the ``continuous'' version of the problem where every point in $\mathbb{R}^2$ can be chosen as a cluster center, including the points clients are located at. That is, if we use this reduction without modification then the first property is not actually true (since the minimum distance is 0). However, in the proof of correctness for this construction \cite{Mentzer16} shows that (both for the optimal solution and any algorithmic solution) it suffices to only consider cluster centers located at the centers of a set of $\ell_p$ discs of radius 1 chosen such that every client lies on at least one of these discs and no client is contained within any of these discs. So, taking this reduction and then restricting the choice of cluster centers to the centers of these discs, we retrieve an instance with the desired properties.

Now, consider the corresponding instance as a universal $\ell_p$-clustering instance. Like in the proof of Theorem~\ref{theorem:alphahardness}, if the planar 3-SAT instance reduced from is satisfiable, there exists a clustering solution that is as close as possible to every client, i.e. has regret 0. So $\mr = 0$. Thus, an $(\alpha, \beta)$-approximate clustering solution is within distance $\alpha$ of every client (in the realization where only client $j$ appears, $\opt$ is 1 so an $\alpha$-approximate solution must be within distance $\alpha$ of this client). In turn, using the properties of the reduced clustering instance, an $(\alpha, \beta)$-approximation where $\alpha$ is less than the lower bound given in \cite{Mentzer16} can be converted into an algorithm that solves planar 3-SAT.
\end{proof}

\subsection{Hardness of Approximating $\beta$}

We can also show that $\beta = 1$ is NP-hard in $\mathbb{R}^2$ using a similar reduction: 

\begin{theorem}
For all $p \geq 1$, finding an $(\alpha, \beta)$-approximate solution for universal $\ell_p$-clustering in $\mathbb{R}^2$ using the $\ell_q$-norm where $\beta = 1$ for $q = 1, 2, \infty$ is NP-hard.
\end{theorem}
\begin{proof}
We again use a reduction from planar 3-SAT due to \cite{Mentzer16}. This time, we use the reductions in Section 4 of \cite{Mentzer16} for simplicity, which has the properties that:

\begin{itemize}
    \item Every client is distance 0 from a co-located cluster center, and the distance to the second-closest cluster center is 1. 
    \item There exists a $k$-cluster center solution which is distance 1 from all but $k$ clients and distance 0 from $k$ clients (the ones at the cluster centers) if the planar 3-SAT instance is satisfiable, and none exists if the instance is unsatisfiable.
\end{itemize}

Consider any instance reduced from a satisfiable planar 3-SAT instance. The solution in the resulting instance $\sol^*$ with the second property above has regret $k^{1/p}$ (and in fact, this is $\mr$): by the first property above, no solution can be less than distance 1 away from any clients other than the $k$ clients co-located with its cluster centers. In turn, the regret of the $\sol^*$ against any adversarial $\sol$ is maximized by the realization $C'$ only including the clients co-located with the $k$ cluster centers in the $\sol$. We then get $\sol^*(C') - \sol(C') = k^{1/p} - 0 = k^{1/p}$.

Now consider an arbitrary $(\alpha, 1)$-approximate universal solution $\alg$ in this instance. Consider any set of $k$ clients $C'$ not co-located with $\alg$'s cluster centers. $\opt(C') = 0$, so we get $\alg(C') \leq \alpha \cdot \opt(C') + \mr = \mr \leq k^{1/p}$. $\alg$ is distance at least 1 from all clients in $C'$ by construction, so this only holds if $\alg$ is distance 1 from all clients in $C'$. This gives that $\alg$ is distance 1 from all but $k$ clients (those co-located with cluster centers in $\alg$), and distance 0 from the remaining clients. In turn, $\alg$ satisfies the property of a solution corresponding to a satisfying assignment to the planar 3-SAT instance. This shows that an $(\alpha, 1)$-approximate universal solution to $\ell_p$-clustering in $\mathbb{R}^2$ can be used to solve planar 3-SAT.
\end{proof}

\section{Future Directions}

In this paper, we gave the first universal algorithms for clustering problems: $k$-median, $k$-means, and $k$-center (and their generalization to $\ell_p$-clustering). While we achieve constant approximation guarantees for these problems, the actual constants are orders of magnitude larger than the best (non-universal) approximations known for these problems. In part to ensure clarity of presentation, we did not attempt to optimize these constants.  But it is unlikely that our techniques will lead to {\em small} constants for the $k$-median and $k$-means problems (although, interestingly, we got small constants for $k$-center). On the other hand, we show that in general it is \np-hard to find an $(\alpha, \beta)$-approximation algorithm for a universal clustering problem where $\alpha$ matches the approximation factor for the standard clustering problem. Therefore, it is not entirely clear what one should expect: {\em are there universal algorithms for clustering with approximation factors of the same order as the classical (non-universal) bounds?} 

One possible approach to improving the constants is considering algorithms that use more than $k$ cluster centers. For example, our $(9^p, \frac{2}{3} \cdot 9^p)$-approximation for $\ell_p^p$-clustering with discounts can easily be improved to an $(3^p, 3^p)$-approximation if it is allowed to use $2k-1$ cluster centers. This immediately improves all constants in the paper. For example, our $(27, 49)$-approximation for universal $k$-median becomes a $(9, 18)$-approximation if it is allowed to use $2k-1$ cluster centers. Unfortunately, our lower bounds on $\alpha, \beta$ apply even if the algorithm is allowed to use $(1-\epsilon)k \ln n$ cluster centers, but it is an interesting problem to show that e.g. using $(1+\epsilon)k \ln n$ cluster centers allows one to beat either bound.

Another open research direction pertains to Euclidean clustering. Here, we showed that in $\mathbb{R}^d$ for $d \geq 2$, $\alpha$ needs to be bounded away from $1$, which is in  stark contrast to non-universal clustering problems that admit PTASes in constant-dimension Euclidean space. 
But, for $d = 1$, i.e., for universal clustering on a line, the picture is not as clear. On a line, the lower bounds on $\alpha$ are no longer valid, which brings forth the possibility of (non-bicriteria) approximations of regret. Indeed, it is known that 
there is $2$-approximation for universal $k$-median on a line~\cite{KasperskiZ07}, and even better, an {\em optimal} algorithm for universal $k$-center on a line~\cite{AverbakhB97}.
This raises the natural question: {\em can we design a PTAS for the universal $k$-median problem on a line?}



\bibliographystyle{abbrv}
\bibliography{ref}

\begin{thebibliography}{10}

\bibitem{AhmadianNSW17}
S.~Ahmadian, A.~Norouzi-Fard, O.~Svensson, and J.~Ward.
\newblock Better guarantees for k-means and {Euclidean} k-median by primal-dual
  algorithms.
\newblock In {\em Proceedings of the 58th Annual IEEE Symposium on Foundations
  of Computing}, pages 61--72, Oct. 2017.

\bibitem{AlonA92}
N.~Alon and Y.~Azar.
\newblock On-line steiner trees in the {Euclidean} plane.
\newblock In {\em Proceedings of the 8th Annual Symposium on Computational
  Geometry}, pages 337--343, 1992.

\bibitem{AnthonyGGN10}
B.~Anthony, V.~Goyal, A.~Gupta, and V.~Nagarajan.
\newblock A plant location guide for the unsure: Approximation algorithms for
  min-max location problems.
\newblock {\em Math. Oper. Res.}, 35(1):79--101, Feb. 2010.

\bibitem{ArcherRS03}
A.~Archer, R.~Rajagopalan, and D.~B. Shmoys.
\newblock Lagrangian relaxation for the k-median problem: New insights and
  continuity properties.
\newblock In G.~Di~Battista and U.~Zwick, editors, {\em Algorithms - ESA 2003:
  11th Annual European Symposium, Budapest, Hungary, September 16-19, 2003.
  Proceedings}, pages 31--42. Springer Berlin Heidelberg, Berlin, Heidelberg,
  2003.

\bibitem{AroraRR98}
S.~Arora, P.~Raghavan, and S.~Rao.
\newblock Approximation schemes for {Euclidean} k-medians and related problems.
\newblock In {\em Proceedings of the Thirtieth Annual ACM Symposium on Theory
  of Computing}, STOC '98, pages 106--113, New York, NY, USA, 1998. ACM.

\bibitem{AryaGKMMP01}
V.~Arya, N.~Garg, R.~Khandekar, A.~Meyerson, K.~Munagala, and V.~Pandit.
\newblock Local search heuristic for k-median and facility location problems.
\newblock In {\em Proceedings of the Thirty-third Annual ACM Symposium on
  Theory of Computing}, STOC '01, pages 21--29, New York, NY, USA, 2001. ACM.

\bibitem{AverbakhB97}
I.~Averbakh and O.~Berman.
\newblock Minimax regret p-center location on a network with demand
  uncertainty.
\newblock {\em Location Science}, 5(4):247 -- 254, 1997.

\bibitem{AverbakhB00}
I.~Averbakh and O.~Berman.
\newblock Minmax regret median location on a network under uncertainty.
\newblock {\em INFORMS Journal on Computing}, 12(2):104--110, 2000.

\bibitem{BertsimasG89}
D.~Bertsimas and M.~Grigni.
\newblock Worst-case examples for the spacefilling curve heuristic for the
  {Euclidean} traveling salesman problem.
\newblock {\em Operations Research Letter}, 8(5):241--244, Oct. 1989.

\bibitem{BhalgatCK11}
A.~Bhalgat, D.~Chakrabarty, and S.~Khanna.
\newblock Optimal lower bounds for universal and differentially private
  {Steiner} trees and {TSPs}.
\newblock In L.~A. Goldberg, K.~Jansen, R.~Ravi, and J.~D.~P. Rolim, editors,
  {\em Approximation, Randomization, and Combinatorial Optimization. Algorithms
  and Techniques}, pages 75--86, Berlin, Heidelberg, 2011. Springer Berlin
  Heidelberg.

\bibitem{BhattacharyaCMN14}
S.~Bhattacharya, P.~Chalermsook, K.~Mehlhorn, and A.~Neumann.
\newblock New approximability results for the robust k-median problem.
\newblock In {\em Proceedings of the 14th Scandanavian Workshop on Algorithm
  Theory}, pages 51--60, July 1994.

\bibitem{BuschDRRS12}
C.~Busch, C.~Dutta, J.~Radhakrishnan, R.~Rajaraman, and S.~Srivathsan.
\newblock Split and join: Strong partitions and universal {Steiner} trees for
  graphs.
\newblock In {\em 53rd Annual {IEEE} Symposium on Foundations of Computer
  Science, {FOCS} 2012, New Brunswick, NJ, USA, October 20-23, 2012}, pages
  81--90, 2012.

\bibitem{ByrkaGRS13}
J.~Byrka, F.~Grandoni, T.~Rothvo{\ss}, and L.~Sanit{\`{a}}.
\newblock Steiner tree approximation via iterative randomized rounding.
\newblock {\em J. {ACM}}, 60(1):6:1--6:33, 2013.

\bibitem{CalinescuCPV11}
G.~Calinescu, C.~Chekuri, M.~P\'{a}l, and J.~Vondr\'{a}k.
\newblock Maximizing a monotone submodular function subject to a matroid
  constraint.
\newblock {\em SIAM J. Comput.}, 40(6):1740--1766, Dec. 2011.

\bibitem{ChakrabartyS19}
D.~Chakrabarty and C.~Swamy.
\newblock Approximation algorithms for minimum norm and ordered optimization
  problems.
\newblock In {\em Proceedings of the 51st Annual ACM SIGACT Symposium on Theory
  of Computing}, STOC 2019, page 126–137, New York, NY, USA, 2019.
  Association for Computing Machinery.

\bibitem{CharikarCP05}
M.~Charikar, C.~Chekuri, and M.~P\'{a}l.
\newblock Sampling bounds for stochastic optimization.
\newblock In {\em Proceedings of the 8th International Workshop on
  Approximation, Randomization and Combinatorial Optimization Problems, and
  Proceedings of the 9th International Conference on Randamization and
  Computation: Algorithms and Techniques}, APPROX'05/RANDOM'05, pages 257--269,
  Berlin, Heidelberg, 2005. Springer-Verlag.

\bibitem{CharikarGST99}
M.~Charikar, S.~Guha, E.~Tardos, and D.~B. Shmoys.
\newblock A constant-factor approximation algorithm for the k-median problem
  (extended abstract).
\newblock In {\em Proceedings of the Thirty-first Annual ACM Symposium on
  Theory of Computing}, STOC '99, pages 1--10, New York, NY, USA, 1999. ACM.

\bibitem{Chekuri07}
C.~Chekuri.
\newblock Routing and network design with robustness to changing or uncertain
  traffic demands.
\newblock {\em {SIGACT} News}, 38(3):106--129, 2007.

\bibitem{DhamdhereGRS05}
K.~Dhamdhere, V.~Goyal, R.~Ravi, and M.~Singh.
\newblock How to pay, come what may: Approximation algorithms for demand-robust
  covering problems.
\newblock In {\em 46th Annual {IEEE} Symposium on Foundations of Computer
  Science {(FOCS} 2005), 23-25 October 2005, Pittsburgh, PA, USA, Proceedings},
  pages 367--378, 2005.

\bibitem{FeigeJMM07}
U.~Feige, K.~Jain, M.~Mahdian, and V.~S. Mirrokni.
\newblock Robust combinatorial optimization with exponential scenarios.
\newblock In {\em Integer Programming and Combinatorial Optimization, 12th
  International {IPCO} Conference, Ithaca, NY, USA, June 25-27, 2007,
  Proceedings}, pages 439--453, 2007.

\bibitem{Gonzalez85}
T.~F. Gonzalez.
\newblock Clustering to minimize the maximum intercluster distance.
\newblock {\em Theor. Comput. Sci.}, 38:293--306, 1985.

\bibitem{GorodezkyKSS10}
I.~Gorodezky, R.~D. Kleinberg, D.~B. Shmoys, and G.~Spencer.
\newblock Improved lower bounds for the universal and a priori {TSP}.
\newblock In M.~Serna, R.~Shaltiel, K.~Jansen, and J.~Rolim, editors, {\em
  Approximation, Randomization, and Combinatorial Optimization. Algorithms and
  Techniques}, pages 178--191, Berlin, Heidelberg, 2010. Springer Berlin
  Heidelberg.

\bibitem{GrandoniGLMSS08}
F.~Grandoni, A.~Gupta, S.~Leonardi, P.~Miettinen, P.~Sankowski, and M.~Singh.
\newblock Set covering with our eyes closed.
\newblock In {\em Proceedings of the 49th Annual IEEE Symposium on Foundations
  of Computer Science}, Oct. 2008.

\bibitem{GrotschelLS81}
M.~Gr{\"{o}}tschel, L.~Lov{\'{a}}sz, and A.~Schrijver.
\newblock The ellipsoid method and its consequences in combinatorial
  optimization.
\newblock {\em Combinatorica}, 1(2):169--197, 1981.

\bibitem{GuhaM09}
S.~Guha and K.~Munagala.
\newblock Exceeding expectations and clustering uncertain data.
\newblock In {\em Proceedings of the Twenty-Eighth ACM SIGMOD-SIGACT-SIGART
  Symposium on Principles of Database Systems}, PODS ’09, page 269–278, New
  York, NY, USA, 2009. Association for Computing Machinery.

\bibitem{Gupta:2006:OND:1109557.1109665}
A.~Gupta, M.~T. Hajiaghayi, and H.~R\"{a}cke.
\newblock Oblivious network design.
\newblock In {\em Proceedings of the Seventeenth Annual ACM-SIAM Symposium on
  Discrete Algorithm}, SODA '06, pages 970--979, Philadelphia, PA, USA, 2006.
  Society for Industrial and Applied Mathematics.

\bibitem{GuptaNR14}
A.~Gupta, V.~Nagarajan, and R.~Ravi.
\newblock Thresholded covering algorithms for robust and max-min optimization.
\newblock {\em Math. Program.}, 146(1-2):583--615, 2014.

\bibitem{GuptaNR16}
A.~Gupta, V.~Nagarajan, and R.~Ravi.
\newblock Robust and maxmin optimization under matroid and knapsack uncertainty
  sets.
\newblock {\em {ACM} Trans. Algorithms}, 12(1):10:1--10:21, 2016.

\bibitem{GuptaPRS04}
A.~Gupta, M.~P\'{a}l, R.~Ravi, and A.~Sinha.
\newblock Boosted sampling: Approximation algorithms for stochastic
  optimization.
\newblock In {\em Proceedings of the Thirty-sixth Annual ACM Symposium on
  Theory of Computing}, STOC '04, pages 417--426, New York, NY, USA, 2004. ACM.

\bibitem{GuptaT08}
A.~Gupta and K.~Tangwongsan.
\newblock Simpler analyses of local search algorithms for facility location.
\newblock {\em ArXiv}, abs/0809.2554, 2008.

\bibitem{Hajiaghayi:2006:ILU:1109557.1109628}
M.~T. Hajiaghayi, R.~Kleinberg, and T.~Leighton.
\newblock Improved lower and upper bounds for universal {TSP} in planar
  metrics.
\newblock In {\em Proceedings of the Seventeenth Annual ACM-SIAM Symposium on
  Discrete Algorithms}, pages 649--658, 2006.

\bibitem{HochbaumS85}
D.~S. Hochbaum and D.~B. Shmoys.
\newblock A best possible heuristic for the k-center problem.
\newblock {\em Math. Oper. Res.}, 10(2):180--184, May 1985.

\bibitem{HochbaumS86}
D.~S. Hochbaum and D.~B. Shmoys.
\newblock A unified approach to approximation algorithms for bottleneck
  problems.
\newblock {\em J. ACM}, 33(3):533--550, May 1986.

\bibitem{JainV01}
K.~Jain and V.~V. Vazirani.
\newblock Approximation algorithms for metric facility location and k-median
  problems using the primal-dual schema and lagrangian relaxation.
\newblock {\em J. ACM}, 48(2):274--296, Mar. 2001.

\bibitem{JiaLNRS05}
L.~Jia, G.~Lin, G.~Noubir, R.~Rajaraman, and R.~Sundaram.
\newblock Universal algorithms for {TSP}, {Steiner} tree, and set cover.
\newblock In {\em Proceedings of the 36th Annual {ACM} Symposium on Theory of
  Computing}, 2005.

\bibitem{10.1007/11776178_18}
L.~Jia, G.~Noubir, R.~Rajaraman, and R.~Sundaram.
\newblock {GIST}: Group-independent spanning tree for data aggregation in dense
  sensor networks.
\newblock In P.~B. Gibbons, T.~Abdelzaher, J.~Aspnes, and R.~Rao, editors, {\em
  Distributed Computing in Sensor Systems}, pages 282--304, Berlin, Heidelberg,
  2006. Springer Berlin Heidelberg.

\bibitem{KanungoMNPSW02}
T.~Kanungo, D.~M. Mount, N.~S. Netanyahu, C.~D. Piatko, R.~Silverman, and A.~Y.
  Wu.
\newblock A local search approximation algorithm for k-means clustering.
\newblock In {\em Proceedings of the Eighteenth Annual Symposium on
  Computational Geometry}, SCG '02, pages 10--18, New York, NY, USA, 2002. ACM.

\bibitem{KasperskiZ07}
A.~Kasperski and P.~Zielinski.
\newblock On the existence of an {FPTAS} for minmax regret combinatorial
  optimization problems with interval data.
\newblock {\em Oper. Res. Lett.}, 35:525--532, 2007.

\bibitem{KhandekarKMS13}
R.~Khandekar, G.~Kortsarz, V.~S. Mirrokni, and M.~R. Salavatipour.
\newblock Two-stage robust network design with exponential scenarios.
\newblock {\em Algorithmica}, 65(2):391--408, 2013.

\bibitem{KhullerS00}
S.~Khuller and Y.~J. Sussmann.
\newblock The capacitated k-center problem.
\newblock {\em SIAM Journal on Discrete Mathematics}, 13(3):403--418, 2000.

\bibitem{KolliopoulosR99}
S.~G. Kolliopoulos and S.~Rao.
\newblock A nearly linear-time approximation scheme for the {Euclidean}
  k-median problem.
\newblock In J.~Ne{\v{s}}et{\v{r}}il, editor, {\em Algorithms - ESA' 99}, pages
  378--389, Berlin, Heidelberg, 1999. Springer Berlin Heidelberg.

\bibitem{KouvelisY97}
P.~Kouvelis and G.~Yu.
\newblock {\em Robust 1-Median Location Problems: Dynamic Aspects and
  Uncertainty}, pages 193--240.
\newblock Springer US, Boston, MA, 1997.

\bibitem{KumarSS04}
A.~Kumar, Y.~Sabharwal, and S.~Sen.
\newblock A simple linear time $(1 + \epsilon)$-approximation algorithm for
  $k$-means clustering in any dimensions.
\newblock In {\em Proceedings of the 45th IEEE Symposium on Foundations of
  Computer Science}, pages 454--462, 2004.

\bibitem{LiS13}
S.~Li and O.~Svensson.
\newblock Approximating $k$-median via pseudo-approximation.
\newblock In {\em Proceedings of the Forty-fifth Annual ACM Symposium on Theory
  of Computing}, pages 901--910, 2013.

\bibitem{Lloyd82}
S.~P. Lloyd.
\newblock Least squares quantization in pcm.
\newblock {\em IEEE Trans. Information Theory}, 28:129--136, 1982.

\bibitem{Mentzer16}
S.~G. Mentzer.
\newblock Approximability of metric clustering problems.
\newblock Unpublished manuscript, Mar. 2016.

\bibitem{NagarajanSS13}
V.~Nagarajan, B.~Schieber, and H.~Shachnai.
\newblock The {Euclidean} k-supplier problem.
\newblock In M.~Goemans and J.~Correa, editors, {\em Integer Programming and
  Combinatorial Optimization}, pages 290--301, Berlin, Heidelberg, 2013.
  Springer Berlin Heidelberg.

\bibitem{FisherNW78}
G.~L. Nemhauser, L.~A. Wolsey, and M.~L. Fisher.
\newblock An analysis of approximations for maximizing submodular set
  functions---i.
\newblock {\em Mathematical Programming}, 14(1):265--294, 1978.

\bibitem{Platzman:1989:SCP:76359.76361}
L.~K. Platzman and J.~J. Bartholdi, III.
\newblock Spacefilling curves and the planar travelling salesman problem.
\newblock {\em J. ACM}, 36(4):719--737, Oct. 1989.

\bibitem{SCHALEKAMP2008}
F.~Schalekamp and D.~B. Shmoys.
\newblock Algorithms for the universal and a priori {TSP}.
\newblock {\em Operations Research Letters}, 36(1):1--3, 2008.

\bibitem{ShmoysS06}
D.~B. Shmoys and C.~Swamy.
\newblock An approximation scheme for stochastic linear programming and its
  application to stochastic integer programs.
\newblock {\em J. ACM}, 53(6):978--1012, Nov. 2006.

\bibitem{SwamyS06}
C.~Swamy and D.~B. Shmoys.
\newblock Approximation algorithms for 2-stage stochastic optimization
  problems.
\newblock {\em {SIGACT} News}, 37(1):33--46, 2006.

\bibitem{SwamyS12}
C.~Swamy and D.~B. Shmoys.
\newblock Sampling-based approximation algorithms for multistage stochastic
  optimization.
\newblock {\em {SIAM} J. Comput.}, 41(4):975--1004, 2012.

\end{thebibliography}

\appendix 

\section{Relations Between Vector Norms}\label{section:vector}

For completeness, we give a proof of the following well known fact that relates vector $\ell_p$ norms. 

\begin{fact}
\label{fact:plogn}
For any $1 \leq p \leq q$ and $x \in \mathbb{R}^n$, we have $$\norm{x}_q \leq \norm{x}_p \leq n^{1/p - 1/q} \norm{x}_q.$$
\end{fact}
\begin{proof}
For all $p$, repeatedly applying Minkowski's inequality we have:

$$\norm{x}_p = \left(\sum_{i = 1}^n |x|^p\right)^{1/p} \leq \left(\sum_{i = 1}^{n-1} |x|^p\right)^{1/p} + |x_n| \leq \left(\sum_{i = 1}^{n-2} |x|^p\right)^{1/p} + |x_{n-1}| + |x_n| \leq \ldots \leq \sum_{i=1}^n |x_i| = \norm{x}_1.$$

Then we bound $\norm{x}_q$ by $\norm{x}_p$ as follows:

$$\norm{x}_q = \left(\sum_{i = 1}^n |x|^q\right)^{1/q} = \left(\left(\sum_{i = 1}^n \left(|x|^{p}\right)^{q/p}\right)^{p/q}\right)^{1/p}\leq \left(\sum_{i = 1}^n |x|^{p}\right)^{1/p} = \norm{x}_p.$$

The inequality is by applying $\norm{x'}_{q/p} \leq \norm{x'}_1$ to the vector $x'$ with entries $x_i' = |x_i|^p$. 
To bound $\norm{x}_p$ by $\norm{x}_q$, we invoke Holder's inequality as follows:

$$\norm{x}_p^p = \sum_{i = 1}^n |x|^p = \sum_{i = 1}^n |x|^p \cdot 1 \leq \left(\sum_{i=1}^n\left(|x_i|^p\right)^{q/p}\right)^{p/q} \left(\sum_{i=1}^{n}1^{q/(q-p)}\right)^{1-p/q} = \norm{x}_q^p \cdot n^{1-p/q}.$$

Taking the $p$th root of this inequality gives the desired bound.
\end{proof}

The limiting behavior as $q \rightarrow \infty$ shows that $\norm{x}_{\infty} \leq \norm{x}_{c \log n} \leq n^{1/c \log n} \norm{x}_\infty = 2^{1/c} \norm{x}_\infty$, i.e. that the $\ell_\infty$-norm and $\ell_p$-norm for $p = \Omega(\log n)$ are within a constant factor.
  \section{Approximations for All-Clients Instances are not Universal}
  \label{sec:approx-fail}
  
  In this section, we demonstrate that even $(1+\epsilon)$-approximate (integer) solutions for the ``all clients'' instance for clustering problems are not guaranteed to be $(\alpha, \beta)$-approximations for any finite $\alpha, \beta$. This is in sharp contrast to the optimal (integer) solution, which is known to be a $(1,2)$-approximation for a broad range of problems including the clustering problems considered in this paper~\cite{KasperskiZ07}.
 
 Consider an instance of universal $1$-median with clients $c_1, c_2$ and cluster centers $f_1, f_2$. Both the cluster centers are at distance $1$ from $c_1$, and at distances $0$ and $\epsilon$ respectively from $c_2$ (see Figure~\ref{fig:badapprox}). $f_2$ is a $(1+\epsilon)$-approximate solution for the realization containing both clients. In this instance, $\mr$ is 0 and so $f_2$ is not an $(\alpha, \beta)$-approximation for any finite $\alpha, \beta$ due to the realization containing only $c_2$. The same example can be used for the $\ell_p$-clustering objective for all $p \geq 1$ since $f_2$ has approximation factor $(1+\epsilon^p)^{1/p} \leq (1 + \epsilon)$ when all clients are present. In the case of $k$-center, $f_2$ is an optimal solution when all clients are present.

\begin{figure}[H]
    \centering
    \includegraphics[width=.5\textwidth]{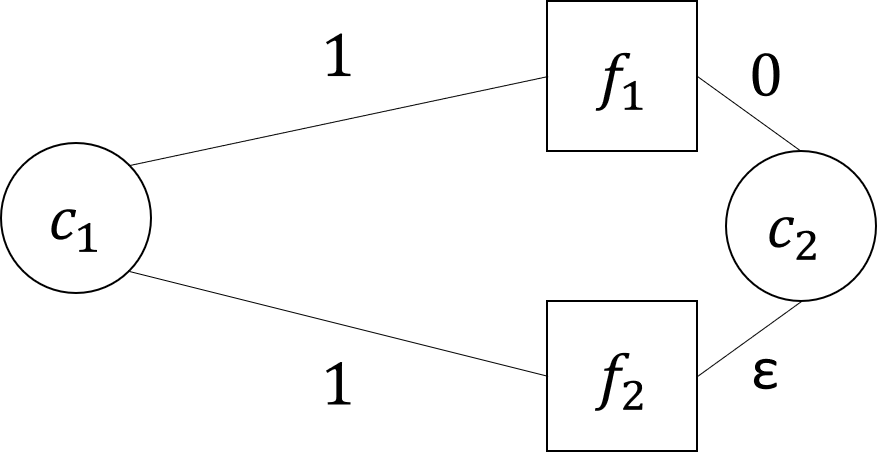}
    \caption{Example where a $(1+\epsilon$)-approximation for all clients has no $(\alpha, \beta)$-approximation guarantee, for any $\ell_p$-clustering objective including $k$-center.}
    \label{fig:badapprox}
\end{figure}

\section{Algorithms for $k$-Median and $\ell_p^p$-Clustering with Discounts}
\label{section:kmwithdiscounts}

In this section, we prove Lemma~\ref{lemma:kmbicriteria-lpp} which states that there exists a $(9^p, \frac{2}{3} \cdot 9^p)$-approximation algorithm for the $\ell_p^p$-clustering with discounts problem. As a corollary, by setting $p=1$, we will obtain Lemma~\ref{lemma:kmbicriteria} which states that there exists a $(9, 6)$-approximation algorithm for the $k$-median with discounts problem. To prove Lemma \ref{lemma:kmbicriteria-lpp}, we will first use a technique due to Jain and Vazirani~\cite{JainV01} to design a Lagrangian-preserving approximation for the $\ell_p^p$-facility location with discounts problem. $\ell_p^p$-facility location with discounts (FLD) is the same as $\ell_p^p$-clustering with discounts, except rather than being restricted to buying $k$ cluster centers, each cluster center has a cost $f_i$ associated with buying it (discounts and cluster centers costs are not connected in any way).

\subsection{Algorithm for $\ell_p^p$-Facility Location with Discounts} \label{section:mfld}

Since FLD is a special case of non-metric facility location, we can consider the standard linear programming primal-dual formulation for the latter. The primal program is as follows:
%
%
\begin{equation*}
\begin{array}{lll}
& \min & \sum_{i \in F} f_ix_i + \sum_{i \in F, j \in C} (c_{ij}^p-r_j^p)^+y_{ij}\\
\textnormal{s.t.} &\forall j \in C: &\sum_{i \in F} y_{ij} \geq 1\\
&\forall i \in F,j \in C:& y_{ij} \leq x_i\\
&\forall i \in F: &x_i \geq 0 \\
&\forall i \in F, j \in C: & y_{ij} \geq 0
\end{array}
\end{equation*}

The dual program is as follows:

\begin{equation*}
\begin{array}{lll}
& \max &\sum_{j \in C} a_j\\
\textnormal{s.t.} &\forall i \in F, j \in C:& a_j - (c_{ij}^p-r_j^p)^+ \leq b_{ij}\\
&\forall i \in F: &\sum_{j \in C} b_{ij} \leq f_i\\
& \forall j \in C: & a_j \geq 0\\
& \forall i \in F,j \in C: & b_{ij} \geq 0\\
\end{array}
\end{equation*}

We design a primal-dual algorithm for the FLD problem. This FLD algorithm operates in two phases. In both programs, all variables start out as $0$. 

In the first phase, we generate a dual solution. For each client $j$ define a ``time'' variable $t_j$ which is initialized to $0$. We grow the dual variables as follows: we increase the $t_j$ uniformly. We grow the $a_j$ such that for any $j$, at all times $a_j = (t_j - r_j^p)^+$ (or equivalently, all $a_j$ start at 0, and we increase all $a_j$ at a uniform rate, but we only start growing $a_j$ at time $r_j^p$). Each $b_{ij}$ is set to the minimum feasible value, i.e. $(a_j - (c_{ij}^p-r_j^p)^+)^+$. If the constraint $\sum_{j \in C} b_{ij} \leq f_i$ is tight, we stop increasing all $t_j, a_j$ for which $b_{ij} = a_j - (c_{ij}^p-r_j^p)^+$, i.e., for the clients $j$ that contributed to increasing the value $\sum_{j \in C} b_{ij}$ (we say these clients put weight on this cluster center). We continue this process until all $t_j$ stop growing. Note that at any time the dual solution grown is always feasible.

In the second phase, consider a graph induced on the cluster centers whose constraints are tight, where we place an edge between cluster centers $i, i'$ if there exists some client $j$ that put weight on both cluster centers. Find a maximal independent set $S$ of this graph and output this set of cluster centers. Let $\pi$ be a map from clients to cluster centers such that $\pi(j)$ is the cluster center which made $t_j$ stop increasing in the first phase of the algorithm. If $\pi(j) \in S$, connect $j$ to cluster center $\pi(j)$, otherwise connect $j$ to one of $\pi(j)$'s neighbors in the graph arbitrarily. 

We can equivalently think of the algorithm as generating an integral primal solution where $x_i = 1$ for all $i \in S$ and $x_i = 0$ otherwise, and $y_{ij} = 1$ if $j$ is connected to $i$ and is $y_{ij} = 0$ otherwise. Based again on the technique of \cite{JainV01}, we can show the following Lemma holds:
%
%
%
\begin{lemma}\label{lemma:uflbicriteria}
Let $x, y$ be the primal solution and $a, b$ be the dual solution generated by the above FLD algorithm. $x, y$ satisfies 
$$\sum_{j \in C} \sum_{i \in F} (c_{ij}^p - 3^p r_j^p)^+y_{ij} + 3^p \sum_{i \in F} f_ix_i \leq 3^p \sum_{j \in C} a_j $$
\end{lemma}

%
%
%

\begin{proof}
Let $C^{(1)}$ be the set of clients in $j$ such that $\pi(j) \in S$ and $C^{(2)} = C \setminus C^{(1)}$. For any $i \in S$, let $C_i$ be the set of all clients $j$ such that $\pi(j) = i$. Note that:

$$\sum_{j \in C_i} a_j = \sum_{j \in C_i} [b_{ij} + (c_{ij} - r_j^p)^+] =  \sum_{j \in C_i} (c_{ij} - r_j^p)^+ + f_i$$
%
%
%
No client in $C^{(1)}$ contributes to the sum $\sum_{i \in F} b_{ij}$ for multiple $i$ in $S$ (because $S$ is an independent set). This gives us:

%
%
%
%
%
%
%

\begin{eqnarray} \label{eq:lagrangianproof1}
\sum_{j \in C^{(1)}} \sum_{i \in F} (c_{ij}^p - 3^pr_j^p)^+y_{ij} + \sum_{i \in F} f_ix_i 
& \leq & \sum_{j \in C^{(1)}} \sum_{i \in F} (c_{ij}^p - r_j^p)^+y_{ij} + \sum_{i \in F} f_ix_i\nonumber\\
& = & \sum_{i \in S}[f_i + \sum_{j \in C_i} (c_{ij}^p - r_j^p)^+]\nonumber\\ 
& = & \sum_{i \in S} \sum_{j \in C_i} a_j\nonumber\\ 
& = & \sum_{j \in C^{(1)}} a_j.
\end{eqnarray}

%
%
%
%
%
%
For each client $j$ in $C^{(2)}$, $j$ is connected to one of $\pi(j)$'s neighbors $i$. Since $\pi(j)$ and $i$ are neighbors, there is some client $j'$ that put weight on both $\pi(j)$ and $i$. Since $j'$ put weight on $\pi(j)$ and thus $\pi(j)$ going tight would have stopped $t_{j'}$ from increasing, $t_{j'}$ stopped increasing before or when $\pi(j)$ went tight, which was when $t_j$ stopped growing. Since all $t_j$ start growing at the same time and grow uniformly, $t_{j'} \leq t_j$. Since $j$ put weight on $\pi(j)$, we know $a_j - (c_{\pi(j)j} - r_j^p)^+ > 0$ and thus $(t_j - r_j^p)^+ - (c_{\pi(j)j} - r_j^p)^+ > 0$, implying $t_j \geq c_{\pi(j)j}^p$. Similarly, $t_{j'} \geq c_{ij'}^p, c_{\pi(j)j'}^p$. Triangle inequality gives $c_{ij} \leq c_{ij'} + c_{\pi(j)j'} + c_{\pi(j)j} \leq 3t_j^{1/p}$. Then, we get:

\begin{equation} \label{eq:lagrangianproof2}
\sum_{j \in C^{(2)}} \sum_{i \in F} (c_{ij}^p - 3^pr_j^p)^+y_{ij} \leq 
\sum_{j \in C^{(2)}} (3^pt_j - 3^pr_j^p)^+ = 3^p\sum_{j \in C^{(2)}} (t_j - r_j^p)^+ = 3^p\sum_{j \in C^{(2)}} a_j
\end{equation}

Adding $3^p$ times Eq.~\eqref{eq:lagrangianproof1} to Eq.~\eqref{eq:lagrangianproof2} gives the Lemma.
\end{proof}

\subsection{Algorithm for $\ell_p^p$-Clustering with Discounts}

We now move on to finding an algorithm for $\ell_p^p$-clustering with discounts. We can represent the problem as a primal/dual linear program pair as follows. The primal program is:

\begin{equation*}
\begin{array}{lll}
& \min & \sum_{i \in F, j \in C} (c_{ij}^p-r_j^p)^+y_{ij}\\
\textnormal{s.t.} &\forall j \in C: &\sum_{i \in F} y_{ij} \geq 1\\
&\forall i \in F,j \in C: &y_{ij} \leq x_i\\
& &\sum_{i \in F} x_i \leq k\\
& \forall i \in F: &x_i \geq 0 \\
& \forall i \in F,j \in C: &y_{ij} \geq 0
\end{array}
\end{equation*}

The dual program is as follows:
\begin{equation*}
\begin{array}{lll}
& \max &\sum_{j \in C} a_j - kz\\
\textnormal{s.t.} &\forall i \in F,j \in C:& a_j - (c_{ij}^p-r_j^p)^+ \leq b_{ij}\\
&\forall i \in F: &\sum_{j \in C} b_{ij} \leq z\\
& \forall j \in C: & a_j \geq 0\\
& \forall i \in F,j \in C: & b_{ij} \geq 0\\
\end{array}
\end{equation*}
%
%
%

%
%
%

%
%
%
%
%
%
We now describe the algorithm we will use to prove Lemma \ref{lemma:kmbicriteria-lpp}, which uses the FLD algorithm from Section \ref{section:mfld} as a subroutine. By taking our $\ell_p$-clustering with discounts instance and assigning all cluster centers the same cost $z$, we can produce a FLD instance. When $z = 0$, the FLD algorithm will either buy more than $k$ cluster centers, or find a set of at most $k$ cluster centers, in which case we can output that set. When $z = |C|\max_{i,j} c_{ij}$, the FLD algorithm will buy only 1 cluster center. Thus, for any $\epsilon$ such that $\log{\frac{1}{\epsilon}} = n^{O(1)}$, via bisection search using polynomially many runs of this algorithm we can find a value of $z$ such that this algorithm buys a set of cluster centers $S_1$ of size $k_1 \geq k$ when cluster centers cost $z$ and a set of cluster centers $S_2$ of size $k_2 \leq k$ cluster centers when cluster centers cost $z + \epsilon$ (the bisection search starts with the range $[0, |C|\max_{i,j} c_{ij}]$ and in each iteration, determines how many cluster centers are bought when $z$ is the midpoint value in its current range. It then recurses on the half $[a, b]$ of its current range which maintains the invariant that when $z = a$, at least $k$ cluster centers are bought and when $z = b$, at most $k$ cluster centers are bought).

If either $k_1 = k$ or $k_2 = k$, we output the corresponding cluster center set. Otherwise, we will randomly choose a solution which is roughly a combination of $S_1$ and $S_2$ (we will describe how to derandomize this process as is required to prove Lemma \ref{lemma:kmbicriteria} later). Let $\rho$ be the solution in $[0, 1]$ to $\rho k_1 + (1-\rho)k_2 = k$, i.e. $\rho = \frac{k-k_2}{k_1 - k_2}$. Construct a set $S_1'$ that consists of the closest cluster center in $S_1$ to each cluster center in $S_2$. If  the size of $S_1'$ is less than $k_2$, add arbitrary cluster centers from $S_1 \setminus S_1'$ to $S_1'$ until its size is  $k_2$. Then, with probability $\rho$, let $S^* = S_1'$, otherwise let $S^* = S_2$. Then, sample a uniformly random subset of $k - k_2$ elements from $S_1 \setminus S_1'$ and add them to $S^*$. Then output $S^*$ (note that $S_1 \setminus S_1'$ is of size $k_1 - k_2$ so every element in $S_1 \setminus S_1'$ has probability $\rho$ of being chosen). 

\begin{proof}[Proof of Lemma \ref{lemma:kmbicriteria-lpp}]
Note that if the FLD algorithm ever outputs a solution which buys exactly $k$ cluster centers, then by Lemma \ref{lemma:uflbicriteria} we get that for the LP solution $x, y$ encoding this solution and a dual solution $a$:

\begin{equation*}
\begin{aligned}
\sum_{j \in C} \sum_{i \in F} (c_{ij}^p - 3^pr_j^p)^+y_{ij} + 3^p \sum_{i \in F} f_ix_i &\leq 3^p \sum_{j \in C} a_j \\
\sum_{j \in C} \sum_{i \in F} (c_{ij}^p - 3^pr_j^p)^+y_{ij} + 3^pkz &\leq 3^p \sum_{j \in C} a_j \\
\sum_{j \in C} \sum_{i \in F} (c_{ij}^p - 3^pr_j^p)^+y_{ij} + &\leq 3^p [\sum_{j \in C} a_j - kz]
\end{aligned}
\end{equation*}

Which by duality means that this solution is also a $(3^p, 3^p)$-approximation for the $\ell_p$-clustering with discounts instance.

If bisection search never finds a solution with exactly $k$ cluster centers, but instead a pair of solutions $S_1, S_2$ where $|S_1| > k, |S_2| < k$, the idea is that the algorithm constructs a "bi-point" fractional solution from these solutions (i.e. constructs a fractional solution that is just a convex combination of the two integral solutions) and then rounds it. 

Consider the primal/dual solutions $x^{(1)}, y^{(1)}, a^{(1)}$ and $x^{(2)}, y^{(2)}, a^{(2)}$ corresponding to $S_1, S_2$. By Lemma \ref{lemma:uflbicriteria} we get:

\begin{equation*}
\begin{aligned}
\sum_{j \in C} \sum_{i \in F} (c_{ij}^p - 3^pr_j^p)^+y_{ij}^{(1)} + 3^p k_1z &\leq 3^p \sum_{j \in C} a_j^{(1)} \\
\sum_{j \in C} \sum_{i \in F} (c_{ij}^p - 3^pr_j^p)^+y_{ij}^{(2)} + 3^p k_2(z+\epsilon) &\leq 3^p \sum_{j \in C} a_j^{(2)} \\
\end{aligned}
\end{equation*}

By combining the two inequalities and choosing $\epsilon$ appropriately we can get that:

\begin{equation*}
\sum_{j \in C} \sum_{i \in F} (c_{ij}^p - 3^pr_j^p)^+(\rho y_{ij}^{(1)}+(1-\rho)y_{ij}^{(2)}) \leq (3^p+\epsilon') [\sum_{j \in C} (\rho a_j^{(1)}+(1-\rho)a_j^{(2)}) - kz]
\end{equation*}

For an $\epsilon'$ we will fix later.

%
%
%
Note that $\rho(x^{(1)}, y^{(1)}, a^{(1)})+(1-\rho)(x^{(2)}, y^{(2)}, a^{(2)})$ and $z$ form a feasible (fractional) primal/dual solution pair for the $\ell_p$-clustering with discounts problem, and by the above inequality the primal solution is a $(3^p, 3^p+\epsilon')$-approximation.

Then, we round the convex combination of the two solutions as described above. Let $c_j$ be the connection cost of client $j$ in the rounded solution, and $c_j^{(1)}, c_j^{(2)}$ the connection cost of client $j$ in solutions $S_1, S_2$. Then since $(3^p+\epsilon')(2 \cdot 3^{p-1}-\epsilon') < \frac{2}{3} \cdot 9^p$ for $\epsilon \in [0, 1]$ to prove the lemma it suffices to show that for each client $j$ the expected contribution to the objective using discount $9r_j$ for client $j$ is at most $2 \cdot 3^{p-1}-\epsilon'$ times the contribution of client $j$ to the primal solution's objective using discount $3r_j$. That is:

$$\mathbb{E}[(c_j^p - 9^pr_j^p)^+] \leq (2 \cdot 3^{p-1}-\epsilon')[\rho(c_j^{(1)p} - 3^pr_j^p)^+ + (1-\rho)(c_j^{(2)p} - 3^pr_j^p)^+]$$

%
%
%

Suppose client $j$'s nearest cluster center in $S_1$ is in $S_1'$. Then with probability $\rho$, $j$ is connected to that cluster center at connection cost $c_j^{(1)}$, and with probability $1-\rho$ it is connected to the nearest cluster center in $S_2$ at connection cost $c_j^{(2)}$. Then:

$$\mathbb{E}[(c_j^p - 9^pr_j^p)^+] \leq \mathbb{E}[(c_j^p - 3^pr_j^p)^+] = \rho(c_j^{(1)p} - 3^pr_j^p)^+ + (1-\rho)(c_j^{(2)p} - 3r_j^p)^+$$

%
%
%

Suppose client $j$'s nearest cluster center in $S_1$ (call it $i_1$) is not in $S_1'$. Note that each cluster center in $S_1 \setminus S_1'$ has probability $\rho$ of being opened. Thus with probability $\rho$, we can upper bound $c_j$ by the distance from $j$ to $i_1$. If this does not happen, let $i_2$ be $j$'s nearest cluster center in $S_2$ and $i_1'$ be the cluster center nearest to $i_2$ in $S_1$. One of $i_1', i_2$ must be opened, so we can bound $j$'s connection cost by its connection cost to whichever is opened. Then one of three cases occurs:
%
%
%
\begin{itemize}
\item With probability $\rho$, $j$'s nearest cluster center in $S_1$ is opened. Then $c_j$ is at most the distance from $j$ to $i_1$, i.e. $c_j^{(1)}$.

%
%
%
%
%
%
\item With probability $(1-\rho)\rho$, $j$'s nearest cluster center in $S_1$ is not opened and $S_1'$ is opened. $c_j$ is at most the distance from $j$ to $i_1'$. Since $i_1'$ is the cluster center closest to $i_2$ in $S_1$, the distance from $i_1'$ to $i_2$ is at most the distance from $i_1$ to $i_2$, which is at most $c_j^{(1)} + c_j^{(2)}$. Then by triangle inequality, the distance from $j$ to $i_1'$ is at most $c_j^{(1)} + 2c_j^{(2)}$. Using the AMGM inequality, we get $c_{i_1' j}^p \leq 3^{p-1}(c_j^{(1) p} + 2 c_j^{(2)p})$.
\item With probability $(1-\rho)^2$, $j$'s nearest cluster center in $S_1$ is not opened and $S_2$ is opened. $c_j$ is at most the distance from $j$ to $i_2$, i.e. $c_j^{(2)}$.
\end{itemize}

Then we get:
\begin{equation*}
\begin{aligned}
& \mathbb{E}[(c_j - 9^p r_j^p)^+] \\
&\leq \rho(c_j^{(1)p} - 9^p r_j^p)^+ + (1-\rho)^2(c_j^{(2)p} - 9^p r_j^p)^+ + (1-\rho)\rho(3^{p-1}(c_j^{(1)} + 2c_j^{(2)}) - 9^p r_j^p)^+ \\
&= \rho(c_j^{(1)p} - 9^p r_j^p)^+ + (1-\rho)^2(c_j^{(2)p} - 9^p r_j^p)^+ + 3^{p-1}(1-\rho)\rho(c_j^{(1)} + 2c_j^{(2)} - 3 \cdot 3^p r_j^p)^+ \\
&\leq \rho(c_j^{(1)} - 3^p r_j^p)^+ + (1-\rho)^2(c_j^{(2)} - 3^p r_j^p)^+ + 3^{p-1}(1-\rho)\rho(c_j^{(1)} - 3^p r_j^p)^+ + 2\cdot 3^{p-1}(1-\rho)\rho(c_j^{(2)} - 3^p r_j^p)^+\\
&= (3^{p-1} (1 - \rho) +1)[\rho(c_j^{(1)} - 3^p r_j^p)^+] + (2 \cdot 3^{p-1} \rho +1 - \rho)[(1-\rho)(c_j^{(2)} - 3^p r_j^p)^+]\\
&\leq (2 \cdot 3^{p-1} - \min\{\rho, 1-\rho\})[\rho(c_j^{(1)} - 3^p r_j^p)^+ + (1-\rho)(c_j^{(2)} - 3^p r_j^p)^+]\\
&\leq (2 \cdot 3^{p-1} - \epsilon')[\rho(c_j^{(1)} - 3^p r_j^p)^+ + (1-\rho)(c_j^{(2)} - 3^p r_j^p)^+].\\
\end{aligned}
\end{equation*}

Where the last step is given by choosing $\epsilon'$ to be at most $\frac{1}{|F|}$, since $p$ = $\frac{k-k_2}{k_1 - k_2}$ and $1 \leq k_2 < k < k_1 \leq |F|$ and thus $\rho$ and $1-\rho$ are both at least $\frac{1}{|F|}$. This gives the lemma, except that the algorithm is randomized. However, the randomized rounding scheme can easily be derandomized: first, we choose $S^*$ to be whichever of $S_1', S_2$ has a lower expected objective. Then, to choose the remaining $k - k_2$ cluster centers to add to $S^*$, we add cluster centers by one by one.  When we have $c$ cluster centers left to add to $S^*$, we add the cluster center $i$ from $S_1 \setminus S_1'$  that minimizes the expected objective achieved by $S^* \cup \{i\}$ and $c-1$  random cluster centers from $(S_1 \setminus S_1') \setminus (S^* \cup \{i\})$. Each step of the derandomization cannot increase the expected objective, so the derandomized algorithm achieves the guarantee of Lemma \ref{lemma:kmbicriteria-lpp}. 
\end{proof}

Again, we note that Lemma~\ref{lemma:kmbicriteria} is obtained as a corollary of Lemma~\ref{lemma:kmbicriteria-lpp}, where we set $p=1$.

\end{document}